%% file: main.tex
\documentclass[aps,superscriptaddress,nofootinbib,showpacs,notitlepage, twocolumn]{revtex4-1}

\usepackage{amsfonts}
\usepackage{amsmath}
\usepackage{amsthm}
\usepackage{braket}
\usepackage{graphics}
\usepackage{ulem}

\usepackage[dvipsnames]{xcolor}
\usepackage{amssymb}
\usepackage{dsfont}
\usepackage{verbatim}
\usepackage{amsmath,amscd}
\usepackage[all,cmtip]{xy}
\usepackage{multirow}

\usepackage{etoolbox,tikz, tikz-cd}
\usetikzlibrary{decorations.pathmorphing}

\usetikzlibrary{cd}
\usepackage{relsize}

\tikzset{fontscale/.style = {font=\relsize{#1}}
    }

\definecolor{darkblue}{RGB}{0,0,127} % choose colors
\definecolor{darkgreen}{RGB}{0,130,80}
\definecolor{darkred}{RGB}{160,5,5}

\graphicspath{{./Figures/}}

\usepackage{tikz,tikz-3dplot,pgfplots}\pgfplotsset{compat=newest}
\usetikzlibrary{external,calc,decorations.pathreplacing,decorations.markings,decorations.pathmorphing,arrows.meta,shapes.geometric}
\newlength\figureheight
\newlength\figurewidth

\usepackage{graphicx}
\usepackage{tabularx}
\usepackage{booktabs}
\usepackage{array}
\usepackage{amsmath}
\usepackage{amsfonts}
\usepackage{amssymb}
\usepackage{amsthm}
\usepackage{enumitem}
\usepackage{permute}
\usepackage{url}
\usepackage{comment}
\usetikzlibrary{shapes}

\usepackage{mathtools}

\DeclareMathOperator{\Gal}{Gal}

\DeclareMathOperator{\Char}{char}
\DeclareMathOperator{\Ann}{Ann}

\DeclareFontFamily{U}{mathb}{\hyphenchar\font45}
\DeclareFontShape{U}{mathb}{m}{n}{
      <5> <6> <7> <8> <9> <10> gen * mathb
      <10.95> mathb10 <12> <14.4> <17.28> <20.74> <24.88> mathb12
      }{}
\DeclareSymbolFont{mathb}{U}{mathb}{m}{n}
\DeclareMathSymbol{\bigast}{2}{mathb}{"06}

\def\XXint#1#2#3{{\setbox0=\hbox{$#1{#2#3}{\int}$}
     \vcenter{\hbox{$#2#3$}}\kern-.5\wd0}}

 \newtheorem{lemma}{Lemma}

 \newtheorem{theorem}{Theorem}

\theoremstyle{definition}

\newtheoremstyle{remark}
{}   % ABOVESPACE, \topsep for no space
{}   % BELOWSPACE, \topsep for no space
{\normalfont}  % BODYFONT
{}       % INDENT (empty value is the same as 0pt)
{\itshape} % HEADFONT
{.}         % HEADPUNCT
{5pt plus 1pt minus 1pt} % HEADSPACE
{}          % CUSTOM-HEAD-SPEC

\theoremstyle{remark}
\newtheorem*{remark}{Remark}

\newlist{alternative}{enumerate}{4}     % this creates a dedicated counter named 'subtaski'
\setlist[alternative,1]{label=\arabic*., ref=\arabic*}
\setlist[alternative,2]{label=(\alph*), ref=\thealternativei.(\alph*)}
\setlist[alternative,3]{label=\roman*., ref=\thealternativei.(\thealternativeii).\roman*}
\setlist[alternative,4]{label=\Alph*., ref=\thealternativei.(\thealternativeii).\thealternativeiii.\Alph*}

\setlist[enumerate,1]{label=\arabic*., ref=\arabic*}
\setlist[enumerate,2]{label=(\alph*), ref=\theenumi.(\alph*)}
\setlist[enumerate,3]{label=\roman*., ref=\theenumi.(\theenumii).\roman*}
\setlist[enumerate,4]{label=\Alph*., ref=\theenumi.(\theenumii).\theenumiii.\Alph*}

\newcommand{\oline}[1]{\overline{#1}}

\newcommand{\R}[1]{{Ref.~\onlinecite{#1}}}

\AtBeginEnvironment{tikzcd}{\tikzexternaldisable}
\AtEndEnvironment{tikzcd}{\tikzexternalenable}

\newcommand{\drawgenerator}[8]{%
\xymatrix@!0{%
& #8 \ar@{-}[ld]\ar@{.}[dd] \ar@{-}[rr] & & #7 \ar@{-}[ld]  \\%
#1 \ar@{-}[rr] \ar@{-}[dd] &  & #2 \ar@{-}[dd] &            \\%
& #6 \ar@{.}[ld] &  & #5 \ar@{-}[uu] \ar@{.}[ll]       \\%
#3 \ar@{-}[rr] &  & #4 \ar@{-}[ru]                       %
}}

\newcommand{\plaquette}[4]{
\xymatrix@!0{%
#1 \ar@{-}[r] \ar@{-}[d]  & #2 \ar@{-}[d] 
\\
#3 \ar@{-}[r]  & #4
}}

\usepackage{floatrow}
\usepackage[caption=false,label font={bf,normalsize}]{subfig}
\global\long\def\av#1{\left\langle #1 \right\rangle }

\global\long\def\im{\text{Im}}

\usepackage{hyperref}
\hypersetup{colorlinks,
	linkcolor=darkblue,
	citecolor=darkgreen,
	filecolor=red,
	urlcolor=blue,
	pdftitle={},
	pdfauthor={}
}
\usepackage[capitalize]{cleveref}

\Crefname{enumi}{Case}{Cases}
\Crefname{subsection}{Subsection}{Subsections}
\pgfdeclarelayer{bg}    % declare background layer
\pgfsetlayers{bg,main}

\begin{document}

\title{Bifurcating entanglement-renormalization group flows of fracton stabilizer models}

\author{Arpit Dua}
%\thanks{arpit.dua@yale.edu}
\affiliation{Department of Physics, Yale University, New Haven, CT 06520-8120, USA}
\affiliation{Yale Quantum Institute, Yale University, New Haven, CT 06520, USA}

\author{Pratyush Sarkar}
\affiliation{Department of Mathematics, Yale University, New Haven, CT 06511, USA} 

\author{Dominic~J. Williamson}
\thanks{Current Address: Stanford Institute for Theoretical Physics, Stanford University, Stanford, CA 94305, USA}
\affiliation{Department of Physics, Yale University, New Haven, CT 06520-8120, USA}

\author{Meng Cheng}
\affiliation{Department of Physics, Yale University, New Haven, CT 06520-8120, USA}
\affiliation{Yale Quantum Institute, Yale University, New Haven, CT 06520, USA}

\begin{abstract}
We investigate the entanglement-renormalization group flows of translation-invariant topological stabilizer models in three dimensions. Fracton models are observed to bifurcate under entanglement renormalization, generically returning at least one copy of the original model. 
Based on this behavior we formulate the notion of bifurcated equivalence for fracton phases, generalizing foliated fracton equivalence. 
The notion of quotient superselection sectors is also generalized accordingly. We calculate bifurcating entanglement-renormalization group flows for a wide range of examples and, based on those results, propose conjectures regarding the classification of translation-invariant topological stabilizer models in three dimensions.  
\end{abstract}
  
\maketitle

The renormalization group (RG) is a ubiquitous concept throughout theoretical physics. 
Intuitively, real-space renormalization~\cite{Wilson1975} involves the rescaling of a system such that the short-distance correlations are integrated out. 
At the scale-invariant fixed points of an RG transformation, the system has either zero or infinite correlation length, corresponding to gapped or critical phases respectively. 
As such, RG serves as an important tool for the analysis of long-wavelength physics in a given theory. 

Since rescaling transformations increase the local Hilbert space dimension, conventional RG approaches have relied on truncating the local Hilbert space in a manner that does not affect long-range correlations. 
To further ensure that the long-range entanglement structure is preserved, which is particularly crucial for the classification and characterization of quantum matter, a more careful approach must be taken.  
A real-space renormalization procedure that exclusively removes short-range entanglement consists of applying local unitary circuits and projecting out degrees of freedom that have been completely disentangled into trivial states only~\cite{haah2014bifurcation}. 
Along with the coarse-graining of degrees of freedom to rescale the system, such a procedure is referred to as Entanglement Renormalization (ER)~\cite{Vidal2007}. 
Stable fixed points under ER are identified as representatives of quantum phases of matter. For many exactly solvable models with zero correlation length, ER can be implemented directly at the level of the Hamiltonian~\cite{Konig2009,Aguado2008}. 

In this paper we apply ER to study the long-range entanglement structures of topological stabilizer models. In two dimensions, every translation-invariant topological stabilizer model is equivalent to copies of the 2D toric code~\cite{Haah2018a,bombin2012universal,bombin2014structure} and hence their classification is complete. The 2D toric code is a fixed point under ER. Hence, under ER, any 2D translation-invariant topological stabilizer model flows towards copies of the 2D toric code. In contrast, translation-invariant topological stabilizer models in three dimensions exhibit a rich variety of quantum phases due to the existence of fracton topological order~\cite{chamon2005quantum,PhysRevB.81.184303,bravyi2011topological,Chamon_quantum_glassiness,PhysRevB.95.245126,vijay2016fracton,Williamson_cubic_code,vijay2017isotropic,vijay2017generalization,PhysRevB.96.165106,HHB_models,PhysRevB.97.155111,PhysRevB.97.041110,prem2018cage,Bulmash2018,hao_twisted,Brown2019,finite_temp_Xcube,hsieh_halasz_partons,Prem2019,Bulmash2019}. 
While no systematic classification theorem or procedure has yet been established for these models, they can be organized into broad classes~\cite{Dua_Classification_2019} based on the properties of their excitations and compactifications~\cite{Dua2019_compactify}. 

In \R{haah2014bifurcation} it was found that Haah's cubic code, the canonical example of a type-II fracton topological order with no string operators, bifurcates under ER. 
More precisely, under ER the cubic code splits into two decoupled models on a coarse-grained lattice: the original cubic code and cubic code B. The possibility of bifurcating ER was previously envisioned in Refs.~\onlinecite{Evenbly2014real,Evenbly2014class,Evenbly2014scaling} with the goal of describing critical states that violate the area law in two or more dimensions. 
Our goal is to extend the classification of topological phases in terms of ER fixed points to models that may bifurcate under ER. 
This approach faces an immediate conceptual hurdle: bifurcating models are not scale-invariant. In fact, under the conventional notion of quantum phase~\cite{chen2010local}, a bifurcating model defined on a lattice of spacing $a$ and the same model defined on a lattice of spacing $2a$ belong to different phases of matter. 
Therefore one needs to revisit the definition of thermodynamic quantum phases under such circumstances.

Fracton models are examples of bifurcating models. Under ER they may self-bifurcate or bifurcate into distinct models. 
Since self-bifurcating models give rise to an arbitrarily large number of copies of themselves under repeated ER, it is reasonable to consider them as free resources in the infrared limit. One can then formulate a generalized notion of fixed point where such free resources, i.e. the self-bifurcating models, are quotiented out. 
This generalizes the usual disentangling and projection steps in conventional ER where trivial product state degrees of freedom, which are the simplest example of a self-bifurcating state, serve as the free resource. 
The classification of bifurcating models via ER is then divided into two steps: first the classification of self-bifurcating fixed points and then the classification of quotient fixed points. 
The classification of self-bifurcating fixed points was previously studied from a resource oriented renormalization group perspective in Ref.~\onlinecite{swingle_greevy1,swingle_greevy2}. 

A similar point of view on fixed points of bifurcating ER has already played a key role in the understanding of foliated fracton models~\cite{shirley2018FoliatedFracton,shirley2017fracton,shirley2018Foliated,shirley2018universal,shirley2018Fractional}. 
The notion of foliated equivalence, local unitary equivalence up to adding stacks of 2D toric code, is rooted in the ER flow of foliated fracton models which produce stacks of 2D toric codes when they bifurcate. A stack of 2D toric codes is the simplest nontrivial self-bifurcating topological state in 3D. It is easy to see that it self-bifurcates under coarse-graining in the direction orthogonal to the toric code planes. 
The X-cube model is the canonical example with nontrivial foliated fracton order: under ER that coarse-grains by a factor of two along any axis it bifurcates, returning a copy of itself and a stack of 2D toric codes orthogonal to the axis. This reveals the X-cube model's foliation structure and that it is a foliated fixed point, as it is foliated equivalent to itself after coarse-graining. 

In this work, we propose to generalize the notion of foliated fracton equivalence to \textit{bifurcated equivalence} which allows any self-bifurcating state as a resource, thus providing an equivalence relation that is relevant for all fracton models. 
We study the ER of a large range of fracton models, including 17 of Haah's cubic codes~\cite{haah2011local} and all of Yoshida's first-order fractal spin liquids~\cite{yoshida2013exotic}, and find that they are all bifurcating fixed points. 
That is, under ER they either self-bifurcate, producing several identical copies, or they bifurcate into a copy of themselves along with some distinct models, referred to as B models. 
We demonstrate that the B models are self-bifurcating for all the examples considered. 
Our ER results are presented in table~\ref{results_ERG}. 
Moreover, we find that the form of the B models is constrained by the mobilities of the topological quasi-particles in the original models. 
For instance, the presence of a planon in the original model causes a stack of 2D toric codes to appear amongst the B models.  
Such constraints on the B models motivate us to put forward several conjectures concerning the structure of 3D topological stabilizer models that have implications for their classification.

The paper is laid out as follows: in Sec.~\ref{ER}, we introduce ER for gapped quantum phases of matter, discuss conventional and bifurcating ER fixed points, and define bifurcated equivalence and quotient superselection sectors. In Sec.~\ref{poly_frame}, we review Haah's polynomial framework. 
In Sec.~\ref{examples}, we explicitly identify bifurcating behavior in the ERG flows of a large range of models including 17 cubic codes~\cite{haah2011local} and all first-order fractal spin liquids~\cite{yoshida2013exotic}. 
In Sec.~\ref{classes_TO_ER}, we discuss possible ERG flows for the distinct classes of topological order~\cite{Dua_Classification_2019}.
In Sec.~\ref{sec:quotient_sectors}, we study quotient superselection sectors in several examples. 
In the appendix, we provide numerical results for the number of encoded qubits for all models discussed, complimented by derivations of analytical expressions for a select few. We also present the ER for the X-cube model explicitly. The explicit ER process for all other models as listed is shown in the {\small {MATHEMATICA}} file SMERG.nb provided as supplementary material. 

\section{Entanglement renormalization and phase equivalence}
In this section we introduce the notions of quantum phase of matter and entanglement renormalization. We discuss how the definition of phase equivalence is informed by ERG fixed points. We then introduce the more general notion of bifurcated equivalence based upon bifurcating ERG fixed points. We also discuss how the notion of superselection sector is generalized to quotient superselection sector for bifurcating ERG fixed point models.

\subsection{Gapped quantum phases of matter}

Throughout the paper we consider lattice Hamiltonians in 3D with short-range interactions. We only consider translation-invariant Hamiltonians defined on cubic lattices. 

A model is in a zero temperature gapped quantum phase if, in the energy spectrum, there is a finite gap between a nearly degenerate ground state subspace and the first excited state. 
Technically the energy gap must be uniformly lower bounded by a positive constant for a sequence of increasing system sizes approaching the thermodynamic limit~\cite{qimqm}. 
The energy splitting for the nearly degenerate ground state subspace must also vanish super-polynomially in the thermodynamic limit. 
The thermodynamic limit is approached as the ratio of system size $L$ to the lattice constant, or the short-distance cutoff length scale, $a$ goes to infinity $L/a\rightarrow \infty$. 
Two Hamiltonians are in the same phase if they are connected by a path of uniformly gapped Hamiltonians, possibly after adding in additional spins governed by the trivial paramagnetic Hamiltonian whose ground state is a product state. 
The addition of trivial degrees of freedom serves to stabilize the equivalence relation and allows for the comparison of models with a different number of spins per unit cell. In particular, models with the same $L$ but different $a$ can be meaningfully compared. 
Conversely, a quantum phase transition between gapped phases necessarily involves the gap closing. 

The physical characteristics of zero temperature gapped phases can be studied via their ground states. 
Such ground states are in the same phase if and only if they are related by a quasi-adiabatic evolution~\cite{hastings2005quasiadiabatic}, possibly after stabilization by adding spins in a trivial product state. 
Constant-depth local unitary circuits are often used as a proxy for phase equivalence~\cite{chen2010local}, although strictly speaking they provide only a sufficient condition~\cite{Haah2018}. 
To compare bulk phase equivalence more generally one should consider stabilized (approximate) locality-preserving unitary maps, or quantum cellular automata. 
For dispersionless commuting projector Hamiltonians it appears sufficient to consider stabilized exact locality-preserving unitary maps, up to a change in the choice of local Hamiltonian terms that preserves the ground space and gap but may shift higher energy levels. 

Throughout this work, topological orders are stable gapped quantum phases defined by a topological degeneracy on torus and characterized by the existence of nontrivial quasiparticles. This includes: 
\begin{itemize}
    \item Topological quantum liquid phases, whose ground state degeneracies on the torus have a constant upper bound. It is widely believed that these phases can be described by topological quantum field theories (TQFT) at low energy. We will thus refer to topological orders belonging to class as TQFT phases.
    \item Fracton phases whose ground state degeneracies do not have a constant upper bound. The low-energy behaviors of these phases are not described by any conventional TQFTs. 
\end{itemize}

\subsection{Entanglement renormalization transformations}
\label{ER}
We now describe ER transformations following the pioneering work~\cite{haah2014bifurcation}. Given a gapped Hamiltonian, an ER transformation consists of the following steps:
\begin{enumerate}
    \item Coarse-grain the system by enlarging the unit cell by a factor $c>1$.
    \item Apply local unitary transformations to the Hamiltonian to remove short-range entanglement.
    \item Project out local degrees of freedom that are completely disentangled into a trivial state. 
\end{enumerate}
The result is a new Hamiltonian in the same phase of matter, defined on the coarse-grained lattice.

In order to maintain a well-defined notion of phase throughout a renormalization procedure, we fix the ratio $L/a$ to be infinite by taking an infinite system size $L$ while successively coarse-graining a finite lattice constant $a$ by a factor $c>1$. 
For a system of finite size, coarse-graining amounts to reducing the number of unit cells in the system. 
In the thermodynamic limit this remains true, as the ratio $\frac{L_1/a_1}{L_2/a_2}$ of the number of unit cells before and after coarse graining is equal to $c$ even though both numerator and denominator diverge. 

% \begin{figure}[t]
% \centering
% \includeTikzrm{}{
% \begin{tikzpicture}[baseline= (a).base]
% \node[scale=.9] (a) at (0,0){
% \begin{tikzcd}[execute at end picture={
% \draw[black, dotted]  (-4.6,-1.4) rectangle (0,2.6);
% \draw[black, dotted] (-0,-1.4) rectangle (4.6,2.6);}]
%         |[alias=H_4]| H_4\arrow[r] & |[alias=H_{FP_1}]|H_{\text{FP}_1} \arrow[loop above] & &[0.22em] |[alias=H_7]| H_7 \arrow[dr]& H_9\arrow[d] &\\
%          & |[alias=H_2]|H_2 \arrow[u] & H_3\arrow[ul] &[0.25em]   & H_{\text{FP}_2}\arrow[loop below] & H_6 \arrow[l] \\
%          & H_1\arrow[u] & &[0.25em] H_8 \arrow[ur] &  & |[alias=H_5]| H_5\arrow[ul]\\
% \end{tikzcd}
% };
% \end{tikzpicture}
% }
% \caption{Conventional RG flows within gapped phases. $H_i$ indicate different models that flow towards fixed points models $H_{FP_i}$. Here, $H_{1\text{--}4}$ are in one phase while $H_{5\text{--}9}$ are in another. These phases are represented by regions in parameter space, separated by a dotted phase transition line.} 
% \label{Conventional RG}
% \end{figure}

\begin{figure}[t]
\centering
\includegraphics[scale=0.8]{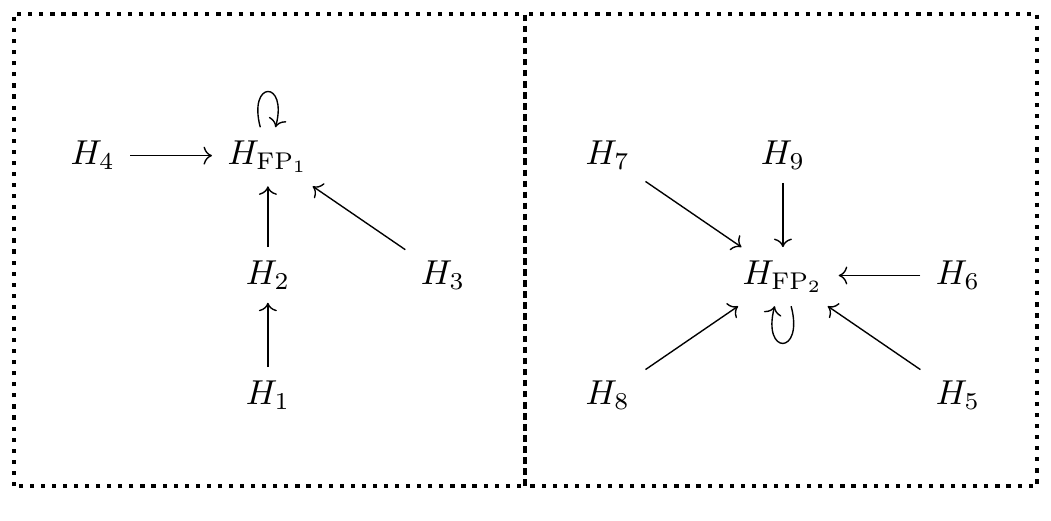}
\caption{Conventional RG flows within gapped phases. $H_i$ indicate different models that flow towards fixed points models $H_{FP_i}$. Here, $H_{1\text{--}4}$ are in one phase while $H_{5\text{--}9}$ are in another. These phases are represented by regions in parameter space, separated by a dotted phase transition line.} 
\label{Conventional RG}
\end{figure}

Conventionally it is expected that ERG flows within a gapped phase carry models towards a representative fixed point model that is invariant under the ERG flow. We depict such a situation schematically in Fig.~\ref{Conventional RG}, which shows Hamiltonians flowing towards RG fixed points in two gapped phases. 
The set of stable ERG fixed points then suffice to classify the gapped phases.

We remark that under ER copies of the trivial Hamiltonian are produced at an exponential rate, as the trivial Hamiltonian  in $D$ spatial dimensions splits into $c^D$ copies of itself under coarse-graining by a factor $c$. 
Hence to find any fixed points we must mod out by this self-replicating trivial model. 
Furthermore, for an ERG fixed point model to lie in the same phase after a change of scale, the trivial Hamiltonian must also be modded out in the definition of phase. This is clearly a necessary condition for a definition of phase that does not rely on a choice of lattice scale, as is commonly desired.

Fixed point Hamiltonians under the form of ER we consider here must satisfy 
\begin{eqnarray}
     UH(a)U^\dagger\equiv H(ca)  \, ,
    \label{convRG}
\end{eqnarray} 
for some finite-depth quantum circuit $U$, where $c$ is the coarse-graining factor.   
The Hamiltonian $H(ca)$ has the same local terms as the original Hamiltonian $H(a)$, but the degrees of freedom sit on the sites of a coarse-grained lattice with spacing $ca$. 
The equivalence $\equiv$ denotes equality up to the addition of disentangled spins governed by the trivial Hamiltonian to either side, and changing the choice of local Hamiltonian terms in a way that preserves the ground space. In particular, this implies that $H(a)$ and $H(ca)$ are in the same phase. 
This relation holds for commuting projector Hamiltonians that describe TQFT phases~\cite{qdouble,Levin2005,koenig2010quantum,walker2012,williamson2016hamiltonian}, and hence they are fixed points under ER~\cite{Konig2009,Aguado2008}. This allows the lattice scale to be ignored in the definition of TQFT phases. 

\subsection{Bifurcating entanglement renormalization and bifurcated equivalence}
Models that are governed by conventional ERG fixed-point Hamiltonians with topological order fall into the category of TQFT phases. All known two-dimensional fixed-point models are of this type. 
However, in three dimensions, fracton models with topological order have been discovered that bifurcate under ER. 
This means that performing one step of ER on a model produces multiple nontrivial decoupled models. That is, for some finite-depth quantum circuit $U$ we have
\begin{eqnarray}
    UH(a)U^\dagger\equiv H_1(c a)+ H_2(c a)+...+H_b(c a) \, ,
\label{bif}
\end{eqnarray}
where $c$ is again the coarse-graining factor and $b$ is the number of nontrivial decoupled models, or branches. We remark that this decomposition into decoupled models may not be unique. 
An illustration of a bifurcating ER transformation is shown in Fig.~\ref{bifRGillus}.

\begin{figure}[t]
\centering
\includegraphics[scale=1]{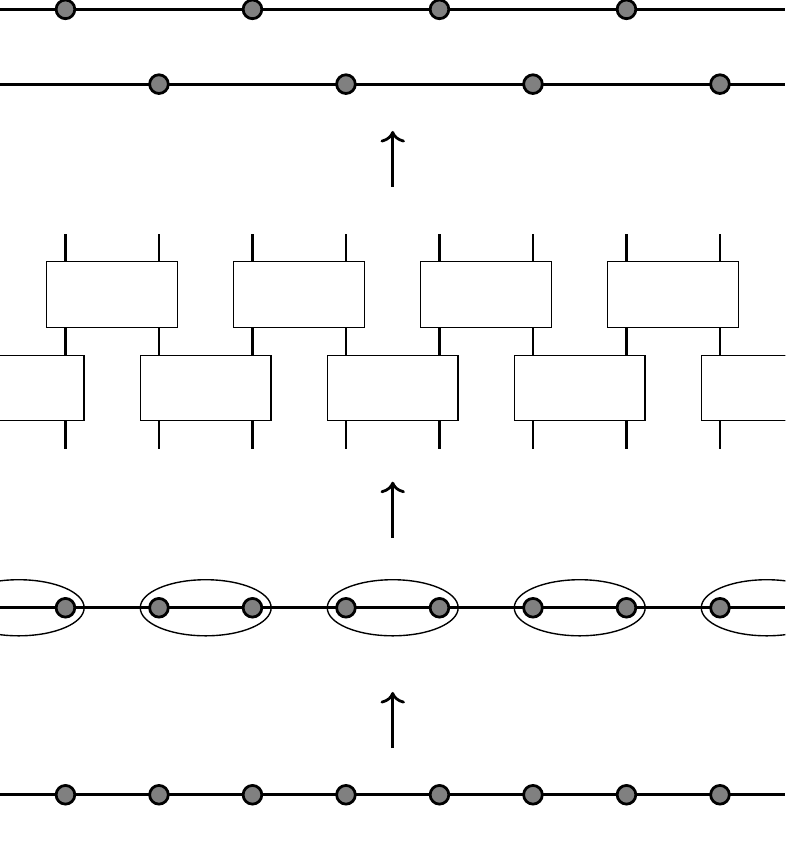}
\caption{An illustration of bifurcating ER for a 1D model. We start with a translation-invariant spin chain, perform coarse-graining by grouping pairs of spins into new sites and then apply a local unitary consisting of two-spin gates that is translation-invariant on the coarse-grained lattice. 
This disentangles two decoupled copies of the original spin chain on the coarse-grained lattice.} 
\label{bifRGillus}
\end{figure}

% \begin{figure}[t]
% \centering
% \includeTikzrm{bifRGillus}{
% \begin{tikzpicture}[scale=0.95]
% \clip (0.3,-0.5) rectangle (8.7,8.5);

% \foreach \x in {0,1,...,9}
% \foreach \y in {0,2}
% \draw[fill=gray,thick]  (\x,\y) circle[radius=0.1];

% \foreach \x in {0.5,2.5,...,8.5}
% \draw (\x,2) ellipse (0.7 and 0.3);

% \foreach \xo in {-0.2,1.8,3.8,5.8,7.8}
% \foreach \yo in {4}
% \foreach \stx in {1.4}
% \foreach \sty in {0.7}
% {\draw[fill=white] (\xo,\yo) rectangle (\xo+\stx,\yo+\sty); 
% \draw[fill=white] (\xo+1,\yo+1) rectangle (\xo+1+\stx,\yo+1+\sty);} 

% \foreach \x in {0,2,...,8}
% \draw[fill=gray,thick]  (\x,7.6) circle[radius=0.1];

% \foreach \x in {1,3,...,9}
% \draw[fill=gray,thick]  (\x,8.4)
% circle[radius=0.1];

% \begin{pgfonlayer}{bg}
% \foreach \x in {1,2,...,8}
% \draw[-,thick] (\x,3.7) -- (\x,6);
% \draw[-,thick] (0.3,2) -- (8.7,2) ;
% \draw[-,thick] (0.3,0) -- (8.7,0) ;
% \draw[-,thick] (0.3,7.6) -- (8.7,7.6) ;
% \draw[-,thick] (0.3,8.4) -- (8.7,8.4) ;
% \end{pgfonlayer}

% \foreach \yo in {0.5,2.75,6.5}
% \draw[->, thick] (4.5,\yo)--(4.5,\yo+ 0.6);

% \end{tikzpicture}}
% \caption{An illustration of bifurcating ER for a 1D model. We start with a translation-invariant spin chain, perform coarse-graining by grouping pairs of spins into new sites and then apply a local unitary consisting of two-spin gates that is translation-invariant on the coarse-grained lattice. 
% This disentangles two decoupled copies of the original spin chain on the coarse-grained lattice.} 
% \label{bifRGillus}
% \end{figure}

A model $H(a)$ is a \textit{bifurcating} fixed point if any of the resulting models on the right hand side of Eq.~\eqref{bif} are equivalent to it, i.e. $H_1(a)=H(a)$ without loss of generality. In particular, all of the examples that appear in this work are bifurcating fixed points. Bifurcating fixed points can be either \textit{self-bifurcating} fixed points or \textit{quotient} fixed points which are defined as follows:

\begin{itemize}
    \item A model $H(a)$ is a \textit{self-bifurcating} fixed point with \textit{branching number} $b$ if all of the resulting models  on the right hand side of Eq.~\eqref{bif} are equivalent to it, i.e. $H_i(a)=H(a)$ for $i=1,\dots,b$. For example, a self-bifurcating fixed point with branching number $b=2$ satisfies
    \begin{eqnarray}
        UH_{SB}(a)U^\dagger\equiv H_{SB}(c a)+ H_{SB}(c a) \, .
        \label{eq:sbex}
    \end{eqnarray}
    \item A bifurcating fixed point model $H(a)$ is a \textit{quotient} fixed point if the models $H_i(a)$ for $i\geq 2$ are self-bifurcating fixed points that are not equivalent to $H_1(a)$. More specifically we may refer to such an $H(a)$ as a quotient fixed point with respect to the self-bifurcating Hamiltonian $\sum_{i\geq 2}H_i(a)$. An example of a quotient fixed point model, with respect to two decoupled $b=2$ self-bifurcating fixed points, is given by 
    \begin{eqnarray}
        UH_{B}(a)U^\dagger\equiv H_{B}(c a)+ H_{SB1}(c a)+ H_{SB2}(c a) \, ,
    \label{eq:qfpex}
    \end{eqnarray}
where $H_{SB1}$ and $H_{SB2}$ both satisfy Eq.\eqref{eq:sbex}.
\end{itemize}

\begin{figure}[t]
\centering
% %
\sidesubfloat[]{\includegraphics[scale = 1.12]{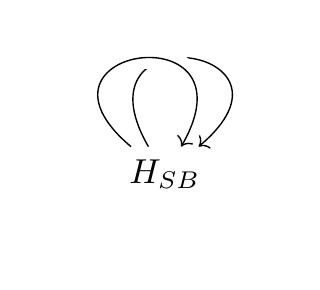}}\\
\sidesubfloat[]{\includegraphics[scale = 1.12]{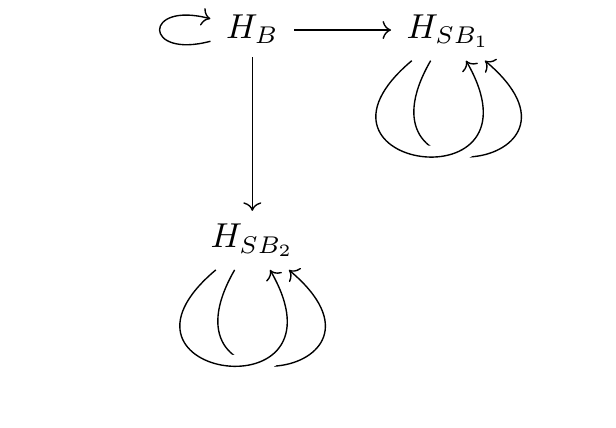}}\\
\caption{Bifurcating RG flow. (a) $H_{SB}$ denotes self-bifurcating fixed point models. (b) $H_{B}$ denotes bifurcating fixed point model while $H_{SB1}$ and $H_{SB2}$ denote self-bifurcating fixed point models.} 
\label{Bifurcating RG}
\end{figure}

% \begin{figure}[t]
% \centering
% \sidesubfloat[]{\includeTikzrm{BFP}{
% \begin{tikzcd}[ampersand replacement=\&, execute at end picture={}]
% |[alias=SB]|H_{SB} \arrow[out=120, in=40,loop] \arrow[out=140, in=60, loop, crossing over]\\
% \end{tikzcd}
% }}\\
% \sidesubfloat[]{\includeTikzrm{B}{
% \begin{tikzcd}[ampersand replacement=\&]
% \& \& |[alias=B_2]| H_{B}\arrow[loop left]{} \arrow[d]{} \arrow[r]{} \& |[alias=SB_1]|H_{SB_1} \arrow[out=240, in=320, loop] \arrow[out=220, in=300, loop,crossing over]\\[6ex]
% \& \& |[alias=SB_2]|H_{SB_2}\arrow[out=240, in=320, loop] \arrow[out=220, in=300, loop,crossing over]\& \\ 
% \end{tikzcd} }}
% \caption{Bifurcating ERG flow diagrams. 
% Our notation specifies the flow of a model under one step of ER by following all arrows that originate from that model. 
% (a) $H_{SB}$ denotes a self-bifurcating fixed point model with $b=2$. 
% The two arrows indicate that under ER, the model flows to two copies of itself, see Eq.~\eqref{eq:sbex}. (b) $H_{B}$ denotes a quotient fixed point model, while $H_{SB_1}$ and $H_{SB_2}$ denote self-bifurcating fixed point models with $b=2$. 
% The single self arrow pointing to $H_B$ indicates that one copy of the original model $H_B$ persists throughout ER and hence it is a quotient fixed point, see Eq.~\eqref{eq:qfpex}. }
% \label{Bifurcating RG}
% \end{figure}

In Fig.~\ref{Bifurcating RG}, we represent the above self-bifurcating and quotient ERG fixed point examples using a diagrammatic notation. 
In our bifurcating ERG fixed point diagrams the models resulting from one step of ER, corresponding to the right hand side of Eq.~\eqref{bif}, are found by following all arrows leaving a model, corresponding to the left hand side of Eq.~\eqref{bif}. 
Such a diagram represents a generalized fixed point when all arrows leaving models in the diagram return to models within the diagram. This captures conventional fixed points, bifurcating fixed points and limit cycles, which can be removed by increasing the amount of coarse-graining performed during one step of ER. 
Further examples of bifurcating ERG fixed point models are presented in section~\ref{examples}, with the results summarized in table~\ref{results_ERG}. We remark that since Eq.~\eqref{bif} is not unique for a given Hamiltonian $H(a)$ it is possible for a model to be a self-bifurcating fixed point under one ERG flow, and a quotient fixed point under another, see section~\ref{ex:cc14} for such an example. 

In contrast to conventional ER fixed points, bifurcating ER fixed point Hamiltonians on different lattices are not in the same phase, i.e. $H(a)$ is not phase equivalent to $H(c a)$. This is evident from the presence of nontrivial models $H_{i}$ with ${i\geq 2}$ on the right hand side of Eq.~\eqref{bif}. 
For a self-bifurcating model, $H_{SB}(a)$, that satisfies Eq.\eqref{eq:sbex} with $c=2$ the reason for this inequivalence is especially clear: two copies of $H_{SB}(2a)$ cannot be equivalent to a single copy unless $H_{SB}$ is in the trivial phase. This is depicted in Fig.~\ref{phase_equiv_SB}. 

% \begin{figure}[t]
% \centering
% \includeTikzrm{phase_equiv}{
% \begin{tikzpicture}
% \draw[dotted, very thick] (0,2) -- (2,2) -- (2,4);
% \draw (0,0) rectangle (4,4);
% \draw (1,3) node {$H_{SB}(2a)$};
% \draw (1,1) node {$H_{SB}(a)$};
% \draw (3,3) node {$H_{SB}(2a)^{\otimes 2}$};
% \draw[->] (1.66,1) to [bend right = 45] (3.22,2.8);
% \end{tikzpicture}}
% \caption{Under one step of ER indicated by an arrow, a self-bifurcating model denoted $H_{SB}$ with a lattice constant $a$ becomes two copies of the model on the coarse-grained lattice $H_{SB}(2a)$. Since $H_{SB}(2a)^{\otimes 2}$ is not in the same phase as $H_{SB}(2a)$, $H_{SB}(a)$ is not in the same phase as $H_{SB}(2a)$. A phase boundary is indicated by a dotted line.} 
% \label{phase_equiv_SB}
% \end{figure}

\begin{figure}[t]
\centering
\includegraphics[scale=1.4]{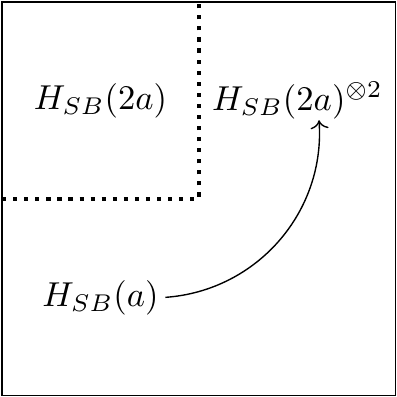}
\caption{Under one step of ER indicated by an arrow, a nontrivial self-bifurcating model denoted $H_{SB}$ with a lattice constant $a$ becomes two copies of the model on the coarse-grained lattice $H_{SB}(2a)$. Since $H_{SB}(2a)^{\otimes 2}$ is not in the same phase as $H_{SB}(2a)$, $H_{SB}(a)$ is not in the same phase as $H_{SB}(2a)$. A conventional phase boundary is indicated by a dotted line.} 
\label{phase_equiv_SB}
\end{figure}
Self-bifurcating fixed point models produce copies of themselves at an exponential rate under ERG flow. This rate is given by the branching number, which obviously satisfies $b \leq c^D$, but can be further shown to satisfy $b < c^{D-1}$ for states that satisfy an area law~\cite{haah2014bifurcation}, which are the relevant ones in the study of gapped phases. 
This is analogous to the exponential splitting of trivial models in the ERG flow of a conventional fixed point model. 
Inspired by the modding out of trivial models in the definition of gapped phase equivalence, here we introduce the notion of \textit{bifurcated} equivalence where all self-bifurcating models are modded out. 
More specifically, we write $H_1(a_1) \sim_{H_B} H_2(a_2)$ if $H_1(a_1)$ stacked with some number of copies of $H_B$ lies in the same conventional phase as $H_2(a_2)$ stacked with some number of copies of $H_B$. This is bifurcated equivalence with respect to a self-bifurcating fixed point model $H_B$. When $H_B$ is a trivial Hamiltonian we recover the conventional phase equivalence, when $H_B$ is a stack of 2D topological orders we recover foliated fracton equivalence. 
Similarly we write $H_1(a_1) \prescript{}{H_{B_1}}{\sim}_{H_{B_2}} H_2(a_2)$ if $H_1(a_1)$ stacked with copies of $H_{B_1}$ is equivalent to $H_2(a_2)$ stacked with copies of $H_{B_2}$. 
More generally we denote bifurcated equivalence by ${H_1(a_1) \sim H_2(a_2)}$ whenever there exist self-bifurcating models $H_{B_1},H_{B_2}$ such that ${H_1(a_1) \prescript{}{H_{B_1}}\sim_{H_{B_2}} H_2(a_2)}$. 
See  Fig.~\ref{CC_13_14_17}~(b) for a nontrivial example involving both foliated fracton equivalence and the more general bifurcated equivalence. 

The bifurcated equivalence relation serves to essentially remove dependence on the lattice scale from the equivalence class of a quotient fixed point Hamiltonian since $H(a) \sim H(ca)$ in that case.  
We remark that models may be bifurcated equivalent, even when they are not in the same conventional gapped phase. In particular any self-bifurcating model is in the trivial bifurcated equivalence class, even though the model may be in a nontrivial conventional phase. In Fig.~\ref{phase_equiv_SB} the dotted line denotes a conventional phase boundary, whereas the whole diagram lies in the same trivial bifurcated equivalence class. It is also useful to generalize the equivalence relation $\equiv$ accordingly by allowing for stacking with arbitrary self-bifurcating models, which we denote $\cong$. 
With this definition in hand the condition for a quotient fixed point resembles the conventional fixed point condition 
\begin{eqnarray}
     UH(a)U^\dagger \cong H(ca) \, .
\end{eqnarray}

\subsection{(Quotient) superselection sectors}

A nontrivial superselection sector on some region $\mathcal{R}$ is an excitation that can be supported on $\mathcal{R}$ but not created by any operator within a neighborhood of $\mathcal{R}$. Superselection sectors are equivalent if they are related by the application of an operator within a neighborhood of $\mathcal{R}$, or equivalently fusion with a local excitation in $\mathcal{R}$. 
Under ER excitations may split due to a change in the choice of local Hamiltonian terms allowed in the $\equiv$ relation. This may create local excitations in the trivial Hamiltonians that are modded out by the $\equiv$ relation. Hence for superselection sectors to be invariant under ER, local excitations must be modded out in their definition.  

For the same reason, to define \textit{quotient superselection sectors} (QSS) that are invariant under quotient ER we must mod out any excitations that can flow into a self-bifurcating fixed point model, as these models are modded out by the $\cong$  relation. 
Representatives of potential QSS are then given by the fixed point excitations under quotient ER.  
This captures the notion of QSS for foliated fracton models as a special case when the self-bifurcating model is taken to be a stack of 2D topological orders~\cite{shirley2018Fractional}. 

\section{Entanglement renormalization in the polynomial framework}
\label{poly_frame}

In this section we introduce translation-invariant stabilizer models,  Clifford ER transformations and phase equivalence relations, along with their descriptions in the language of polynomial rings from commutative algebra.

\subsection{Translation-invariant stabilizer models}

Translation-invariant stabilizer Hamiltonians are specified by a choice of mutually commuting local Pauli stabilizer generators $h^{(i)}$. The generators become the interaction terms in a Hamiltonian, 
\begin{align}
    H = \sum_{\vec{v}} (\openone -  h^{(i)}_{\vec{v}})  \, ,
\end{align}
where $\vec{v}$ are lattice vectors. In the above equation, $h^{(i)}_{\vec{v}}$ indicates a local generator $h^{(i)}$ after translation by a lattice vector $\vec{v}$. The local generator $h^{(i)}$ is a tensor product of local Pauli operators acting on a set of qubits or qudits. Without loss of generality, we consider stabilizer models on a cubic lattice.

\subsubsection{Clifford phase equivalence and entanglement renormalization}
Maps between translation-invariant stabilizer Hamiltonians are given by locality-preserving Clifford operations, which map local Pauli operators to local Pauli operators. These Clifford operations include local Clifford circuits which are generated by CNOT, Phase and Hadamard gates, and nontrivial Clifford cellular automata, which are required to disentangle certain invertible phases~\cite{Haah2018}.  In addition, locality-preserving automorphisms of the lattice, such as the redefinition of coordinates via modular transformations and coarse graining, are also included.  

Phase equivalences of translation-invariant stabilizer Hamiltonians are given by locality-preserving Clifford operations up to stacking with trivial models. 

For our Clifford ER transformations, we restrict to local Clifford circuits, coarse-graining, and discarding trivial models, as the modular transformations and other nontrivial locality-preserving operations can be moved to a single final step when comparing models. A change in the choice of local stabilizer generators that preserves the stabilizer group is also allowed when comparing two models, as in the $\equiv$ relation above. As CNOT, Phase and Hadamard gates generate the Clifford group, they are sufficient to implement ER of Pauli stabilizer codes. 

In 2D all translation-invariant topological stabilizer models were classified and shown to be equivalent, under locality-preserving Clifford operations, to copies of the 2D toric code. This implies that all translation-invariant topological stabilizer models in 2D flow to ER fixed points. Conversely in 3D examples of translation-invariant topological stabilizer models that bifurcate under ER are known, and the classification problem remains completely open, due to the existence of fracton models. Our goal is to study the bifurcating ERG flows of known fracton stabilizer models to gain clues about the 3D classification problem.

The examples considered in this work are all in CSS form~\cite{PhysRevA.54.1098,Steane2551}, which should be preserved under ER, hence we have found it sufficient to consider Clifford circuits that consist of CNOT gates alone. 
In particular, the unitaries used in the ER of our examples are given by Clifford circuits ${U=U_1 U_2...U_N}$ where $N$ is finite and each layer of gates $U_i$ is a translation-invariant tensor product of CNOT gates.
 
\subsection{The polynomial framework}
 
Translation-invariant stabilizer Hamiltonians can be conveniently expressed in terms of polynomials. The use of a polynomial description in a similar context dates back to work on classical cyclic codes~\cite{Imai1977TDC,cecc,multivariable_cecc}. 
The polynomial approach for quantum codes on a lattice was primarily developed by Haah. Interested readers are directed to Ref.~\onlinecite{haah2013commuting} for further details. We proceed by introducing several definitions from the polynomial language that are necessary and sufficient to describe ER. These definitions are demonstrated via examples.

\subsubsection{The stabilizer map}

For a stabilizer model on a cubic lattice, the stabilizer generators supported on a cubic unit cell can be expressed in terms of the position labels of the vertices on the cube as shown below
\begin{eqnarray}
\begin{array}{c}
\drawgenerator{xz}{xyz}{x}{xy}{y}{1}{yz}{z}
\end{array}
\, .
\label{pos_unit_cella}
\end{eqnarray}

The canonical example of a type-II model is Haah's code, or cubic code 1, which has the following stabilizer generators 
\begin{align}
\begin{array}{c}
\drawgenerator{XI}{II}{IX}{XI}{IX}{XX}{XI}{IX}
\quad
\drawgenerator{ZI}{ZZ}{IZ}{ZI}{IZ}{II}{ZI}{IZ}
\end{array}
\, .
\end{align}
In the $X$-stabilizer generator, the sites on which the first qubit is acted upon by the Pauli $X$ operator are at positions $1$, $xy$, $xz$ and $yz$ of the unit cell. 
We take this set of positions $(1,xy,xz,yz)$ and write a polynomial corresponding to this set $1+xy+xz+yz$. Similarly, the polynomial corresponding to the action of the Pauli $X$ operator on the second qubit on vertices in the unit cell is $1+x+y+z$. 
The polynomials corresponding to the action of the $Z$ operator on the first and second qubits on each site involved in the $Z$-stabilizer generator are given by $xy+xz+yz+xyz$ and $x+y+z+xyz$.  These polynomials refer to exponents of Pauli operators, and hence the addition of polynomials corresponds to the multiplication of operators. 
One can consider this polynomial representation to be a map from the position labels to the set of Pauli operators. This map is called the stabilizer map which can be written as a $2q\times t$ matrix where $q$ is the number of qubits on each vertex in the unit cell and $t$ is the number of stabilizer generators per unit cell in the translation-invariant Hamiltonian. In such a matrix, each column represents a term in the Hamiltonian and all translates of these terms can be generated by acting on the column by multiplication with monomials of translation variables. 
We illustrate the polynomial representation for a non-CSS model using the example of Wen's plaquette model in Fig.~\ref{poly_non_CSS}.

In this paper, we focus on CSS models for which the stabilizer map takes the form  
\begin{align}
    \sigma = \left(\begin{array}{cc}
         \sigma_X & 0  \\
          0 & \sigma_Z 
    \end{array}\right)
    \, ,
\end{align}
where $\sigma_{X\,(Z)}$ is the map for the $X\,(Z)$-sector. Each column, labeled $C_i$, specifies a stabilizer generator. 
The rows of the $\sigma_{X\,(Z)}$ block are labeled $R^X_i$ ($R^Z_i$), where the row index $i$ is the index of the qubit in the unit cell. For example, the stabilizer map for cubic code 1 is 
\begin{eqnarray}
    \sigma = \left(\begin{array}{cc}
         1+xy+xz+yz & 0   \\
         1+x+y+z & 0 \\
          0 & xy+xz+yz+xyz \\
          0 & x+y+z+xyz 
    \end{array}\right)
     .
    \label{CC1rep1}
\end{eqnarray}
We remark that all stabilizers of the model can be generated by the action of the stabilizer map on a column of translation variables. 
The translation action can be expressed in terms of the position variables $x,y,z$. For example, the $X$-stabilizer on the unit cell at position $x$ relative to the origin is denoted 
\begin{align}
\left(\begin{array}{cc}
         (1+xy+xz+yz)x \\
         (1+x+y+z)x\\
         0\\
         0
    \end{array}\right)
    \, .
\end{align}
Using this, or any other translation of the $X$-stabilizer generator in Eq~\eqref{CC1rep1} lead to the same Hamiltonian. 
Hence, multiplying columns by monomials results in an equivalent stabilizer map. 
For example, we could divide the second column of the cubic code stabilizer map by $xyz$, corresponding to a unit translation in the negative direction along each axis. The resulting equivalent stabilizer map can be written as
\begin{eqnarray}
        \sigma = \left(\begin{array}{cc}
         1+xy+xz+yz & 0   \\
         1+x+y+z & 0 \\
          0 & 1+\overline{x}+\overline{y}+\overline{z} \\
          0 & 1+\overline{x}\overline{y}+\overline{x}\overline{z}+\overline{y}\overline{z}
    \end{array}\right)
    \label{CC1rep2}
\end{eqnarray}
where we have introduced the inverse variables $\overline{x}, \overline{y}$ and $\overline{z}$ which satisfy $x\overline{x}=1$, $y\overline{y}=1$ and $z\overline{z}=1$. In the terminology of commutative algebra, introducing negative powers for translation variables involves going from a polynomial ring to a Laurent polynomial ring.

\begin{figure}
    \centering
    \includegraphics[scale=0.94]{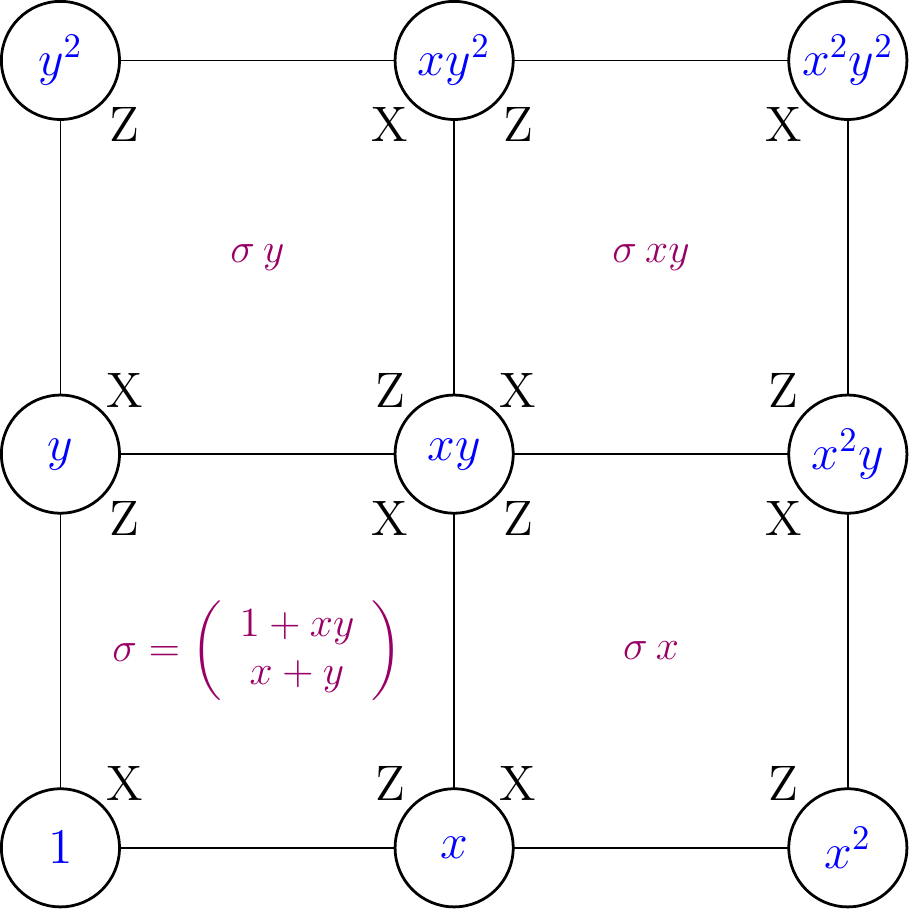}
    \caption{The stabilizer map of Wen's plaquette model. The qubit positions are expressed in terms of the translation variables $x,y,z$ in blue. The polynomial representation for each plaquette stabilizer generator is written in purple, where the first (second) row contains the positions acted upon by the $X$ ($Z$) operator. These representations can be obtained by applying the stabilizer map $\sigma$ to monomials that specify the positions of the respective stabilizer generators.}
    \label{poly_non_CSS}
\end{figure}

\subsubsection{Entanglement renormalization in the polynomial language}
As discussed above, in the Clifford ER procedure after coarse-graining certain operations are allowed. These operations consist of acting on the stabilizer generators with Clifford gates, changing the choice of generators for the same stabilizer group and shifting the lattice sites. In the polynomial language, these operations are represented by matrices acting on the stabilizer map, from the left for Clifford gates and qubit shifts, and from the right for the shifting and redefinition of stabilizer generators. These operations were described in Ref.~\onlinecite{haah2014bifurcation}. 
As our focus is on CSS models, we restrict our discussion to Clifford circuits made up of CNOT gates. 
In which case the action of ER operations on the stabilizer map are as follows:
\begin{itemize}[leftmargin=*]
    \item Row operations 
    \begin{itemize}
        \item Elementary row operations: \\
        An elementary row operation on a stabilizer map with rows $R^{X(Z)}_i$ is specified by two row indices $a\neq b$ and a monomial $f$ and acts as  follows in the $X$-sector, 
        \begin{align}
        \phantom{---}\text{CNOT}(a,b,f):{R^X_a\mapsto R^X_a +f(x,y,z) R^X_b} \, .
        \end{align}
        This operation corresponds to a translation-invariant implementation of CNOT gates between the target qubits specified by $a$ and $f$ with the control qubits specified by $b$. 
        The corresponding action in the $Z$-sector is given by 
        \begin{align}
        \phantom{---}\text{CNOT}(a,b,f): R^Z_b \mapsto R^Z_b+f(\overline{x},\overline{y},\overline{z}) R^Z_a \, .
        \end{align} 
       
    \item Row multiplication by a monomial: \\ 
    Multiplying any of the rows in the stabilizer map by a monomial corresponds to shifting those qubits in some direction. 
    For the polynomial entries $\alpha_{ab}$, in a row specified by constant $a$,  the transformation ${\alpha_{ab}\rightarrow x^i y^j z^k \alpha_{ab}}$ is allowed, for any finite integers $i,j,k$. 
    \end{itemize}
    \item{Column operations}
    \begin{itemize}
        \item Elementary column operations: \\
        An elementary column operation on a stabilizer map with columns $C_i$ is specified by two column indices $a\neq b$ and a monomial $f$ and acts as  follows: 
        \begin{align}
        \text{Col}(a,b,f): {C_a\mapsto C_a +f(x,y,z) C_b} \, .
        \end{align} 
        This changes the choice of stabilizer generators by replacing those corresponding to $C_a$ by products of themselves with translates of the generator corresponding to $C_b$. 
        Such a change of choice of generators results in a phase equivalent Hamiltonian with the same ground space and shifted excitation energy levels.
        
        \item  Column multiplication by a monomial: \\
        This corresponds to changing the choice of a stabilizer generator, translating it by a monomial. This has no effect on the Hamiltonian described by the stabilizer map, as it already includes all translations of the generators. 
        For example, we used such a transformation above to modify the polynomial representation for the $Z$-stabilizer term of the cubic code to go from Eq.~\eqref{CC1rep1} to Eq.~\eqref{CC1rep2}. 
    \end{itemize}
\end{itemize}

\subsubsection{Modular transformations}

The Clifford ER process only involves equivalences between models generated by coarse-graining, local Clifford gates, shifting the lattice sites and changing the choice of column generators. 
More generally modular transformations, which are locality-preserving automorphisms of the cubic lattice including shear transformations, preserve the bulk properties of a model.  Hence, when deciding phase equivalence, modular transformations need to be taken into account. Formally, they correspond to the redefinition of translation variables $(x,y,z)$ to $(f_1(x,y,z),f_2(x,y,z),f_3(x,y,z))$, where $f_i(x,y,z)$ for $i=1,2,3$ are monomials in the translation variables that induce a bijection of the lattice sites. 

Modular transformations are not used during the ER process. However, they are used when checking the equivalence of models that result from ER. 
We have also used them to find equivalences between some of the cubic codes. For example, cubic code 5 is related to cubic code 9 (and cubic code 15 is related to cubic code 16) via a modular transformation and hence they are in the same phase~\cite{Dua_Classification_2019}.

\subsubsection{The excitation map}

\begin{figure}
    \centering
    \includegraphics[scale=0.94]{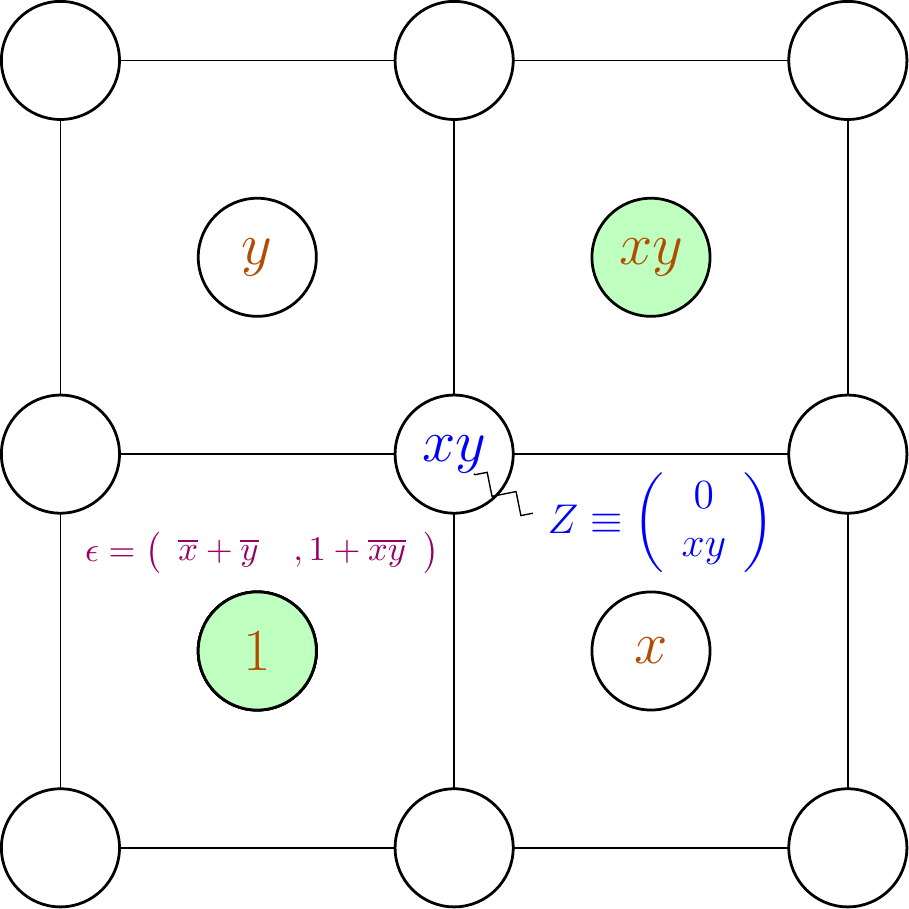}
    \caption{Excitation map for Wen's plaquette model. The positions of stabilizer generators are written on the dual lattice. All but one of the labels for qubits and Pauli operators have been omitted for simplicity. 
    A local operator $Z$ acting at position $xy$ and its polynomial column representation are depicted. The first (second) row of the column contains the position where $X$ ($Z$) acts. 
    Applying the excitation map $\epsilon$, also depicted, to the polynomial column returns the positions of excited stabilizers $1+xy$ (indicated by shaded circles).}
    \label{excmap_WP}
\end{figure}

The excitations created by the action of a Pauli operator can be found by applying the excitation map $\epsilon$ to the polynomial representation of that Pauli operator. 
The map $\epsilon$ is expressed in terms of the stabilizer map $\sigma$ via
\begin{align}
\epsilon:= \sigma^\dagger \lambda,
&& \text{where} &&
    \lambda=\left(
\begin{array}{cc}
     \bf{0} & \bf{1}  \\
     \bf{-1} & \bf{0} 
\end{array}\right)
\end{align}
is the symplectic matrix and $\bf{1}$ denotes the $q\times q$ identity matrix. The rows in the excitation map correspond to excitations of different stabilizer generators. The operators that do not excite translations of a particular generator correspond to columns that give 0 when acted upon by the corresponding row of $\epsilon$. 
For example, one can see from $\lambda$ that no $X$ stabilizer generators are excited by the action of a Pauli $X$ operator, as expected. 
More generally, the kernel of $\epsilon$ contains the operators that commute with all local Hamiltonian terms. In particular, the condition that the Hamiltonian terms themselves commute can be recast as $\epsilon \sigma = 0$, i.e. $\im\, \sigma \subseteq \ker\, \epsilon$. 
This containment is saturated, $\im\, \sigma = \ker\, \epsilon$,  for infinite boundary conditions if and only if the Hamiltonian described by $\sigma$ is topologically ordered~\cite{haah2013commuting,haah2013}. Since $\ker\, \epsilon$ consists of operators of finite extent in the polynomial formalism, $\im\, \sigma = \ker\, \epsilon$ implies that any operator that commutes with the Hamiltonian is in the span of the stabilizer group and so can be written as a sum of products of generators. This implies that any local operator must act as the identity on the degenerate ground space of the Hamiltonian with periodic boundary conditions, up to a proportionality constant that may be~0. 

In Fig.~\ref{excmap_WP}, we depict how the excitation map gives the positions of excited stabilizers due to the action of a local operator $Z$ on a qubit at position $xy$ in the Wen-plaquette model. Another phase equivalent example is the $\mathbb{Z}_2$ toric code whose stabilizer map is given by
\begin{eqnarray}
    \sigma=\left(\begin{array}{cc}
     1+y & 0  \\
     1+x & 0  \\
     0   & 1+\overline{x}\\
     0   & 1+\overline{y}
\end{array}\right) \, ,
\label{sigmatc}
\end{eqnarray}
and whose excitation map is given by 
\begin{align}
\epsilon= \left(\begin{array}{cccc}
     0 & 0 & 1+\overline{y} & 1+\overline{x}  \\
     1+x & 1+y & 0 & 0  \\
\end{array}\right) \, .
\end{align}
Considering the action of a local $Z$ operator on the first qubit at position $x$. The action of this local operator is represented by the following column
\begin{align}
\hat{o}=\left(\begin{array}{c}   0 \\
     0 \\
     x \\
     0
\end{array}\right) \, ,
\end{align} 
and the action of the excitation map on this operator gives
\begin{align}
    \epsilon \hat{o}= \left(\begin{array}{cc}
         0  \\
         x+x\overline{y} 
    \end{array} \right) \, .
\end{align}
This implies that the $Z$ stabilizers at positions $x$ and $x\overline{y}$ are excited due to the action of $\hat{o}$ on $\mathbb{Z}_2$ toric code. Similarly as for the polynomial description of Pauli operators, addition of excitation polynomials corresponds to fusion of excitations. 

The image of the excitation map $\text{im}(\epsilon)$ contains topologically trivial configurations of excitations. For cubic code stabilizer maps,  which take the following form  
\begin{align}
    \sigma=\left(\begin{array}{cc}
     f & 0  \\
     g & 0  \\
     0   & \overline{g}\\
     0   & \overline{f}
\end{array}\right).
\label{smap_fg}
\end{align}
The excitation map is given by 
\begin{align}
    \epsilon=  \left(\begin{array}{cccc}
     0 & 0 & \overline{f} & \overline{g}  \\
     g & f & 0 & 0  \\
\end{array}\right) \, 
\label{emap_fg}.
\end{align} 
For any CSS model, the $X$ and the $Z$ generator excitation sectors are decoupled, and in the case of the cubic codes they are related by a spatial inversion transformation. 
For cubic codes, the polynomials in the image of the excitation map for $Z$ generators belong to the ideal\footnote{An ideal $I$ of a polynomial ring $R$ contains elements $r_I$ such that $r_I r\in I$ for all $r\in R$ and all $r_I \in I$.} generated by $f(x,y,z)$ and $g(x,y,z)$ i.e. $p(x,y,z)f+q(x,y,z)g$ where $p(x,y,z)$ and $q(x,y,z)$ are arbitrary polynomials with $\mathbb{Z}_2$ coefficients. We refer to this ideal, $\av{f,g}$, as the stabilizer ideal~\cite{haah2013commuting}.

\subsubsection{Coarse-graining the stabilizer and excitation maps}
\label{cgmap}

Coarse-graining by a factor of 2 enlarges the unit cell by the same factor. 
Hence, after coarse-graining, the original translation variables $x^2$, $y^2$ and $z^2$ are transformed to $x'$, $y'$ and $z'$ on the new lattice. 
Suppose coarse-graining by a factor of 2 in the $x$-direction is performed, i.e. $x^2 \mapsto x'$. 
The coarse-grained unit cell has double the number of qubits and stabilizer generators. This coarse-graining transformation is implemented by the transformation of the original translation variables 
\begin{align}
x\mapsto \left(\begin{array}{cc}
     0 & x^\prime \\
     1 & 0
\end{array}\right) \, ,
&&
y\mapsto \left(
\begin{array}{cc}
     y & 0 \\
     0 & y
\end{array}\right) \, , 
&& z\mapsto \left(
\begin{array}{cc}
     z & 0 \\
     0 & z 
\end{array}\right) \, ,
\end{align}
where $x^\prime=x^2$ is a new translation variable. For example, this coarse-graining sends the stabilizer map of the 2D toric code from Eq.~\eqref{sigmatc} to 
\begin{eqnarray}
  \sigma^\prime=  \left(\begin{array}{cccc}
        1+y & 0 & & \\
         0 & 1+y & &\\
         1 & x^\prime  & &  \\
         1 & 1 & & \\
          & & 1 & 1\\
          & & \overline{x}^\prime & 1\\
          & & 1+\overline{y} & 0\\
          & & 0 & 1+\overline{y}
    \end{array}
    \right) \, .
    \label{CG_TC}
\end{eqnarray}
It is shown below that after the application of local CNOT gates, column operations and the removal of qubits in the trivial state, the original stabilizer map is recovered. 

The coarse-grained excitation map is defined in terms of the coarse-grained stabilizer map via $\epsilon^\prime={\sigma^\prime}^\dagger \lambda$.

\subsubsection{Coarse-graining factor and trivial charge configurations}

For the qubit-stabilizer models studied in this paper, we consider coarse-graining by factors of $2$, i.e. $c=2$. 
It was shown in Ref.~\onlinecite{haah2013,haah2014bifurcation} for cubic code 1 that the set of trivial charge configurations referred to as the annihilator of the charge module~\cite{haah2013}, denoted $\mathbf{A}$, shows self-reproducing behavior under coarse-graining by a factor of 2. 
The annihilator of the cubic code 1, $\mathbf{A}$, is given by the ideal $\av{1+x+y+z,1+xy+yz+xz}$. Here, for example, $1+x+y+z$ specifies a trivial charge configuration with the stabilizers excited at positions $1$, $x$, $y$ and $z$ under the action of a local operator. After coarse-graining, the annihilator $\mathbf{A}_c$ is given by $\av{1+x^\prime+y^\prime+z^\prime,1+x^\prime y^\prime+y^\prime z^\prime+x^\prime z^\prime}$ in terms of the coarse-grained variables $x^\prime=x^2, y^\prime=y^2, z^\prime = z^2$ and hence has the same form as $\mathbf{A}$. This suggests that cubic code 1 renormalizes into a model similar to itself. In fact, the annihilators of the charge modules for the two codes that are extracted after ER of cubic code 1 i.e. itself and cubic code 1B, are exactly the same. We find this self-reproducing behavior for all the cubic codes and the corresponding B models; the form of the annihilator after coarse-graining by a factor of 2, $\mathbf{A}_c$ retains the original form as $\mathbf{A}$ just like in the case of cubic code 1. Conversely, under coarse-graining by a factor of 3, the annihilator does not retain the original form. 
For any self-bifurcating fixed point model, it obviously follows that the models extracted after ER retain the same annihilator. 
For the bifurcating quotient fixed point models, which split into a copy of themselves and some B models under ER, having the same form of annihilator after coarse-graining implies that the annihilators of the B models contain the original annihilator $\mathbf{A}$.

\begin{table*}[t!]
\centering
\setlength{\tabcolsep}{8pt}
\renewcommand{\arraystretch}{1.45}
\begin{tabular}{c|ccccc}
Model & Particle mobilities & Type & ER & ER of the B models & QSS group
%the new model 
\tabularnewline
\hline 
3DTC & 3 & TQFT & 3DTC & & ${\mathbb{Z}_2}$ \tabularnewline
X-cube & 0,1,2 & foliated type-I & X-cube+$\text{STC}$ & $\text{STC}$+$\text{STC}$ & ${\mathbb{Z}}_2\oplus{\mathbb{Z}}_2^2$ \tabularnewline
CC$_1$& 0 & type-II & CC$_1$+$\text{CCB}_1$& $\text{CCB}_1$+$\text{CCB}_1$ & ${\mathbb{Z}}_2\oplus{\mathbb{Z}}_2$\tabularnewline
CC$_{2,3,4,7,8,10}$ & 0  & type-II & CC$_i$+CC$_i$ & & 0 \tabularnewline
CC$_{5,6,9}$ & 0,1 & fractal type-I & CC$_i$+CC$_i$ & & 0 \tabularnewline
CC$_{11\text{--}17}$\footnote{The ER results shown require an initial coarse-graining step for CC$_{13}$ and CC$_{17}$.}\textsuperscript{,\ref{xf}} %behavior up to some pre-coarse-graining } 
& 0,1,2 & fractal type-I & CC$_i$+$\text{CCB}_i$+$\text{STC}$ & $\text{CCB}_i$+$\text{CCB}_i$ & ${\mathbb{Z}}_2\oplus {\mathbb{Z}}_2$\footnote{This QSS was calculated for $\text{CC}_{11}$ in Sec.~\ref{CC11QSS}. The QSS of $\text{CC}_{12\text{--}17}$ can be calculated similarly.} \tabularnewline
CC$_{14}$\footnote{CC$_{14}$ shows both self-bifurcating and bifurcating ERG behavior\label{xf}} & 0,1,2 & fractal type-I & CC$_{14}$+CC$_{14}$ & & 0 \tabularnewline
First-order FSL  & model dependent & model dependent & FSL+FSL & & 0
\end{tabular}    \caption{Entanglement renormalization group (ERG) flows and quotient superselection sectors (QSS) of 3D stabilizer models.}
\label{results_ERG}
\end{table*}

\section{Examples of bifurcating entanglement renormalization}
\label{examples}

In this section, we find bifurcating ERG flows for explicit examples of fracton models. 
Some of the fracton models treated are found to be self-bifurcating fixed points while others are quotient bifurcating fixed points. 
Moreover, we find that some models may admit several qualitatively different ERG flows. We give such an example that is either a self-bifurcating fixed point or a quotient bifurcating fixed point, depending on the ER transformation chosen. 

The models we consider include several examples with known ERG flows: the 3D toric code (3DTC), a stack of 2D toric codes along an axis $\hat{i}$ (STC$_{\hat{i}}$), the X-cube model (X-cube), and Haah's cubic codes 1 ($\text{CC}_1$) and 1B ($\text{CCB}_1$). Beyond these known examples we also find the ERG flows of the 16 other CSS cubic codes~\cite{haah2011local,haah2013} {2-17} ($\text{CC}_{2\text{--}17}$) and the B codes thus produced (CCB$_{11\text{--}17}$), as well as all first-order fractal spin liquids~\cite{yoshida2013exotic} (FSL), a simple example of which is the Sierpinski fractal spin liquid~\cite{Chamon_quantum_glassiness} (SFSL). In the next section we put the ERG flows found for these examples into context by organizing them according to the type of 3D topological order each model displays~\cite{Dua_Classification_2019}. 

We have followed a simple heuristic to find ER transformations for the example models listed in table~\ref{results_ERG}. We outline this process in appendix~\ref{XC_ER}, and apply it to the  X-cube model as a demonstrative example.

\subsection{Self-bifurcating fixed points}

\begin{figure}[t]
\centering
\sidesubfloat[]{\includegraphics[scale=1.12]{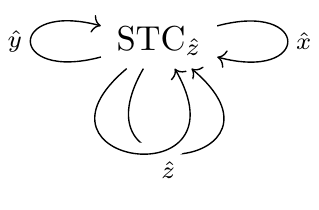}}\\
\sidesubfloat[]{\includegraphics[scale=1.12]{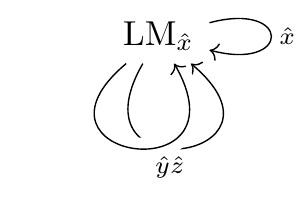}}\\
\sidesubfloat[]{\includegraphics[scale=1.12]{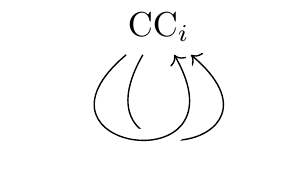}}\\
\sidesubfloat[]{\includegraphics[scale=1.12]{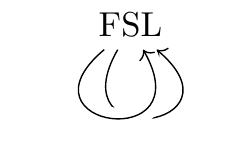}}\\
\caption{
	Self-bifurcating ERG flow diagrams. 
(a) ER of a stack of 2D toric codes along the $\hat{z}$ direction, parallel to the $xy$ plane, denoted by $\text{STC}_{\hat{z}}$. The arrows labeled $\hat{x}$ and $\hat{y}$ indicate directions in which the model is invariant under coarse-graining. Conversely, under coarse-graining along $\hat{z}$, $\text{STC}_{\hat{z}}$ self-bifurcates into two copies, which is indicated by a pair of arrows labeled $\hat{z}$. 
An arrow with no label indicates coarse-graining in all three lattice directions, in this case that produces the same result as coarse-graining in $\hat{z}$. 
(b) ER of a lineon model, denoted $\text{LM}_{\hat{x}}$, where the topological excitations of the model can all move along $\hat{x}$. Cubic codes 5, 6, 7 and 9 are examples of similar lineon models. 
(c) ER of self-bifurcating cubic codes CC$_{i}$ for $i$ in the range 2--10 or 14.
(d) ER of any first-order fractal spin liquid~\cite{yoshida2013exotic}.
} 
\label{SBERG}
\end{figure}

The coarse-graining step of the ER procedure can involve one, two or all three lattice directions. Self-bifurcating models can be divided into three categories according to the minimal number of directions one must coarse-grain for the model to bifurcate. 
Sorting the self-bifurcating models in this way is convenient when it comes to organizing them according to the type of topological order they exhibit, as discussed in next section. 

An important necessary condition for self-bifurcation with $c=b=2$ is that the number of encoded qubits under periodic boundary conditions doubles when the number of sites $L/a$ along each axis is doubled. This is because a self-bifurcating model with $c=b=2$ on a lattice with $2L/a$ sites along each axis splits into two copies of itself when coarse grained by a factor of 2, each of which lives on a decoupled lattice with $L/a$ sites. 

We discuss examples of stabilizer models that we have found to self-bifurcate under ER below. Examples of self-bifurcating models and their ERG flows are depicted in Fig.~\ref{SBERG} using our diagrammatic notation for bifurcating ERG fixed points.

\subsubsection{A stack of 2D toric codes}
The stabilizer map for a stack of 2D toric codes along the $\hat{z}$ direction, parallel to the $xy$ plane, appears identical to that in Eq.~\eqref{sigmatc}. 
After coarse-graining in the $\hat{x}$ direction the stabilizer map is identical to that in Eq.~\eqref{CG_TC}. 
Applying a row operation 
\begin{align*}
    \text{CNOT}(2,4,1+y)  = \left(\begin{array}{cccccccc}
        1 & 0 & 0 & 0 &  & & &\\
         0 & 1 & 0 & 1+y &  & & &\\
         0 & 0 & 1 & 0 &  & & & \\
          0 & 0 & 0 & 1 &  & & &\\
           & & & &  1 & 0 & 0 & 0\\
          & & & &  0 & 1 & 0 & 0\\
          & & & &  0 & 0 & 1 & 0\\
          & & & &  0 & 1+\bar{y} & 0 & 1
    \end{array}
    \right)
\end{align*}
to the coarse-grained stabilizer matrix from Eq.~\eqref{CG_TC} results in
\begin{align}
\text{CNOT}(2,4,1+y)
&\small{
    \left(\begin{array}{cccc}
        1+y & 0 & & \\
         0 & 1+y & &\\
         1 & x^\prime  & &  \\
         1 & 1 & & \\
          & & 1 & 1\\
          & & \bar{x}^\prime & 1\\
          & & 1+\bar{y} & 0\\
          & & 0 & 1+\bar{y}
    \end{array}
    \right)}\\
&=
    \left(\begin{array}{cccc}
        1+y & 0 & & \\
         1+y & 0 & &\\
         1 & x^\prime  & &  \\
         1 & 1 & & \\
          & & 1 & 1\\
          & & \bar{x}^\prime & 1\\
          & & 1+\bar{y} & 0\\
          & & \bar{x}^\prime(1+\bar{y}) & 0
    \end{array}
    \right)
    \, .
\end{align}
After performing additional row operations CNOT$(3,4,x')$, CNOT$(1,2,x')$, CNOT$(1,3,1+y)$ and column operations Col$(1,2,1)$ and Col$(4,3,1)$, the model becomes
\begin{align}
  \sigma^\prime=  \left(\begin{array}{cccc}
        0 & 0 & & \\
         1+y & 0 & &\\
         1+x^\prime & 0  & &  \\
         0 & 1 & & \\
          & & 1 & 0\\
          & & 0 & 1+\overline{x}^\prime\\
          & & 0 & 1+\overline{y}\\
          & & 0 & 0
    \end{array}
    \right)
\end{align}
which is nothing but a stack of 2D toric codes and two qubits per site in a product state. Hence, the stack of toric codes along the $\hat{z}$ axis is a fixed point under ER in $x$. 
The stack of 2D toric codes stabilizer map in Eq.~\eqref{sigmatc} has $x\leftrightarrow y$ symmetry up to relabeling of the qubits. Hence, it is also a fixed point under ER in $y$. 
Furthermore, for the stack of 2D toric codes along the $\hat{z}$ direction,  we can trivially coarse-grain the stabilizer map along $\hat{z}$ by taking 
\begin{align}
x\mapsto \left(
\begin{array}{cc}
     x & 0 \\
     0 & x
\end{array}\right) \, , && y\mapsto \left(
\begin{array}{cc}
     y & 0 \\
     0 & y
\end{array}\right) \, ,
\end{align}
as $z$ does not enter  the stabilizer map. 
This simply results in two decoupled copies of the stack of 2D toric codes Hamiltonian. 
Hence, a stack of 2D toric codes is self-bifurcating under ER. 
This provides an example of a model that self-bifurcates after ER in only one direction and is a fixed point under ER along either of the orthogonal directions.

\subsubsection{Yoshida's fractal spin liquids} 
We now move on to show that a far more interesting class of examples, the first-order fractal spin liquids of Yoshida~\cite{yoshida2013exotic}, are all self-bifurcating under ER. 
The general form of the stabilizer map for these models is 
\begin{equation}
    \begin{pmatrix}
    1+f(x) y & 0\\
    1+g(x) z & 0\\
    0 & 1+g(\overline{x})\overline{z}\\
    0 & 1+f(\overline{x})\overline{y}
    \end{pmatrix},
\end{equation}
where $f$ and $g$ are polynomials in the single translation variable $x$. Such a model is type-II if and only if $f$ and $g$ are not algebraically related~\cite{yoshida2013exotic}. 
This class of models also contains fractal type-I lineon models for ${f=1}$,  ${g\neq 0,1}$, stacks of 2D toric code for ${f=g=1}$, stacks of 2D fractal subsystem symmetry-protected models~\cite{devakul2018fractal,devakul2018universal,Stephen2018computationally,Devakul2018,Daniel2019} for $f=0$, ${g\neq 0,1}$, and decoupled 1D cluster states for $f=0$, $g=1$. We remark that $f$ and $g$ can be exchanged in the above models up to a redefinition of the lattice and an on-site qubit swap. 

Under coarse-graining along $\hat{y}$: 
\begin{align}
x\mapsto \left(
\begin{array}{cc}
     x & 0 \\
     0 & x
\end{array}\right) \, , && 
y\mapsto \left(\begin{array}{cc}
     0 & y^\prime \\
     1 & 0
\end{array}\right) 
&& z\mapsto \left(
\begin{array}{cc}
     z & 0 \\
     0 & z
\end{array}\right) \, ,
\end{align}
where $y^\prime=y^2$ is the translation variable on the coarse-grained lattice, the stabilizer map becomes
\begin{equation}
\begin{pmatrix}
    1 & fy^\prime &  & \\
    f & 1 &  & \\
    1+gz & 0 &  & \\
    0 & 1+gz &  & \\
     &  & 1+\overline{g}\, \overline{z} & 0\\
     &  & 0 & 1+\overline{g}\, \overline{z}\\
      &  & 1 & f\\
     &  &  \overline{f}\,\overline{y^\prime} & 1
    \end{pmatrix}
    \, .
\end{equation}
Applying row operations: CNOT$(2,1,f)$, CNOT$(3,1,1+gz)$, CNOT${(3,4,fy)}$ and column operations Col$(2,1,f(x)y)$, Col$(4,3,f(\overline{x}))$ leads to the stabilizer map
\begin{equation}
\begin{pmatrix}
    1 & 0 &  & \\
    0 & 1+f^2(x)y^\prime &  & \\
    0 & 0 &  & \\
    0 & 1+g(x)z &  & \\
     &  & 0 & 0\\
     &  & 0 & 1+g(\overline{x})\overline{z}\\
      &  & 1 & 0\\
     &  &  0 & 1+f^2(\overline{x})\overline{y^\prime}
    \end{pmatrix}.
\end{equation}
The first and third qubits are disentangled in the above stabilizer map and hence can be removed\footnote{Up until this step, the same ER process as shown works if ${1+g(x)z}$ is generalized to $1+g_1(x)z+g_2(x)z^2+\cdots$, along with a similar generalization for $1+g(\oline{x})\oline{z}$.}. 
The resulting stabilizer map is 
\begin{equation}
\label{eq:fsler1}
\begin{pmatrix}
    1+f^2y^\prime & 0\\
    1+gz & 0\\
    0 & 1+\oline{g}\oline{z}\\
    0 & 1+\oline{f}^2\oline{y^\prime}
    \end{pmatrix}\, ,
\end{equation}
which is also a first-order FSL, where $f(x)$ has been replaced by $f^2(x)$. We now notice that due to the symmetry between $f$ and $g$ in a first-order FSL one can apply essentially the same ER transformation, this time coarse-graining  $z$, so that $g$ is replaced by $g^2$. For a polynomial over $\mathbb{F}_2$, a useful property is that $f^2(x)\equiv f(x^2)$. This leaves only functions of $x^2$ in the coarse-grained stabilizer map, so we coarse-grain again, along $x$ this time, sending
\begin{align}
x^2 \mapsto 
\begin{pmatrix}x^\prime & 0 \\ 0 & x^\prime\end{pmatrix} \, .
\end{align} 
This results in the stabilizer map splitting into two copies of the original model, and hence all qubit first-order FSL models are self-bifurcating fixed points with $b=2$. This includes the stack of 2D toric codes for $f=g=1$ and the trivial model for $f=g=0$. Furthermore, if $f =0,1$ then $f^2=f$ and the model is a conventional fixed point under ER that coarse-grains along $y$ only, while being a self-bifurcating fixed point under ER that coarse-grains both $x$ and $z$, and similarly if $f$ and $g$ are swapped. 

A particular example, the SFSL model, is obtained when $f(x)=1+x$, $g(x)=1$, this model has a fractal logical operator in the $xy$ plane and a lineon operator along $\hat{z}$~\cite{yoshida2013exotic}.  Our results indicate that the model is invariant under ER along the direction of the string operator and self-bifurcates under ER in the  plane of the fractal operator.

The above ER transformations were essentially based on the observation that the 2D first-order fractal subsystem symmetry breaking (classical) spin models, from which the first-order FSLs are built, are self-bifurcating under ER. 
This can easily be seen by following the ER transformation of the FSL before Eq.~\eqref{eq:fsler1} with the second qubit and generator, as well as the $z$ coordinate, dropped from the stabilizer map. 
This connection generalizes straightforwardly to provide ER transformations for 3D FSLs based upon self-bifurcating 2D fractal subsystem symmetry breaking spin models that may not be first-order. 
We remark that higher-order FSLs need not be self-bifurcating, as cubic code 1, which is not self-bifurcating, is equivalent to a second-order FSL~\cite{yoshida2013exotic}. 
FSL forms for this, and some other cubic codes are presented in appendix~\ref{CCFSLs}. The explicit mapping transformations can be found in the supplementary {\small{MATHEMATICA}} file SMERG.nb. 
An interesting open problem is the classification and characterization of higher-order self-bifurcating 2D fractal subsystem symmetry breaking spin models and the 3D FSLs built from them~\cite{ShirleyERG}. 

We remark that a form of real-space RG was considered for the FSLs in Ref.~\onlinecite{yoshida2013exotic}, however it does not conform to the strict definition of ER used here, where the phase of matter cannot change. Instead, degrees of freedom that were not fully disentangled were projected out, which is capable of changing the phase by projecting out an arbitrary nontrivial decoupled model. Our ER results are consistent with the RG results in Ref.~\onlinecite{yoshida2013exotic}. 

For the details of the ER procedures for the other examples discussed below, we refer the reader to the {\small{MATHEMATICA}} file SMERG.nb in the Supplementary Material.  

\subsubsection{Cubic codes 5, 6 and 9}
These cubic codes are fixed points under ER along one of the lattice directions, while they self-bifurcate under ER along the orthogonal plane, see Fig~\ref{SBERG}~(b).  
This is consistent with these models being fractal type-I lineon models~\cite{Dua_Classification_2019}. 

\subsubsection{Cubic codes 2--4, 7, 8 and 10}
These cubic codes are not fixed points under ER along any single lattice directions, and self-bifurcate only after doing ER along all three lattice directions together, see Fig~\ref{SBERG}~(c).  
This is consistent with them being either fractal type-I, or type-II, fracton models~\cite{Dua_Classification_2019}.  

% \begin{figure}[t]
%     \centering
% \includeTikzrm{CC1_RG}{\begin{tikzcd}[ampersand replacement=\&]
% \text{CC}_{1}\arrow[out=230, in=310,loop] \arrow[r] \& \text{CCB}_1 \arrow[out=240, in=320, loop] \arrow[out=220, in=300, loop,crossing over] \\
% \end{tikzcd}}
% \caption{Bifurcating ERG flow of Haah's cubic code 1. }
%     \label{CC1RG}
% \end{figure}

\begin{figure}[t]
    \centering
\includegraphics[scale=1.12]{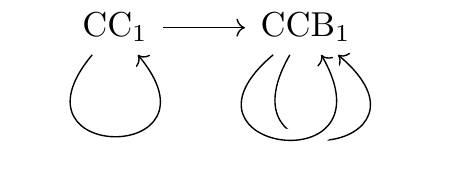}
\caption{Bifurcating ERG flow of Haah's cubic code 1. }
\label{CC1RG}
\end{figure}

\subsection{ Quotient bifurcating fixed points}
A quotient bifurcating fixed point is defined by its branching under ER into a copy of itself and an inequivalent self-bifurcating model, which may consist of several decoupled self-bifurcating models. 
Verifying the inequivalence of models can be quite subtle, as it requires one to consider arbitrary locality-preserving unitaries, including coarse graining and modular transformations. 
Ref.~\onlinecite{haah2014bifurcation} introduced an approach to proving the inequivalence of models based upon the behavior of their ground space degeneracies. 
In particular, the following sufficient condition was proposed to confirm that a model could not be a self-bifurcating fixed point
\begin{align}
k(c L)=\alpha k(L) + \beta
\, ,
\label{eq:qbcond}
\end{align}
for integers $\alpha >1,$ $\beta >0$. Where $k(L)=\log_2\,\text{GSD}(L)$ is the number of encoded qubits in the ground space on an $L\times L \times L$ system with periodic boundary conditions. 
This is in contrast with the necessary condition that $\beta=0$ in Eq.~\eqref{eq:qbcond}
for any self-bifurcating model. 

\subsubsection{Cubic code 1}
It was shown in Ref.~\onlinecite{haah2014bifurcation} that CC$_1$ bifurcates into a copy of itself and an inequivalent model CCB$_1$ under ER and hence is a quotient fixed point, see Fig.~\ref{CC1RG}. 
To prove this inequivalence the scaling of the number of encoded qubits in the ground space, ${k(2L) = 2k(L)+2}$, was used. This scaling can be directly inferred from the following formula for CC$_1$~\cite{haah2013,haah2014bifurcation,haah2013commuting} 
\begin{align}
 k(L) = 2^{l + 2}\deg_x\Big(\gcd\big((&x + 1 )^{L'} + 1, (\zeta_3 x + 1)^{L'} + 1, \nonumber \\
&({\zeta_3}^2 x + 1)^{L'} + 1\big)\Big)- 2
\, .
\end{align}
Thus, any model equivalent to CC$_1$ cannot self-bifurcate. Since $\text{CCB}_1$ was demonstrated to be a self-bifurcating model, $\text{CC}_1$ and CCB$_1$ cannot be in the same phase. 

\subsubsection{Cubic codes 11-17} 
We have found that these cubic codes are quotient fixed points under ER. CC$_i$, for $i=11,\dots,17$, branches into a copy of CC$_i$, a stack of 2D toric codes along one direction and a self-bifurcating model CCB$_{i}$, see Fig.~\ref{BRG_CC11}. 
This is consistent with these models being fractal type-I topological orders that support planons. See Appendix~D of Ref.~\onlinecite{Dua_Classification_2019} for a description of these planons. 
For CC$_{13}$ and CC$_{17}$ the quotient fixed point behavior occurs after an initial step of coarse-graining, see Fig.~\ref{CC_13_14_17}. 

\begin{figure}[t]
\centering
\includegraphics[scale=1.12]{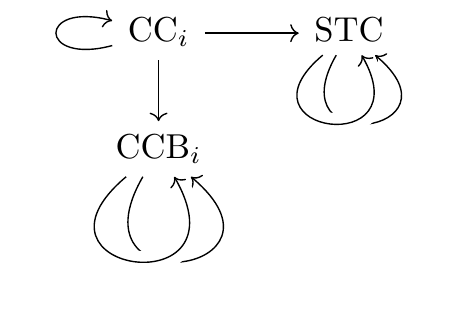}
\caption{ Bifurcating ERG flow of CC$_i$ for $i=11\text{--}17$, which support planon excitations. For CC$_{13}$ and CC$_{17}$ in particular, the ERG flow depicted requires an initial coarse-graining step shown in Fig.~\ref{CC_13_14_17}. CC$_{14}$ also admits a self-bifurcating ERG flow, see Fig.~\ref{SBERG}~(c).}
\label{BRG_CC11}
\end{figure}

\begin{figure}[t]
\centering
\sidesubfloat[]{\includegraphics[scale=1.12]{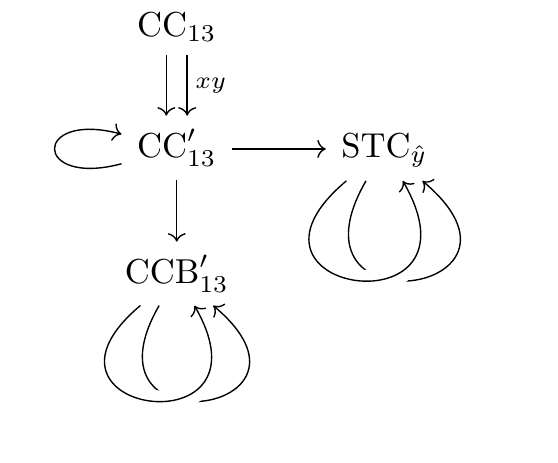}}\\
\sidesubfloat[]{\includegraphics[scale=1.12]{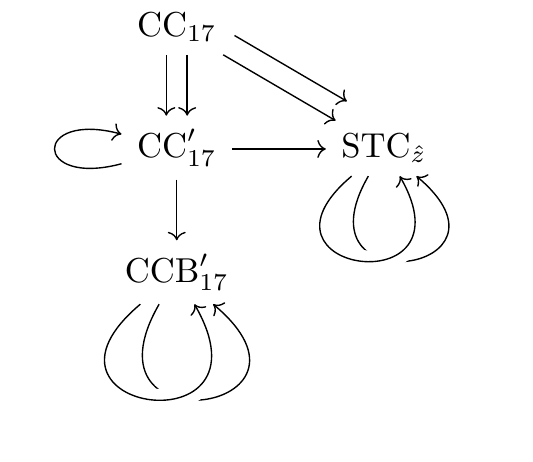}}\\
\sidesubfloat[]{\includegraphics[scale=1.12]{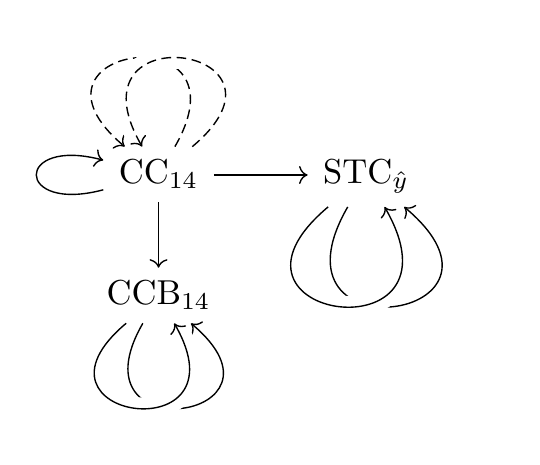}}\\
\caption{Examples of bifurcating ERG flows. a) ER of CC$_{13}$. The $xy$ label on the arrows indicates the directions under which the pre-coarse-graining is done. After an initial coarse-graining step, CC$_{13}$ splits into two copies of CC$_{13}^\prime$. CC$_{13}^\prime$ is a quotient fixed point of the type shown in Fig.~\ref{BRG_CC11}~(b). b) CC$_{17}$ splits into two copies of CC$_{17}^\prime$ and two stacks of 2D toric codes. CC$_{17}^\prime$ is another quotient fixed point of the type shown in Fig.~\ref{BRG_CC11}~(b). c) CC$_{14}$ appears to be a quotient fixed point of the type shown in Fig.~\ref{BRG_CC11}~(b) under one ERG flow, depicted with solid lines, but is a self-bifurcating fixed point under another ERG flow, depicted with dotted lines.} 
\label{CC_13_14_17}
\end{figure}

\subsection{Quotient bifurcating and self-bifurcating fixed point behavior for cubic code 14} 
\label{ex:cc14}

We have found two distinct ERG flows for CC$_{14}$: 
under one ER procedure CC$_{14}$ is self-bifurcating with $b=2$. 
Under another ER procedure  CC$_{14}$ appears to be a quotient fixed point as it bifurcates into a copy of itself, a stack of 2D toric codes and another self-bifurcating model CCB$_{14}$, see Fig.~\ref{CC_13_14_17}. 
Due to the absence of any correction factor in Eq.~\eqref{eq:qbcond} for CC$_{14}$, as it is a self-bifurcating fixed point, the standard method to argue for inequivalence to CCB$_{14}$ and a stack of 2D toric codes does not apply. On the other hand, we have been unable to find a direct equivalence between CC$_{14}$ and CCB$_{14}$ plus a stack of 2D toric codes. It is clear from the ERG flows that two copies of CC$_{14}$ is phase equivalent to CC$_{14}$, CCB$_{14}$ and a stack of toric codes. Thus in the case that the ERG flows are inequivalent, this would provide an interesting example of a phase equivalence that is catalyzed by the addition of a copy of CC$_{14}$. In any case, all of the aforementioned models are in the same trivial bifurcated equivalence class. 

This example raises the interesting question of whether all ERG flows for which a model is a bifurcating fixed point are equivalent. We remark that this is trivially true for conventional fixed points but requires a careful definition of equivalence for more general bifurcating ERG flows. 

\section{Entanglement renormalization in different types of 3D topological phases} 
\label{classes_TO_ER}

In the recent fracton literature~\cite{Dua_Classification_2019} 3D topological stabilizer models have been coarsely organized into four qualitatively distinct classes: TQFT, foliated and fractal type-I, and type-II topological orders. In this section we describe the salient features of these different classes, which are determined by the mobilities of their topological quasiparticles, and discuss the influence they have on possible ERG flows. 

\subsection{TQFT topological order}

TQFT topological order is characterized by a constant topological ground space degeneracy as the system size increases and deformable logical operators. 
In two dimensions, there is essentially only one type of translation-invariant topological stabilizer model, the 2D toric code which is described by a TQFT at low energy. 
This is due to a structure theorem~\cite{Haah2018,haah2016algebraic} stating that any translation-invariant topological stabilizer model is  equivalent under locality-preserving unitary to copies of the toric code and some disentangled trivial qubits. 

Hence all TQFT stabilizer models in 2D are fixed points under ER, equivalent to a number of copies of 2D toric code. We expect similar behavior in 3D, that all TQFT stabilizer models are fixed points under ER equivalent to a number of copies of 3D toric code and hence are described by a TQFT at low energies, although this remains to be shown.

In 3D the logical operator pairs of the toric code are composed of deformable string and membrane operators. 
More general 3D translation-invariant topological stabilizer models that have a constant ground space degeneracy as system size increases, and so are TQFT topological orders, have been shown to resemble this behavior as their logical operators also come in string-membrane pairs~\cite{yoshida2011classification}. We conjecture that a stronger structure theorem holds for such models in three dimensions: i.e. a TQFT stabilizer model in 3D is locality-preserving unitary equivalent to copies of the 3D toric code (possibly with fermionic point particle~\cite{Levin_wen_fermion,walker2012}), and disentangled trivial qubits. 

\subsection{Type-I fracton topological order}
Type-I fracton topological order is characterized by a sub-extensive ground space degeneracy and rigid string operators, which correspond to excitations with sub-dimensional mobilities. This type of topological order is divided into two broad categories which we treat separately below. 

\subsubsection{Foliated type-I topological order}
Foliated topological order is characterized via a foliation structure~\cite{shirley2018Foliated,shirley2017fracton,Slagle2018foliated}. Such models can be grown by adding layers of a 2D topological order, such as the 2D toric code, according to the foliation structure. 
The most studied example of this type is the X-cube model which can be grown by adding layers of 2D toric code as shown in Fig.~\ref{foliated_TO}. More formally, two Hamiltonians are foliated equivalent~\cite{shirley2018universal} if they are connected by adiabatic evolution and addition of layers of two dimensional gapped Hamiltonians. 
When translation invariance is enforced, foliated equivalence is the same as bifurcated equivalence with respect to stacks of 2D topological orders. 
Foliated stabilizer models can have topological  quasiparticles with a hierarchy of mobilities such as fractons, lineons and planons. Given their foliation structure in terms of 2D toric codes, they must always support planons. 
In fact, the X-cube model has planons in all the three lattice directions due to its foliation structure that allows layering with 2D toric code in the three lattice directions. 
Due to the underlying foliation structure of this class of models, the ground space degeneracy scales exponentially with the size of the system. 

\begin{figure}
\centering
\includegraphics[scale=0.19]{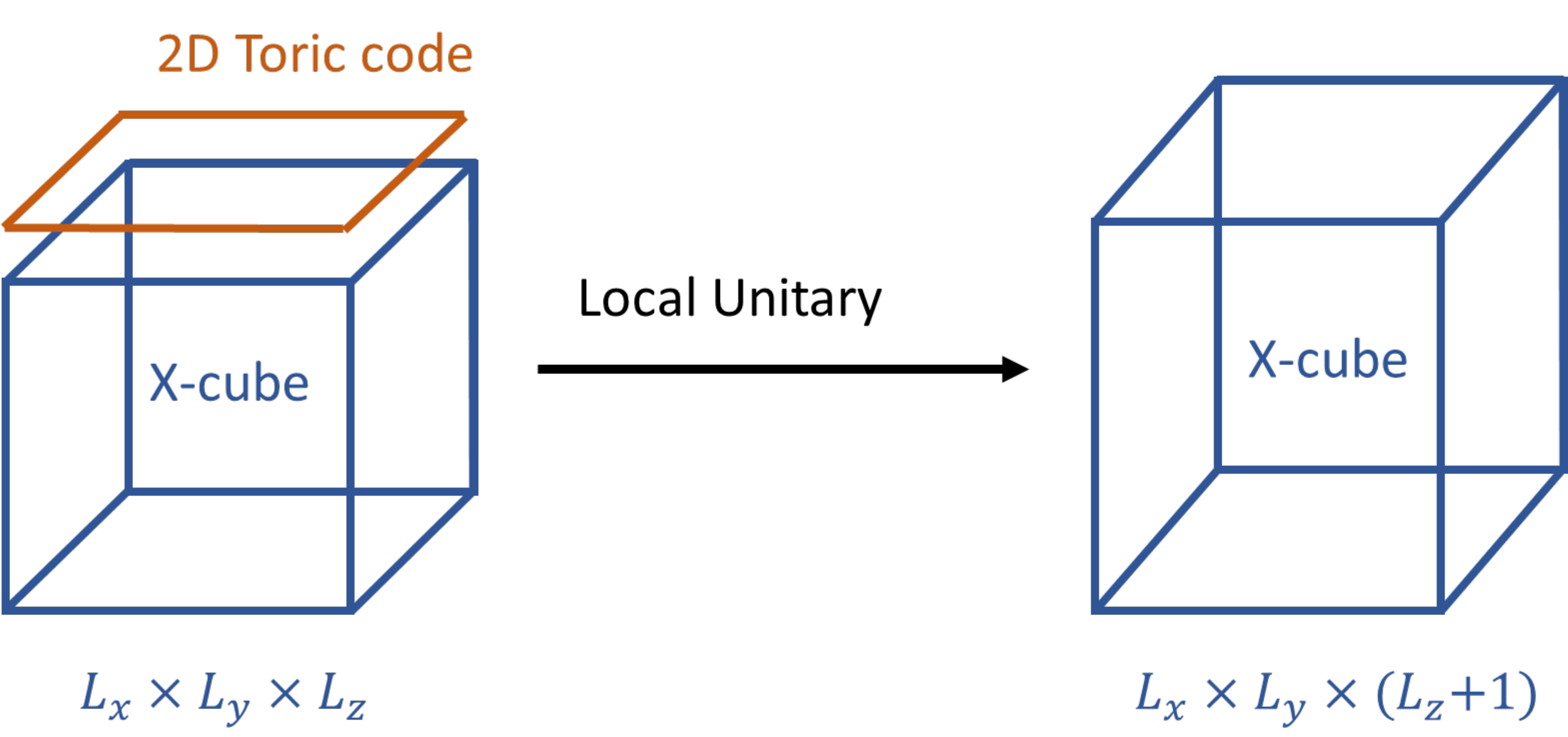}\\
\caption{Foliated topological order in the X-cube model.}
\label{foliated_TO}
\end{figure}

\begin{figure}
\centering
\includegraphics[scale=0.19]{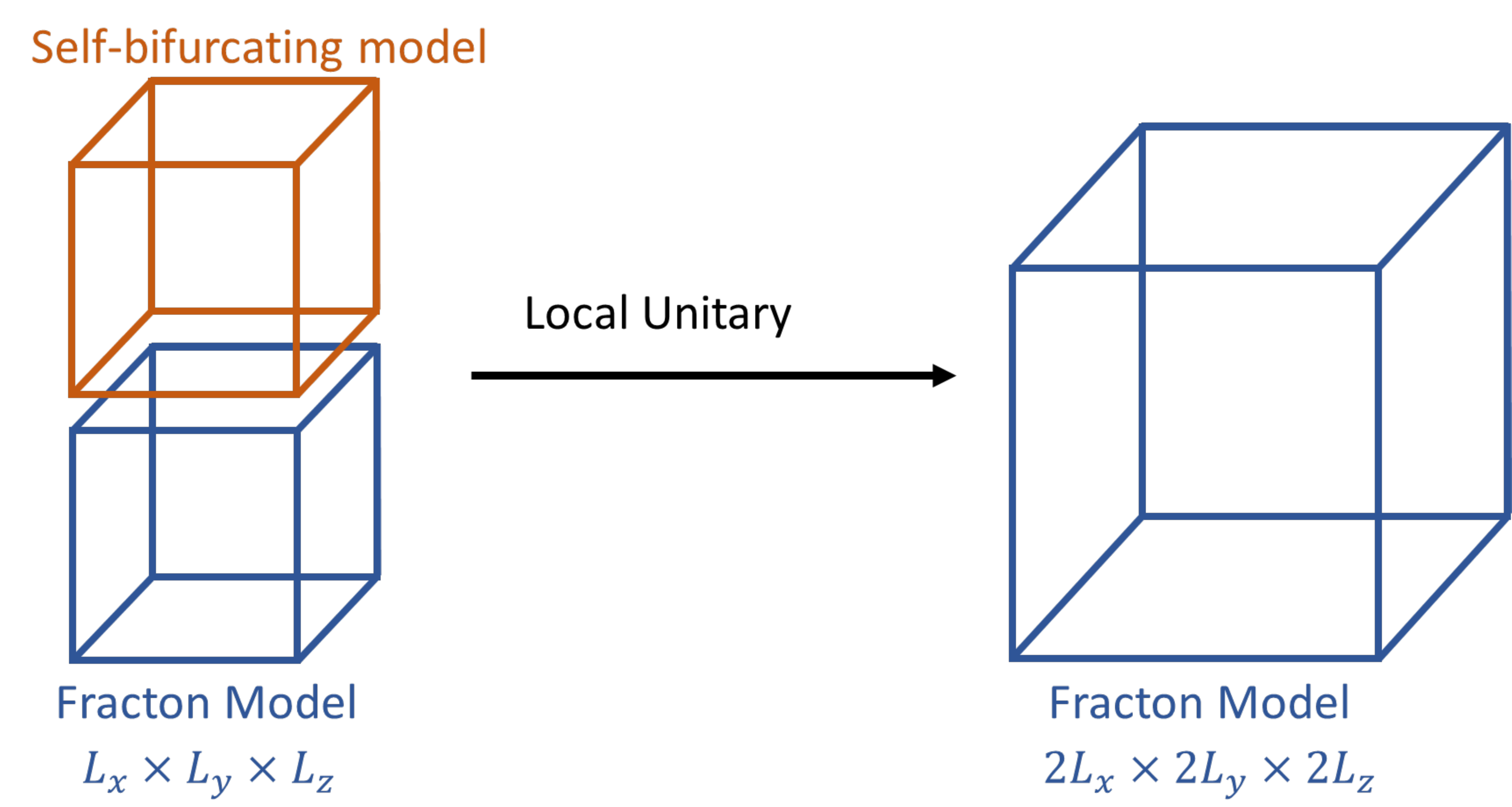}\\
\caption{Entanglement structure of more general fracton models.}
\label{foliated_TO}
\end{figure}

The canonical self-bifurcating example within this class of models is the stack of 2D toric codes. 
This serves as the B model for X-cube, which is a quotient fixed point that bifurcates into a copy of itself and a stack of 2D toric codes along each axis. 
This is a consequence of the fact that X-cube supports planons in all the three lattice planes. 
The ER process for X-cube in the polynomial language is presented in appendix~\ref{XC_ER}. 

Since all 2D topological stabilizer model are essentially copies of 2D toric code, we expect that stacks of 2D toric code are the only self-bifurcating stabilizer models in the foliated class. 
Furthermore, we expect that all nontrivial foliated stabilizer models flow to quotient fixed points with stacks of 2D toric code serving as the B models. 
An interesting open question is whether all such quotient fixed point foliated stabilizer models  are foliated equivalent to a suitably generalized X-cube model which is allowed to take on different foliation structures~\cite{shirley2018Fractional,shirley2018Foliated,shirley2018universal,Wang2019}. 

\subsubsection{Fractal type-I topological order}
Fractal type-I topological order captures type-I models that support fractal logical operators and symmetries, due to which no foliation structure in terms of 2D Hamiltonians is known or likely possible. In fact, these models need not support planons at all. 
Due to the underlying fractal symmetry of this class of models, the ground space degeneracy fluctuates with the system size. 
The simplest example of this type of model is the Sierpinski fractal spin liquid model~\cite{yoshida2013exotic, Chamon_quantum_glassiness} which supports two dimensional fractal operators in the $xy$ plane along with rigid string operators in the $\hat{z}$ direction as all particles in the model are lineons. 
This model was mentioned above as a special case of Yoshida's first-order FSLs which were shown to be self-bifurcating. 
As expected, this model is invariant under ER along the $\hat{z}$ direction alone but self-bifurcates under ER in the $xy$ plane or for all the three directions at once. 

There are also fractal type-I models amongst the cubic codes, some of which support fractons along with lineons and planons. 
We discuss the ER of two such examples below. 
The first is CC$_6$ which is self-bifurcating and the other one is CC$_{11}$ which is a quotient bifurcating fixed point: 

% There are also fractal type-I models which support 3D fractal symmetries and in that sense, are similar to type-II models. These models can still support lineons and planons along with fractons. 
% We discuss the ER of two such examples below. 
% The first is CC$_6$ which is self-bifurcating~\footnote{The appearance of a 3D fractal symmetry depends on how a model is written down, as it is shown in appendix~\ref{CCFSLs} that CC$_6$ can be mapped to a first order FSL, whose symmetries are supported on two sets of orthogonal planes.} and the other one is CC$_{11}$ which is a quotient bifurcating fixed point. 

CC$_6$ is a self-bifurcating fixed point under ER that coarse-grains all three lattice directions together. One can perform a modular transformation on the stabilizer map of CC$_6$ to bring the lineon string operator along a lattice direction. 
The transformed model is invariant under ER that coarse grains in the direction of the string operator, while it self-bifurcates under ER that coarse grains the directions orthogonal to it. The number of encoded qubits of CC$_6$ is given by 
\begin{align*}
k(L)= 2^{l + 1}\deg_z\left(\gcd((z + 1)^{L'} + 1, (z^2 + z + 1)^{L'} + 1)\right).
\end{align*}
This obeys
\begin{align*}
      \boxed{  k(2^r L) = 2^r k(L)}
\end{align*}
which implies that the number of encoded qubits always doubles upon doubling the system size. This is indeed consistent with the self-bifurcating behavior. 

CC$_{11}$ is an important fractal type-I example as it supports topological quasiparticles of sub-dimensional mobilities 0, 1 and 2. 
Performing ER that coarse-grains all the three directions together causes CC$_{11}$ to bifurcate into a copy of itself, a lineon model CCB$_{11}$, and a stack of 2D toric codes along $\hat{x}$. 
We conjecture that this bifurcating behavior is directly related to the mobilities of topological  quasiparticles in CC$_{11}$. 
The presence of planons in the $yz$ planes leads to the extraction of a stack of toric codes parallel to $yz$ planes. 
Similarly, a subset of lineons from the original model CC$_{11}$ flow into the lineon model $\text{CCB}_{11}$ under ER. 
In fact, for each of the cubic codes 11--17, except cubic code 14, there is an  ER procedure that leads to a self-bifurcating lineon B model. 
For CC$_{14}$, we find a self-bifurcating fracton B model CCB$_{14}$ which supports a composite lineon as mentioned in Table~\ref{poly_codes} of the appendix. 
For $\text{CCB}_{11}$, the scaling of the number of encoded qubits with the system size, derived in the appendix, is found to be consistent with the self-bifurcating behavior. 

On the other hand, we have found that the number of encoded qubits for CC$_{11}$ is given by
\begin{align*}
\begin{cases}
2L,  & 3 \nmid L \\
2L - 4  + 2^{l + 2} \deg_z\Big(\gcd((z + 1)^{L'} + 1, &
\\
\phantom{2L - 4  + 2^{l + 2} \deg_z ( \text{gcd}} (\zeta_3 z + 1)^{L'} + 1)\Big), & 3 \mid L \, .
\end{cases}
\end{align*}
Using this result, we notice that there exists an infinite family of system sizes $L$ for which the scaling of number of encoded qubits has a correction factor: for $3|L$, we have
\begin{align}
 \boxed{
 k(2L)= 2k(L)+4 } \, .
\end{align}
The argument from Ref.~\onlinecite{haah2014bifurcation} can then be used to show that the correction factor of $4$ cannot be eliminated for any choice of coarse-graining of the original model. Hence, CC$_{11}$ cannot self-bifurcate and is therefore inequivalent to CCB$_{11}$. 

Numerical results for the number of encoded qubits in the ground space of CC$_{12\text{--}17}$ indicate that such an inconsistency with self-bifurcating behavior may hold for all codes except CC$_{14}$ for which the numerics are consistent with $k(2L)=2k(L)$. 
In fact, as mentioned before, we show explicitly that there are two ERG flows for CC$_{14}$, one in which it self-bifurcates and the other in which it splits into a copy of itself, a B model and a stack of 2D toric codes. 
The ER of CC$_{14}$ is shown in Fig.~\ref{CC_13_14_17}~(c). 

\subsection{Type-II topological order}
Type-II topological order is defined by the absence of any logical string operators and characterized by a sub-extensive ground space degeneracy that fluctuates with the system size. 
For the models in this class, none of the elementary topological excitations or their nontrivial composites are mobile. In the previous section we discussed the canonical example from this class, cubic code 1, which is known to be a quotient fixed point. 
CC$_{1}$ supports a fractal operator that moves excitations apart in three dimensions. 
We show in the {\small{MATHEMATICA}} file SMERG.nb that the other type-II cubic code models\footnote{CC$_{2\text{--}4}$ were rigorously  proven to be type-II in Ref.~\cite{haah2011local} while CC$_7$, CC$_8$ and CC$_{10}$ were found to be type-II according to the methods in Ref.~\cite{Dua_Classification_2019}.}: CC$_{2\text{--}4}$, CC$_7$, CC$_8$, CC$_{10}$ and a type-II first-order quantum FSL~\cite{yoshida2013exotic} are self-bifurcating fixed points.

\section{Quotient Superselection sectors} 
\label{sec:quotient_sectors} 

In this section we explore the flow of excitations under ER. We find the set of quotient superselection sectors for several quotient fixed point models by calculating their fixed point excitations under quotient ERG flow. 

To find fixed point excitations under a quotient ERG flow it is important to understand how excitations split during each step of ER. 
The local unitaries performed as part of the ER process alter the form of stabilizer generators but do not cause any splitting. 
It is the redefinition of the stabilizer generators allowed by the equivalence relation $\equiv$ during ER that determines how an excitation splits. 
This is illustrated in Fig.~\ref{quotient_sectors}. 
In the polynomial description of ER, this redefinition is captured by column operations on the stabilizer map. 
A representation $ExS$ of the excitation splitting map can be found by taking the transpose of the matrix representation of the composition of all column operations involved in an ER transformation. 
For CSS models $ExS$ is block-diagonal and can be decomposed into separate maps $ExS_X$ and $ExS_Z$ for the $X$ and $Z$ sectors, respectively. 

Quotient superselection sectors are defined to be equivalence classes of superselection sectors modulo any sectors that can flow into a self-bifurcating model. For stabilizer models the QSS form an abelian group under fusion. 
Nontrivial QSS only arise in nontrivial quotient fixed point models.  While self-bifurcating fixed point models may support highly nontrivial superselection sectors, these sectors all collapse into the trivial QSS. 
A nontrivial QSS must maintain support on a quotient fixed point model along its ERG flow. Conversely, any excitation that flows fully into self-bifurcating B models lies in the trivial QSS. Hence, representatives of potentially nontrivial QSS are given by fixed point excitations under quotient ER which mods out any B models.

\begin{figure}[t]
\centering
\sidesubfloat[]{\includegraphics[scale=1.0]{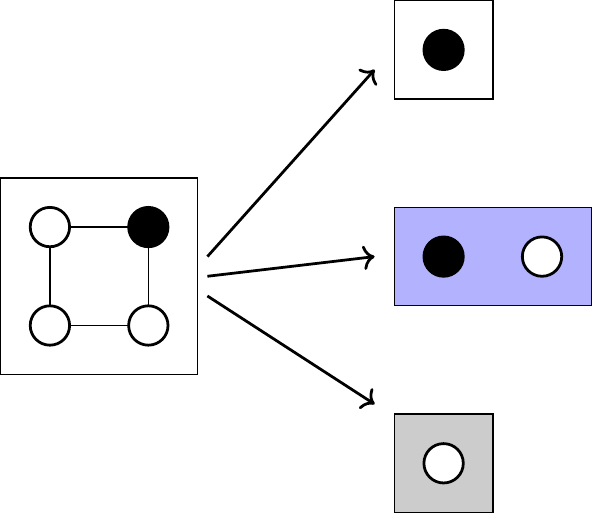}}\\
\vspace{6mm}
\sidesubfloat[]{\includegraphics[scale=1.0]{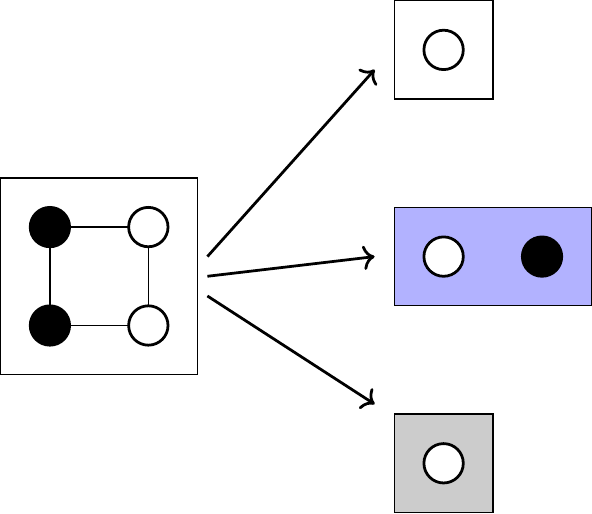}}\\
\caption{An illustration of the flow of excitations under ER for a quotient fixed point model. On the left hand side, the coarse-grained unit cell of a 2D quotient fixed point model with one type of excitation per site is depicted. The presence of an excitation is indicated by a filled circle. Under ER, this excitation can split into multiple excitations that may be supported on a coarse grained copy of the original model (white box) and the B models (blue and grey boxes). The excitation in (a) is a fixed point under quotient ER and hence represents a QSS. On the other hand the excitation in (b) corresponds to a trivial QSS as it flows to the trivial fixed point under quotient ER. } 
\label{quotient_sectors}
\end{figure}

Each row of the map $ExS$  corresponds to a different type of stabilizer generator in the coarse-grained model. 
The entries in a row are the coefficients for expressing a new stabilizer generator as a linear combination of the original stabilizer generators within the coarse-grained unit cell, 
\begin{eqnarray}
 \begin{array}{c}
\drawgenerator{xz}{xyz}{x}{xy}{y}{1}{yz}{z}
\end{array}
\label{pos_unit_cell}.
\end{eqnarray}
Here, $x$, $y$ and $z$ are the translation variables of the original lattice. For the choice of coarse-graining map described in Sec.~\ref{cgmap}, the original stabilizer generators within the coarse-grained unit cell are ordered as follows 
\begin{equation}
 1-z-y-yz-x-xz-xy-xyz \, .
 \label{order_cg_cell}
\end{equation}
According to the above ordering, the position of each stabilizer generator in the coarse-grained unit cell can be mapped to unit vectors $e_i$. This implies, for example, that the generator originally at position $x$ becomes the $5^\text{th}$ generator in the coarse-grained unit cell, corresponding to the unit vector $e_5=\left(\begin{array}{cccccccc}
0 & 0 & 0 & 0 & 1 & 0 & 0 & 0\end{array}\right)$, and similarly the generator at position $xyz$ becomes the last, corresponding to $e_8=\left(\begin{array}{cccccccc}
0 & 0 & 0 & 0 & 0 & 0 & 0 & 1\end{array}\right)$. 

As mentioned above, for all examples studied we have found that a copy of the original model appears after one step of ER. To find the QSS by looking for fixed point excitations under quotient ERG flow using the $ExS$ map, we need to consider only those rows that correspond to the copy of the original model on the coarse-grained lattice. For a CSS model with one type of $X$ stabilizer generator, we need to consider only one row of the $ExS_X$ map, denoted $ExSR_X$. This map, $ExSR_X$, specifies how $X$ type of excitations flow under quotient ER. The quotienting out of excitations that flow into self-bifurcating B models is achieved by ignoring the corresponding rows in $ExS_X$, which specify how excitations flow to those B models. The fixed points of $ExSR_X$ capture the potentially nontrivial QSS. 

To find the fixed points of $ExSR_X$ we take the infrared limit, corresponding to many steps of ER, after which an arbitrary excitation will be supported on the coarse-grained unit cell. This allows us to restrict our attention to columns that contain only $\mathbb{Z}_2$ entries with no polynomial variables. We then utilize the mapping from polynomials to unit vectors described below Eq.~\eqref{order_cg_cell} to translate the monomials in $ExSR_X$ into columns, resulting in a square matrix with $\mathbb{Z}_2$ entries, which we also denote $ExSR_X$. With this replacement, the flow of excitations within the unit cell of a quotient fixed point model under quotient ER has been specified. The fixed points of $ExSR_X$, after the replacement, provide representatives for the QSS. Several examples are worked out below. 

\subsubsection{X-cube model}

Under ER, the X-cube model is mapped to itself, stacks of 2D toric code along each axis, and decoupled trivial qubits. 
In the Quotient Superselection Sectors section of the supplementary \small{MATHEMATICA} file SMERG.nb, column 8 contains the X-stabilizer term of X-cube model while Columns 11 and 19 contain the Z stabilizer terms. The map $ExS$ can be found in terms of the column operations used during ER, as described above. 
The row corresponding to the X-stabilizer term of X-cube in the map $ExS_X$, i.e. $ExSR_X$, is given by
\[
\left(\begin{array}{cccccccc}
x^\prime+y^\prime+x^\prime y^\prime & x^\prime+y^\prime+x^\prime y^\prime & 1 & 1 & 1 & 1 & 1 & 1\end{array}\right)
\, ,
\] 
where $x^\prime$, $y^\prime$ and $z^\prime$ are again translation variables on the coarse-grained lattice. Using the fact that the charge configuration given by $1+x^\prime+y^\prime+x^\prime y^\prime$ is trivial in the X-cube model, we rewrite this row as 
\[
\left(\begin{array}{cccccccc}
1 & 1 & 1 & 1 & 1 & 1 & 1 & 1\end{array}\right)
\, .
\]
To capture the flow of excitations from the coarse-grained unit cell to the new copy of X-cube produced by ER, we write this row in terms of unit vectors as described above in Eq.~\eqref{order_cg_cell}. The $1$ entries are replaced by column vectors $e_1$ 
\[
\left(\begin{array}{cccccccc}
e_{1} & e_{1} & e_{1} & e_{1} & e_{1} & e_{1} & e_{1} & e_{1}\end{array}\right) \, ,
\]
where $e_1=\left(\begin{array}{cccccccc}
1 & 0 & 0 & 0 & 0 & 0 & 0 & 0\end{array}\right)$. 
Here, $e_1$ specifies that all $X$ stabilizer excitations within the unit cell flow to the $X$ stabilizer excitation at position $1$ of the new copy of X-cube on the coarse-grained lattice under quotient ER. 
Accordingly, the fixed point of this matrix is simply $e_1$, corresponding to the $\mathbb{Z}_2$ fracton excitation from the $X$-sector of the X-cube model.

\begin{figure}[t!]
\centering
\includegraphics[scale=1]{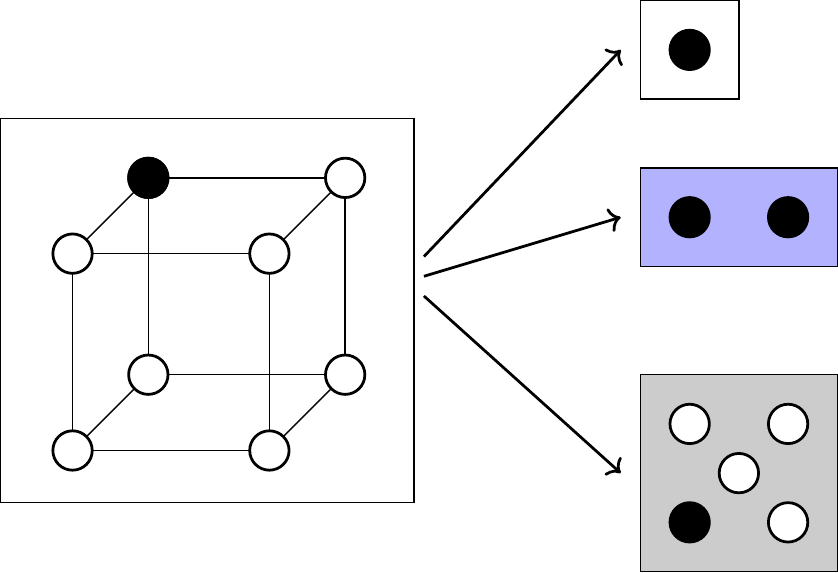}
\caption{The flow of $X$-sector excitations in cubic code 1 under the ER transformation depicted in Fig.~\ref{CC1RG}. The coarse-grained unit cell of CC$_1$ is shown as a cube containing circles that represent $X$ stabilizers, which are filled in black to represent excitations. Under ER, the $X$ excitation at position $z$ in the coarse-grained unit cell, denoted $e_2$, is mapped to excitations of the  $X$ stabilizer of the coarse grained copy of CC$_1$ represented by a white box, both $X$ stabilizers of $\text{CCB}_1$ represented by the purple box and a local excitation in the trivial sector represented by the gray box.} 
\label{quotient_sectors_CC1}
\end{figure}

On the other hand, the $Z$-sector splits into 16 different stabilizer terms during ER. Hence there are 16 columns in the map $ExS_Z$. The rows corresponding to the two $Z$-stabilizer terms in the copy of X-cube on the coarse-grained lattice are 
\begin{widetext}
\[
\left(\begin{array}{cccccccc}
1 & \overline{z}^\prime & 1 & 1 & 1 & z & 1 & 1\\
1+\overline{x}^\prime & \overline{z}^\prime+\overline{x}^\prime \overline{z}^\prime & 0 & 0 & 0 & 0 & 0 & 0
\end{array}\begin{array}{cccccccc}
1+\overline{z}^\prime & 0 & 0 & 0 & 0 & \overline{x}^\prime+\overline{x}^\prime \overline{z}^\prime & 0 & 0\\
\overline{x}^\prime & 1+\overline{z}^\prime+\overline{x}^\prime \overline{z}^\prime & 1 & 1 & \overline{x}^\prime & \overline{x}^\prime+{\overline{x}^\prime}^{2}\overline{z}^\prime+x^\prime \overline{z}^\prime & \overline{x}^\prime & 1
\end{array}\right)
\, .
\]
\end{widetext}
Using the triviality of charge configurations $1+\overline{x}^\prime$ and $1+\overline{z}^\prime$ in the X-cube model, we rewrite these rows as 
\[
\left(\begin{array}{cccccccc}
1 & 1 & 1 & 1 & 1 & 1 & 1 & 1\\
0 & 0 & 0 & 0 & 0 & 0 & 0 & 0
\end{array}\begin{array}{cccccccc}
0 & 0 & 0 & 0 & 0 & 0 & 0 & 0\\
1 & 1 & 1 & 1 & 1 & 1 & 1 & 1
\end{array}\right)
\, .
\]
Again, converting the position labels to unit vectors, we obtain the map for the flow of excitations 
\[
\left(\begin{array}{cccccccc}
e_{1} & e_{1} & e_{1} & e_{1} & e_{1} & e_{1} & e_{1} & e_{1}\\
e_0 & e_0 & e_0 & e_0 & e_0 & e_0 & e_0 & e_0
\end{array}\begin{array}{cccccccc}
e_0 & e_0 & e_0 & e_0 & e_0 & e_0 & e_0 & e_0\\
e_{1} & e_{1} & e_{1} & e_{1} & e_{1} & e_{1} & e_{1} & e_{1}
\end{array}\right) \, ,
\]
where $e_0 =\left(\begin{array}{cccccccc}
0 & 0 & 0 & 0 & 0 & 0 & 0 & 0\end{array}\right)$.  
The fixed points are
\begin{align*} 
\left(\begin{array}{c}
e_{1}\\
e_{0}
\end{array}\right) && \text{and} && \left(\begin{array}{c}
e_{0}\\
e_{1}
\end{array}\right)
\end{align*}
which correspond to the two generating $\mathbb{Z}_2$ lineons in the Z-sector of the X-cube model. 
Hence we find that the QSS group of X-cube is given by $\mathbb{Z}_2 \oplus \mathbb{Z}_2^2$, which agrees with the result in Ref.~\onlinecite{shirley2018Fractional}.

\subsubsection{Cubic code 1}

After ER, cubic code 1 bifurcates into a copy of itself, cubic code B and some disentangled trivial qubits. The matrix $ExSR_X$ is given by
\[
\left(\begin{array}{cccccccc}
1 & 1 & 1 & x^\prime & 1 & 1+x^\prime+z^\prime & 1+x^\prime+y^\prime & x^\prime+y^\prime+z^\prime\end{array}\right).
\]
Using the triviality of the charge configuration $1+x^\prime+y^\prime+z^\prime$ in cubic code 1, we rewrite this row as 
\[
\left(\begin{array}{cccccccc}
1 & 1 & 1 & x^\prime & 1 & y^\prime & z^\prime & 1\end{array}\right) \, ,
\]
which specifies the flow of excitations into the CC$_1$ model on the coarse-grained lattice. Converting the position labels into unit vectors, we can write this as a matrix
\begin{equation}
	\begin{split}
 \left(\begin{array}{cccccccc}
e_{1} & e_{1} & e_{1} & e_{5} & e_{1} & e_{3} & e_{2} & e_{1}\end{array}\right) \, .
\end{split}
\label{CC1_excflow}
\end{equation}
This means, for example, that the $X$-excitation at position $yz$ in the unit cell shown flows to an excitation of the $X$-stabilizer at position $x^\prime$ in the coarse-grained lattice. Considering the next-coarse-grained unit cell, this position is equivalently expressed by the unit vector $e_5$. 
The fixed point of the matrix in Eq.~\eqref{CC1_excflow} is $e_1$ which corresponds to the $\mathbb{Z}_2$ fracton in $X$-sector of cubic code 1. 
%and is the only excitation type in the quotient superselection sector. 
We remark that the $Z$-sector behaves similarly due to the relation between sectors in the cubic codes and hence the QSS group is given by $\mathbb{Z}_2 \oplus \mathbb{Z}_2$. 

In Fig.~\ref{quotient_sectors_CC1}, we depict how the excited $X$ stabilizer at position $z$ splits into different excitations in the new models after ER. One of the split excitations is supported in the copy of the original model CC$_1$ and hence the excitation at position $z$ on the original lattice flows to the nontrivial fixed point under quotient ER and therefore is in the nontrivial QSS. In addition, both $X$ stabilizers of $\text{CCB}_1$ are excited and the local stabilizer on a decoupled qubit is excited. The excitation in the coarse-grained lattice of CC$_1$ is found at position $1$ in the next-coarse-grained unit cell. Upon further quotient ER, the excitation remains at position $1$ as it is a fixed point.

\subsubsection{Cubic code 11}
\label{CC11QSS}
Cubic code 11 supports fractons, lineons and planons, where the lineons and planons are composites of fractons. 
Under ER, cubic code 11 splits into itself, a self-bifurcating lineon model and a stack of 2D toric codes. 
In this case, the $ExSR_X$ map is given by
\[
\left(\begin{array}{cccccccc}
y^\prime & y^\prime & 1 & 1 & y^\prime & 1+y^\prime+z^\prime & 1 & 1+{y^\prime}^{2}+{y^\prime}^{2}z\end{array}\right) \, .
\]
Using the triviality of charge configurations  ${1+y^\prime+{y^\prime}^2+z^\prime+{y^\prime}{z^\prime}+{y^\prime}^2{z^\prime}}$ and ${1+{x^\prime}+{y^\prime}+{y^\prime}{z^\prime}}$ in cubic code 11, we rewrite $ExSR_X$  as 
\[
\left(\begin{array}{cccccccc}
y^\prime & y^\prime & 1 & 1 & y^\prime & 1+y^\prime+z^\prime & 1 & 1+x^\prime+z^\prime\end{array}\right).
\]
Converting the position labels to unit vectors, we get find the matrix 
\[
\left(\begin{array}{cccccccc}
e_{3} & e_{3} & e_{1} & e_{1} & e_{3} & e_{1}+e_{3}+e_{2} & e_{1} & e_{1}+e_{5}+e_{2}\end{array}\right) \, ,
\]
whose only nonzero fixed point is given by $e_1+e_3$. Since the $Z$-sector behaves similarly we find that the QSS group is $\mathbb{Z}_2 \oplus \mathbb{Z}_2$. 

\section{Discussions and Conclusions}

In this work we have systematically studied unconventional bifurcating ERG flows of gapped stabilizer Hamiltonian models. We introduced the notions of self-bifurcating and quotient bifurcating fixed points to organize the possible behavior of bifurcating fixed point models. This inspired us to define the notion of bifurcated equivalence class, generalizing the conventional notion of gapped phase and foliated fracton equivalence. This also led to a natural definition of quotient superselection sectors. 
These ideas were then brought to bear on a large range of stabilizer model examples, including 17 of Haah's cubic codes and Yoshida's fractal spin liquids. All models were found to be bifurcating ERG fixed points. Furthermore, many cubic codes and all first-order fractal spin liquids were found to be self-bifurcating. These results are summarized in table~\ref{results_ERG}.

We found that the long-range entanglement features of stabilizer models provide insights into their structure. 
The mobilities of particles in each model was found to constrain the nature of the self-bifurcating models produced by ER. 
For models that support planon particles, we found that a stack of 2D toric codes could be extracted during ER. 
For example, CC$_{11}$ supports planons with mobility in the lattice planes perpendicular to $\hat{x}$ and consequently a stack of 2D toric codes is extracted during ER. 
This leads us to conjecture~\cite{dua_proof_str} that the existence of a planon in a 3D translation-invariant topological stabilizer model implies that a copy of the 2D toric code can be extracted via a local unitary transformation. Extending this in a translation-invariant way leads to the observed extraction of a stack of 2D toric codes. 
We also conjecture that if the original model has a particle with three dimensional mobility, a 3D toric code can be extracted via a local unitary circuit.

Models in which all the elementary excitations can move along a certain direction $\hat{i}$ are observed to be invariant under ER that coarse-grains along $\hat{i}$. 
This is consistent with the fact that the number of encoded qubits $k$ does not change as the system size grows along this particular direction i.e. ${k(2L_i)=k(L_i)}$. 
%Conversely we conjecture about ER fix pt in directions implies existence of mobile particles
Conversely, it was shown in Ref.~\onlinecite{yoshida2011classification} that all particles must have three-dimensional mobility in any 3D topological stabilizer model that has a constant ground space degeneracy as the system size increases along all three axes, and hence such models are in TQFT-type phases. 
Inspired by this result, we conjecture that this principle can be extended to cover models with a ground space degeneracy that is constant as the system size grows along two axes, in which case we posit that all particles in the theory must be at least mobile in a 2D plane spanned by the two axes, and consequently the model must be a stack of planon and TQFT models. 
Similarly we conjecture for models with a ground space degeneracy that is constant as the system size grows along one axis, that all particles must be at least mobile along that axis, and consequently such a model is a stack of lineon, planon and TQFT models. 
Combining this with conjectures posed previously in Ref.~\onlinecite{Dua_Classification_2019}: that all TQFT stabilizer models are equivalent to copies of the 3D toric code (possibly with fermionic point particle) and planon stabilizer models are equivalent to a stack of 2D toric codes,  
and the fact that lineon models can be compactified~\cite{Dua2019_compactify} along the lineon direction to produce 2D subsystem symmetry breaking models,  
points the way towards general classification results for 3D translation-invariant topological stabilizer models, inspired by the ERG flows we have found in examples.  
We leave the proofs of these conjectures to a future work.

For bifurcating models, when there is an additive correction to the exponential scaling of the number of encoded qubits $k$ as the system size increases by a factor $c$, i.e. ${k(c L)=\alpha k(L)+\beta}$, the model cannot self-bifurcate under coarse graining by $c$. 
For models that do self-bifurcate, the additive correction vanishes i.e. ${k(c L)=\alpha k(L)}$. 
We have calculated the number of encoded qubits for many examples over a range of system sizes, see table~\ref{encoded_qubits_table} in the appendix, and confirmed that they are consistent with the ERG flows found in table~\ref{results_ERG}. 
We remark that a correction to the linear scaling of $k(L)$ was the first indication that X-cube is a nontrivial foliated model~\cite{vijay2016fracton}. It would be interesting to search for invariant quantities that characterize nontrivial bifurcated models, such as the correction $\beta$ or similarly inspired corrections to entanglement entropy, generalizing ideas that arose in the study of foliated fracton models~\cite{shirley2018universal}. 

Our examples have focused on coarse-graining by a factor of two, which appeared natural for qubit systems. In particular each quotient fixed point example is bifurcated equivalent to itself on all lattices with spacings that are related by a multiple of two. It would be interesting to extend this to other primes, in an attempt to remove the lattice scale entirely from the bifurcated equivalence class of these models. 

We also note the appearance of the 2-adic norm of $L$ (for lattice spacing $a=1$) in the ratio of the number of quotient fixed point branches to all branches, the vast majority of which are generated by self-bifurcating B models, when one starts with an initial system on the $L\times L \times L$ 3D torus and applies ERG until all factors of $2$ have been removed from $L$. Similarly for $L_2 = 2^{n} L_1$ the ratio $\mathrm{GSD}(L_1)/\mathrm{GSD}(L_2)$ should be given by the 2-adic norm of $L_2/L_1$. 
The further appearance of $p$-adic norms in fracton models remains an intriguing prospect.

The bifurcating ER concepts and methods employed here apply equally well to subsystem symmetry breaking models. It would be interesting if they could shed light onto the classification of 2D translation-invariant spontaneous symmetry breaking stabilizer models, which has been accomplished in 1D~\cite{haah2013} but remains open in 2D due to the presence of complicated fractal like symmetries~\cite{Williamson_cubic_code,yoshida2013exotic}. This classification problem appears to be contained within the classification of 3D translation-invariant topological stabilizer models due to the existence of lineon models that can be compactified along the lineon direction to produce any known 2D subsystem symmetry breaking model.  
The bifurcating ER methods can also be extend directly to subsystem symmetry-protected topological phases (SSPT), by imposing a subsystem symmetry constraint on the individual gates in each ER step. This should shed light on the question: what is the appropriate definition of SSPT equivalence class?~\cite{subsystemphaserel,Devakul2018,Shirley2019} 

It would also be interesting to extend the ERG approach applied in this work to fractonic U(1) tensor gauge theories~\cite{PhysRevB.95.115139,PhysRevB.96.035119}. It has been found that upon Higgsing such theories may transition to either fractonic or conventional topological orders~\cite{ma2018fracton,PhysRevB.97.235112,bulmash2018generalized,Williamson2018Fractonic}, so it raises the question of how the ``parent'' fractonic U(1) gauge theories behave under ERG.  

Finally, exact ERG flows have been found for TQFT Hamiltonians beyond stabilizer models~\cite{Konig2009}. It would be interesting to extend these ERG flows to bifurcating ERG flows for nonabelian fracton models~\cite{prem2018cage,song2018twisted}. It is currently unclear if foliated equivalence extends straightforwardly to nonabelian fracton models as their ground state degeneracies behave somewhat differently to the abelian case~\cite{hao_twisted}. 

\vspace{.1cm}
\noindent
\textit{Note added}: During the writing of this paper we learnt of the work by Shirley, Slagle and Chen on closely related topics~\cite{ShirleyERG}. 

\acknowledgments
AD and DW thank Jeongwan Haah for useful discussions. DW also thanks Yichen Hu and Abhinav Prem. PS thanks Jayameenakshi Venkatraman for introducing him to quantum error correction. This work is supported by start-up funds at Yale University (DW and MC), NSF under award number DMR-1846109 and the Alfred P. Sloan foundation (MC). 
 
\bibliographystyle{apsrev_1}
\bibliography{DomBib}

\onecolumngrid
\newpage
\appendix
\input{appendix}
\end{document}

%% file: appendix.tex
\section{List of codes and their polynomial representations}
\label{stabideal_codes}
In Table~\ref{poly_codes}, we list the polynomials $f$ and $g$ for various cubic codes, denoted CC$_i$ where $i$ runs from 1 to 17, and their B-models if any, denoted CCB$_i$. The stabilizer map for a cubic code in terms of $f$ and $g$ is given by %\eqref{smap_fg}
\begin{align}
    \sigma=\left(\begin{array}{cc}
     f & 0  \\
     g & 0  \\
     0   & \overline{g}\\
     0   & \overline{f}
\end{array}\right) \, ,
\end{align}
and the stabilizer ideal is defined as the ideal $\av{f,g}$, which contains the polynomials in the image of the excitation map, see Eq.~\eqref{emap_fg}. 
The third column of the table below contains polynomial entries, each of which specifies the positions of two non-trivial clusters of charges that are created at the ends of string operators. For example, $y+xz^2$ implies that the small segment of string operator in CC$_5$ creates elementary excitations at $y$ and $xz^2$. This means that the direction of each string operator can be read off the binomial factor in its description. 

\begin{table}[H]
    \centering
\begin{tabular}{c|cc|c}
Model & $f$ & $g$ & Positions of charges in a trivial pair\\
\hline
CC$_1$& $1+x+y+z$ & $1+xy+yz+xz$ & \tabularnewline
CC$_2$& $x+y+z+yz$ & $1+y+xy+z+xz+xyz$ & \tabularnewline
CC$_3$ & $1+x+y+z$ & $1+xz+yz+xyz$ & \tabularnewline
CC$_4$& $1+x+z+yz$ & $1+y+xy+xz$ & \tabularnewline
 
CC$_5$ & $1+x+z+yz$ & $y+z+xz+yz$ & $y+xz^2$ \tabularnewline
 
CC$_6$ & $1+x+y+z$ & $1+y+xz+yz$ & $x+z^2$ \tabularnewline
 
CC$_7$ & $1+x+y+z$ & $1+z+yz+xyz$ &  \tabularnewline
 
CC$_8$ & $1+x+z+yz$ & $1+y+xy+z+xz+yz$ &  \tabularnewline
 
CC$_9$ & $1+z+xz+yz$ & $1+x+y+xyz$ & $xyz^2+1$ \tabularnewline
 
CC$_{10}$ & $1+x+z+yz$ & $1+y+xy+xz+yz+xyz$ &  \tabularnewline
 
CC$_{11}$ & $1+x+y+yz$ & $x+y+z+xy$ &  $(1+y+y^2)(1+z)$\tabularnewline

$\text{CCB}_{11}$ & $1 + x + y + y z $ & $1 + y + y^2$ & $y^3+1$ \tabularnewline

CC$_{12}$ & $1+x+xy+z$ & $1+x+xz+yz$ & $(1+xy+x^2)(1+z)$ \tabularnewline
 
$\text{CCB}_{12}$ & $ x + x^2 + z$ & $x y + z + x z + y z$ & $z^2+yx^3$ \tabularnewline

CC$_{13}$ & $x+z+xz+yz$ & $1+xz+yz+xy$  & $(1+xy+x+x^2)(1+z)$  \tabularnewline

$\text{CCB}_{13}$ & $y + z$ & $1 + z + x z + z^2$  & $y+z$ \tabularnewline

CC$_{14}$ & $1+x+z+xyz$ & $1+x+xy+xz+yz+xyz$ & $(1+x^2+yx^2)(1+z)$  \tabularnewline

$\text{CCB}_{14}$ & $x + y^2 + x y^2 + x z^2$ & $1 + y + y^2 + z + y z$ & $(1+x+y^2)(1+z)$  \tabularnewline

CC$_{15}$ & $1+xy+z+xz$ & $1+xy+y+xyz$ & $(1+xz+x^2z)(1+y)$ \tabularnewline

$\text{CCB}_{15}$ & $ y + x z + z^2 + x z^2$ & $z x + z x^2 + 1$ & $zyx^3+1$\tabularnewline

CC$_{16}$\footnote{CC$_{16}$ is equivalent to CC$_{15}$ after modular transformations. Hence, we don't include $\text{CCB}_{16}$ explicitly here.} & $1+z+xz+xyz$ & $1+xy+y+z$ & $(x+y+xy)(1+xz)$ \tabularnewline
 
CC$_{17}$ & $1+xy+yz+xz$ & $1+x+y+xy+z+xyz$ & $(1+z+z^2)(1+x)(1+y)$ \tabularnewline

$\text{CCB}_{17}$ & $1 + z + z^2$ & $1 + x y + x z + y z$ & $z^3+1$ \tabularnewline
\tabularnewline

\end{tabular}
    \caption{Polynomial representation for stabilizers and string operator segments in cubic code models. 
    The polynomials in the third column describe string operator segments and belong to the stabilizer ideal $I=\av{f,g}$~\cite{vijay2016fracton}. The binomial part specifies the direction of the string segment while the preceding polynomial specifies the cluster of elementary charges at its each end. The supplementary SAGE (python) file shows that the full polynomials belong to the stabilizer ideal and the preceding polynomials representing the non-trivial charge clusters do not.} 
    \label{poly_codes}
\end{table}

\section{Demonstrating our heuristic procedure for entanglement renormalization on the X-cube model}
\label{XC_ER}
In this appendix, we outline our heuristic procedure for entanglement renormalization on CSS codes and demonstrate it, using the X-cube model as an example. 
The heuristic is quite simple, essentially we coarse grain until a monomial appears in a column. We next use that entry to set the rest of the column to 0 via CNOT gates. The monomial entry can then be used to set its row to 0 via column operations. 
After this the qubit has been disentangled into the trivial state and its row and column can be removed. 
This is repeated until there are no monomials left. At which point the model has either bifurcated or we perform more steps of coarse graining. 
There is an important subtlety to the above recipe, as we do not necessarily need a monomial to set the rest of a column and row to 0. In some cases, such as the X-cube example below, a non-monomial entry, pair of entries, or several entries, in a column can be used to set the rest of the column to zero. This more general step can be used in place of the monomial step mentioned above. 

The stabilizer map for the X-cube model is given by 
\begin{equation}
\begin{pmatrix}
    (1+y)(1+z) & & \\
    (1+x)(1+z) & &\\
    (1+x)(1+y) & &\\
    & 1+\oline{x} & 1+\oline{x}\\
    & 0 & 1+\oline{y}\\
    & 1+\oline{z} & 0
    \end{pmatrix}
    \, .
\end{equation} 
After coarse-graining in $x$, the $X$ and $Z$ sectors of the stabilizer map are respectively given by
\begin{align}
    \begin{pmatrix}
        (1+y)(1+z) & 0\\
        0 & (1+y)(1+z)\\
        1+z & x(1+{z})\\
        1+z & 1+z\\
        1+y & x(1+y)\\
        1+y & 1+y
    \end{pmatrix}\, ,
    &&
    \begin{pmatrix}
        1 & 1 & 1 & 1\\
        \oline{x} & 1 & \oline{x} & 1\\
        0 & 0 & 1+\oline{y} & 0\\
        0 & 0 & 0 & 1+\oline{y}\\
        1+\oline{z} & 0 & 0 & 0\\
        0 & 1+\oline{z} & 0 & 0
    \end{pmatrix} \, ,
\end{align}
where $x$ is now the translation variable in the coarse-grained unit cell.
The goal is to decouple different models or stabilizer terms such that they are supported on non-overlapping sets of qubits. We first do this for the $X$-sector. In doing so, the $Z$-sector will also get modified. In the two columns for the $X$-sector, the $4^{th}$ and $6^{th}$ row elements in both columns are the same. Thus, applying the column operation Col$(2,1,1)$ simplifies the sectors of the stabilizer map
\begin{align}
    \begin{pmatrix}
        (1+y)(1+z) & (1+y)(1+z)\\
        0 & (1+y)(1+z)\\
        1+z & (1+x)(1+{z})\\
        1+z & 0\\
        1+y & (1+x)(1+y)\\
        1+y & 0
    \end{pmatrix}\, , 
    &&
    \begin{pmatrix}
        1 & 1 & 1 & 1\\
        \oline{x} & 1 & \oline{x} & 1\\
        0 & 0 & 1+\oline{y} & 0\\
        0 & 0 & 0 & 1+\oline{y}\\
        1+\oline{z} & 0 & 0 & 0\\
        0 & 1+\oline{z} & 0 & 0
    \end{pmatrix} \, .
\end{align}
Now, since there are same polynomials in certain row elements of the first column, one can get rid of them by applying CNOT$(3,4,1)$ and CNOT$(5,6,1)$ such that the second column is not affected due to the zero entries in it. 
\begin{align}
    \begin{pmatrix}
        (1+y)(1+z) & (1+y)(1+z)\\
        0 & (1+y)(1+z)\\
        0 & (1+x)(1+{z})\\
        1+z & 0\\
        0 & (1+x)(1+y)\\
        1+y & 0
    \end{pmatrix} \, ,
    &&
    \begin{pmatrix}
        1 & 1 & 1 & 1\\
        \oline{x} & 1 & \oline{x} & 1\\
        0 & 0 & 1+\oline{y} & 0\\
        0 & 0 & 1+\oline{y} & 1+\oline{y}\\
        1+\oline{z} & 0 & 0 & 0\\
        1+\oline{z} & 1+\oline{z} & 0 & 0
    \end{pmatrix} \, .
\end{align}
Finally, applying CNOT$(1,6,1+z)$ and CNOT$(1,2,1)$ decouples the two stabilizer terms in the $X$-sector as follows, 
\begin{align}
    \begin{pmatrix}
        0 & 0\\
        0 & (1+y)(1+z)\\
        0 & (1+x)(1+{z})\\
        1+z & 0\\
        0 & (1+x)(1+y)\\
        1+y & 0
    \end{pmatrix} \, ,
    &&
        \begin{pmatrix}
        1 & 1 & 1 & 1\\
        1+\oline{x} & 0 & 1+\oline{x} & 0\\
        0 & 0 & 1+\oline{y} & 0\\
        0 & 0 & 1+\oline{y} & 1+\oline{y}\\
        1+\oline{z} & 0 & 0 & 0\\
        0 & 0 & 1+\oline{z} & 1+\oline{z} 
    \end{pmatrix} \, .
\end{align}
It turns out that once the $X$-sector is decoupled, doing certain column operations in the $Z$-sector is enough to decouple the whole stabilizer map into two different models. This is also observed for other models. Applying Col$(5,6,1)$, Col$(3,4,1)$ and Col$(6,4,1)$ gives 
\begin{align}
    \begin{pmatrix}
        0 & 0\\
        0 & (1+y)(1+z)\\
        0 & (1+x)(1+z)\\
        1+z & 0\\
        0 & (1+x)(1+y)\\
        1+y & 0
    \end{pmatrix} \, ,
    &&
        \begin{pmatrix}
        0 & 1 & 0 & 0\\
        1+\oline{x} & 0 & 1+\oline{x} & 0\\
        0 & 0 & 1+\oline{y} & 0\\
        0 & 0 & 0 & 1+\oline{y}\\
        1+\oline{z} & 0 & 0 & 0\\
        0 & 0 & 0 & 1+\oline{z} 
    \end{pmatrix} \, .
\end{align}
Removing the disentangled qubit (the 1st qubit), we have found that the model splits into two, one is given by
\begin{equation}
\begin{pmatrix}
    1+z & \\
    1+y & \\
    & 1+\oline{y} \\
    & 1+\oline{z}
\end{pmatrix},
\end{equation}
which is just a stack of 2D toric codes parallel to the $yz$ plane, and the other is a coarse-grained  X-cube model.

\section{Cubic codes as fractal spin liquids}
\label{CCFSLs}
Certain cubic codes could be mapped to the following fractal spin liquid (FSL)~\cite{yoshida2013exotic} form of the stabilizer map,
\begin{align}
    \left(\begin{array}{cc}
         1+f(x) y & 0\\
         1+g_1(x)z+g_2(x)z^2 & 0\\
         0 & 1+g_1(\overline{x})\overline{z}+g_2(\overline{x})\overline{z}^2\\
         0 & 1+f(\overline{x})\overline{y}
    \end{array}\right)
    \, ,
    \label{FSLform}
\end{align}
which represents a first-order FSL for $g_2(x)=0$ and a second-order FSL otherwise. Here $\overline{x}_i\equiv x_i^{-1}$ where $x_i$ denotes $x$, $y$ or $z$. In table~\ref{poly_codes}, we write down the polynomials that appear in the FSL forms of various cubic codes. 
The explicit mapping from the cubic codes to FSLs is contained in the supplementary \small{MATHEMATICA} file SMERG.nb.

\begin{table}[H]
\label{tableFSLs}
    \centering
\begin{tabular}{c|ccc|c}
Model & $f(x)$ & $g_1(x)$ & $g_2(x)$\\
\hline
CC$_1$& $1+x+x^2$ & $1+x$ & $1+x+x^2$ \tabularnewline
CC$_2$& $1+x+x^2$ & $1+x^2$ & $1+x+x^2$ \tabularnewline
CC$_3$ & $(1+x+x^2)x^{-1}$ & $1+x$ & $(1+x+x^2)x^{-1}$ \tabularnewline
CC$_5$ & $x^2$ & $1+x$ & $x^2$ \tabularnewline
CC$_6$ & $x^2$ & $1+x+x^2$ & 0 \tabularnewline
CC$_9$\footnote{CC$_9$ is equivalent to CC$_5$} & $x^2$ & $1+x$ & $x^2$ \tabularnewline
\end{tabular}
    \caption{Polynomials in the fractal spin liquid form \eqref{FSLform} of cubic codes. }
    \label{poly_codes}
\end{table}

\section{Numerical results on number of encoded qubits}
In table~\ref{encoded_qubits_table}, we list the number of encoded qubits w.r.t. system size $L$ in the cubic codes and the Sierpinski fractal spin liquid (SFSL) on an $L\times L\times L$ lattice. The results are calculated in the supplementary \small{MATHEMATICA} file encodedqubits.nb. Results for some of these models were recovered analytically in the next section using the polynomial framework. 
\begin{table}[H]
    \centering
\begin{tabular}{c|cccccccccccccccccccc}
$L$ & CC$_1$& CC$_2$& CC$_3$ & CC$_4$& CC$_5$ & CC$_6$ & CC$_7$ & CC$_8$ & CC$_9$ & CC$_{10}$ & CC$_{11}$ & $\text{CCB}_{11}$ & CC$_{12}$ & CC$_{13}$ & CC$_{14}$ & CC$_{15}$ & CC$_{16}$ & CC$_{17}$ & SFSL\tabularnewline
 
\hline 
2 & 6  & 4 & 4 & 4 & 4 & 4 & 4 & 4 & 4 & 4 & 4 &  & 4 & 8 & 4 & 4 & 4 & 8 & 0 \tabularnewline

3 & 2 & 2 & 2 & 6 & 2 & 2 & 2 & 2 & 2 & 2 & 10 & 8 & 6 & 6 & 6 & 6 & 6 & 6 & 4 \tabularnewline

4 & 14 & 8 & 8 & 8 & 8 & 8 & 8 & 8 & 8 & 8 & 8 & 0 & 8 & 16 & 8 & 8 & 8 & 16 & 0 \tabularnewline

5 & 2 & 2 & 2 & 2 & 2 & 2 & 2 & 2 & 2 & 2 & 10 & 0 &10 & 10 & 10 & 10 & 10 & 10 & 0 \tabularnewline

6 & 6 & 4 & 4 & 12 & 4 & 4 & 4 & 4 & 4 & 4 & 24 & 16 & 16 & 16 & 12 & 16 & 16 & 24 & 8 \tabularnewline

7 & 2 & 14 & 2 & 2 & 14 & 14 & 2 & 2 & 14 & 2 & 14 & 0 & 26 & 26 & 26 & 26 & 26 & 14 & 12 \tabularnewline

8 & 30 & 16 & 16 & 16 & 16 & 16 & 16 & 16 & 16 & 16 & 16 & 0 & 16 & 32 & 16 & 16 & 16 & 32 & 0 \tabularnewline

9 & 2 & 2 & 2 & 6 & 2 & 2 & 2 & 2 & 2 & 2 & 22 & 8 & 18 & 18 & 18 & 18 & 18 & 18 & 4 \tabularnewline

10 & 6 & 4 & 4 & 4 & 4 & 4 & 4 & 4 & 4 & 4 & 20 & 0 & 20 & 24 & 20 & 20 & 20 & 40 & 0 \tabularnewline

11 & 2 & 2 & 2 & 2 & 2 & 2 & 2 & 2 & 2 & 2 & 22 & 0 & 22 & 22 & 22 & 22 & 22 & 22 & 0 \tabularnewline

12 & 14 & 8 & 8 & 24 & 8 & 8 & 8 & 8 & 8 & 8 & 52 & 32 & 36 & 32 & 24 & 36 & 36 & 56 & 16\tabularnewline

13 & 2 & 2 & 2 & 2 & 2 & 2 & 2 & 2 & 2 & 2 & 26 & 0 & 26 & 26 & 26 & 26 & 26 & 26 & 0 \tabularnewline

14 & 6 & 28 & 4 & 4 & 28 & 28 & 4 & 4 & 28 & 4 & 28 & 0 & 52 & 56 & 52 & 52 & 52 & 56 & 24 \tabularnewline

15 & 50 & 2 & 18 & 22 & 26 & 26 & 42 & 42 & 26 & 18 & 82 & 56 & 54 & 54 & 54 & 54 & 54 & 78 & 28 \tabularnewline

16 & 62 & 32 & 32 & 32 & 32 & 32 & 32 & 32 & 32 & 32 & 32 & 0 & 32 & 64 & 32 & 32 & 32 & 64 & 0 \tabularnewline

17 & 2 & 2 & 2 & 2 & 2 & 2 & 2 & 2 & 2 & 2 & 34 & 0 & 34 & 34 & 34 & 34 & 34 & 34 & 0 \tabularnewline

18 & 6 & 4 & 4 & 12 & 4 & 4 & 4 & 4 & 4 & 4 & 48 & 16 & 40 & 40 & 36 & 40 & 40 & 72 & 8 \tabularnewline

19 & 2 & 2 & 2 & 2 & 2 & 2 & 2 & 2 & 2 & 2 & 38 & 0 & 38 & 38 & 38 & 38 & 38 & 38 & 0  \tabularnewline

20 & 14 & 8 & 8 & 8 & 8 & 8 & 8 & 8 & 8 & 8 & 40 & 0 & 40 & 48 & 40 & 40 & 40 & 80 & 0  \tabularnewline
\end{tabular}
    \caption{Number of encoded qubits for different stabilizer models on a system of size $L\times L\times L$ where $L$ is the number of stabilizers in each lattice direction. CC$_i$'s refer to the cubic codes, $\text{CCB}_{11}$ to the cubic code 11B and $\text{SFSL}$ to the Sierpinski Fractal spin liquid. For models that self-bifurcate under coarse-graining by a factor of 2, the number of encoded qubits also doubles as the system size is increased by a factor of 2. For $\text{CCB}_{11}$, one needs a minimum cubic system size of $L=3$ since the stabilizer generators are supported on $2\times 2 \times 3$ unit cells. The results in the table are calculated in the supplementary \small{MATHEMATICA} file encodedqubits.nb}.
    \label{encoded_qubits_table}
\end{table}

\section{Number of encoded qubits}
\label{gsd_details}
In Ref.~\onlinecite{haah2013commuting}, Haah used techniques from commutative algebra to derive a formula for the number of encoded qubits for the cubic code CC$_1$ which we recount in \cref{subsec:CC1_degen}. We use the general strategy used by Haah to derive similar formulas for the number of encoded qubits for other cubic codes. We consider CC$_6$ and CC$_{11}$ as examples. We first explain the general definitions and explain the steps used by Haah in his derivation. 

Suppose $\mathfrak{a} = (f, g) \subset \mathbb F_2[x, y, z]$ is an ideal corresponding to a code. Fix some $L \in \mathbb Z_{> 0}$. Imposing the periodicity conditions $x^L - 1 = y^L - 1 = z^L - 1 = 0$, the number of encoded qubits is given by
\begin{align*}
k_{\mathfrak{a}} = 2\dim_{\mathbb F_2}(\mathbb F_2[x, y, z]/I_{\mathfrak{a}})
\end{align*}
where $\dim_{\mathbb F_2}(\mathbb F_2[x, y, z]/I_{\mathfrak{a}})$ is the dimension of $\mathbb F_2[x, y, z]/I_{\mathfrak{a}}$ as a vector space and $I_{\mathfrak{a}}$ is the ideal defined as
\begin{align*}
I_{\mathfrak{a}} = \mathfrak{a} + (x^L - 1, y^L - 1, z^L - 1) = (f, g, x^L - 1, y^L - 1, z^L - 1).
\end{align*}
Due to algebrogeometric reasons, it is preferable to work over an algebraically closed field. Hence we take the algebraic closure $\mathbb F = \overline{\mathbb F_2}$. By extension of scalars of the vector spaces, we have
\begin{align*}
k_{\mathfrak{a}} = 2\dim_{\mathbb F}(\mathbb F[x, y, z]/I_{\mathfrak{a}})
\end{align*}
where $I_{\mathfrak{a}}$ is now the ideal generated in $\mathbb F[x, y, z]$. In order to calculate this, we use the structure theorem for Artinian rings from \cite[Theorem 8.7]{AM16} which we restate in \cref{thm:DecomposionOfArtinianRings}.

\begin{theorem}
\label{thm:DecomposionOfArtinianRings}
Let $A$ be an Artinian ring. Then $A = \prod_{j = 1}^n A_j$ for some $n \in \mathbb Z_{>0}$ where $A_j$ is a local Artinian ring for all integers $1 \leq j \leq n$. Moreover, the decomposition is unique up to isomorphism.
\end{theorem}

\begin{remark}
In the above theorem, if $A$ is also a vector space over $\mathbb F$, we have $A = \bigoplus_{j = 1}^n A_j$ and in particular $\dim_{\mathbb F}(A) = \sum_{j = 1}^n \dim_{\mathbb F}(A_j)$.
\end{remark}

We apply \cref{thm:DecomposionOfArtinianRings} to obtain the formula $k_{\mathfrak{a}} = \sum_{\mathfrak{m}} k_{\mathfrak{a}, \mathfrak{m}}$ where
\begin{align*}
k_{\mathfrak{a}, \mathfrak{m}} = 2\dim_{\mathbb F}((\mathbb F[x, y, z]/I_{\mathfrak{a}})_{\mathfrak{m}})
\end{align*}
and the sum is taken over all maximal ideals $\mathfrak{m} \subset \mathbb F[x, y, z]/I_{\mathfrak{a}}$. By the weak form of Hilbert's Nullstellensatz, any maximal ideal $\mathfrak{m} \subset \mathbb F[x, y, z]/I_{\mathfrak{a}}$ is of the form $\mathfrak{m} = (x - x_0, y - y_0, z - z_0)/I_{\mathfrak{a}}$ where $(x_0, y_0, z_0) \in V(I_{\mathfrak{a}}) \subset \mathbb A_{\mathbb F}^3$, i.e., $(x_0, y_0, z_0) \in \mathbb A_{\mathbb F}^3$ satisfies
\begin{align*}
f(x_0, y_0, z_0) = g(x_0, y_0, z_0) = {x_0}^L - 1 = {y_0}^L - 1 = {z_0}^L - 1 = 0.
\end{align*}
In particular $x_0, y_0, z_0 \neq 0$. Since $x_0^L-1=x_0^{2^l L^\prime}-1=(x_0^{L^\prime}-1)^{2^l}=0$, we also have $x_0^{L^\prime}=1$ and similarly, $y_0^{L^\prime}=z_0^{L^\prime}=1$. We now calculate $k_{\mathfrak{a}, \mathfrak{m}}$ for all maximal ideals $\mathfrak{m} \subset \mathbb F[x, y, z]/I_{\mathfrak{a}}$.

Fix a maximal ideal $\mathfrak{m} = (x - x_0, y - y_0, z - z_0)/I_{\mathfrak{a}} \subset \mathbb F[x, y, z]/I_{\mathfrak{a}}$. Now write $L = 2^l L'$ where $2 \nmid L'$ and $l \in \mathbb Z_{\geq 0}$. Consider the factorization
\begin{align*}
x^L - 1 &= x^L - {x_0}^L = (x^{L'})^{2^l} - ({x_0}^{L'})^{2^l} = (x^{L'} - {x_0}^{L'})^{2^l} \\
&= \left((x - x_0)(x^{L' - 1} + x^{L' - 2}x_0 + \dotsb + x{x_0}^{L' - 2} + {x_0}^{L' - 1})\right)^{2^l} \\
&= (x - x_0)^{2^l}(x^{L' - 1} + x^{L' - 2}x_0 + \dotsb + x{x_0}^{L' - 2} + {x_0}^{L' - 1})^{2^l}.
\end{align*}
The right most factor has no further factors of $x - x_0$ since putting $x = x_0$ results in $L'{x_0}^{L' - 1} \neq 0$ since $2 \nmid L'$ and $x_0 \neq 0$. We obtain similar factorization for the variables $y$ and $z$. We recognize that $(x - x_0)^{2^l} = x^{2^l} + {x_0}^{2^l}$ and similarly for the variables $y$ and $z$ since $\Char(\mathbb F) = 2$, which motivates us to define the ideal
\begin{align*}
J_{\mathfrak{a}} = \mathfrak{a} + (x^{2^l} + a_0, y^{2^l} + b_0, z^{2^l} + c_0) = (f, g, x^{2^l} + a_0, y^{2^l} + b_0, z^{2^l} + c_0)
\end{align*}
where $a_0 = {x_0}^{2^l}$, $b_0 = {y_0}^{2^l}$ and $c_0 = {z_0}^{2^l}$. We have the canonical map $\mathbb F[x, y, z]/I_{\mathfrak{a}} \to \mathbb F[x, y, z]/J_{\mathfrak{a}}$ since $I_{\mathfrak{a}} \subset J_{\mathfrak{a}}$. We would like to apply \cite[Corollary 3.2]{AM16} which we restate in \cref{lem:CriteriaForIsomorphismToLocalization}.

\begin{lemma}
\label{lem:CriteriaForIsomorphismToLocalization}
Let $A$ and $B$ be rings and $S \subset A$ be a multiplicatively closed set. If $g: A \to B$ is a ring homomorphism such that
\begin{enumerate}
\item for all $s \in S$, the element $g(s) \in B$ is a unit
\item for all $a \in \ker(g)$, we have $as = 0$ for some $s \in S$
\item for all $b \in B$, there are $a \in A$ and $s \in S$ such that $b = g(a)g(s)^{-1}$
\end{enumerate}
then there is a unique isomorphism $h: S^{-1}A \to B$ such that $g = h \circ f$ where $f: A \to S^{-1}A$ is the canonical homomorphism.
\end{lemma}

Checking conditions of \cref{lem:CriteriaForIsomorphismToLocalization}, we have the isomorphism
\begin{align*}
(\mathbb F[x, y, z]/I_{\mathfrak{a}})_{\mathfrak{m}} \cong \mathbb F[x, y, z]/J_{\mathfrak{a}}
\end{align*}
as rings and hence also as $\mathbb F$-modules which are of course simply vector spaces over $\mathbb F$. Thus we need to calculate $\dim_{\mathbb F}(\mathbb F[x, y, z]/J_{\mathfrak{a}})$. Note that if $l = 0$, then $(\mathbb F[x, y, z]/I_{\mathfrak{a}})_{\mathfrak{m}} \cong \mathbb F[x, y, z]/J_{\mathfrak{a}} \cong \mathbb F$ is the residue field at $\mathfrak{m}$ and hence we simply get $\dim_{\mathbb F}(\mathbb F[x, y, z]/J_{\mathfrak{a}}) = 1$. We assume $l \in \mathbb Z_{>0}$ henceforth.

Most of the time we can eliminate one of the three variables, say $z$, and use the substitutions $x \mapsto x + 1$ and $y \mapsto y + 1$ to show $\mathbb F[x, y, z]/J_{\mathfrak{a}} \cong \mathbb F[x, y]/J_{\mathfrak{a}}^z$, as rings and hence as vector spaces over $\mathbb F$, where $J_{\mathfrak{a}}^z = (h, x^{2^l} + a, y^{2^l} + b)$, $a = a_0 + 1$, $b = b_0 + 1$ and $h$ is the polynomial derived from $f$ and $g$ by eliminating the variable $z$. This further simplifies the problem to the calculation of $\dim_{\mathbb F}(\mathbb F[x, y]/J_{\mathfrak{a}}^z)$. In the following sections, we calculate this for various cubic codes by first finding appropriate Gr\"{o}bner bases with the help of the computer algebra system SageMath. SageMath cannot calculate Gr\"{o}bner bases for a general $l \in \mathbb Z_{> 0}$ but we test various values of $l$ to accurately predict the Gr\"{o}bner bases in terms of $l$ and then we prove that it is indeed so. Then we use the Gr\"{o}bner bases to calculate $k_{\mathfrak{a}, \mathfrak{m}}$. To calculate this, we use \cite[Proposition 2.1.6]{AL94} which we restate in \cref{prp:VectorSpaceBasisUsingGrobnerBasis}.

\begin{theorem}
\label{prp:VectorSpaceBasisUsingGrobnerBasis}
Let $n \in \mathbb Z_{>0}$ and $R = \mathbb F[x_1, x_2, \dotsc, x_n]$ be a polynomial ring endowed with a monomial ordering. Let $I \subset R$ be an ideal and $G \subset R$ be a Gr\"{o}bner basis for $I$. Then the set
\begin{align*}
\{m + I \in R/I: m \in R \text{ is a monomial which is reduced with respect to } G\}
\end{align*}
is a basis for the vector space $R/I$ over $\mathbb F$.
\end{theorem}

Finally, we use $k_{\mathfrak{a}} = \sum_{\mathfrak{m}} k_{\mathfrak{a}, \mathfrak{m}}$, where the sum is over all maximal ideals $\mathfrak{m} \subset \mathbb F[x, y, z]/I_{\mathfrak{a}}$, which is feasible to compute explicitly whenever we can explicitly count the number of points in certain appropriate subvarieties of $V(I_{\mathfrak{a}})$.
\begin{remark}
All Gr\"{o}bner bases in the rest of the document are using the lexicographic ordering of monomials with $x < y < z$. Polynomial divisions with respect to a Gr\"{o}bner basis are also using the same ordering of monomials.
\end{remark}
\begin{remark}
We often need the roots of the polynomial $x^2 + x + 1 \in \mathbb F[x]$ which are the primitive cube roots of unity. We denote by $\zeta_n \in \mathbb F$ any choice of a primitive $n$\textsuperscript{th} root of unity, for all $n \in \mathbb Z_{>0}$.
\end{remark}

\subsection{Cubic code 1}
\label{subsec:CC1_degen}
The stabilizer generators of cubic code 1 (CC$_1$) are given by 
\begin{align}
\begin{array}{c}
\drawgenerator{XI}{II}{IX}{XI}{IX}{XX}{XI}{IX}
\quad
\drawgenerator{ZI}{ZZ}{IZ}{ZI}{IZ}{II}{ZI}{IZ}
\end{array}
\end{align}
Hence, the stabilizer ideal that defines CC$_1$ is $\mathfrak{a} = (f, g) = (x + y + z + 1, xy + yz + zx + 1)$. We eliminate the variable $z$ and use the substitutions $x \mapsto x + 1$ and $y \mapsto y + 1$ to obtain
\begin{align*}
J_{\mathfrak{a}}^z = (h, x^{2^l} + a, y^{2^l} + b) = (x^2 + xy + y^2, x^{2^l} + a, y^{2^l} + b).
\end{align*}
Recalling definitions, we know that $a = {x_0}^{2^l} + 1, b = {y_0}^{2^l} + 1$ where ${x_0}^{L'} = {y_0}^{L'} = 1$ such that $a^2 + ab + b^2 = 0$. The last equation is satisfied if and only if $a = b = 0$ or $\left(\frac{b}{a}\right)^2 + \frac{b}{a} + 1 = 0$. In the latter case $\frac{b}{a}$ is a primitive cube root of unity and we assume the choice $\zeta_3 = \frac{b}{a}$.

\begin{lemma}
\label{lem:CC1GrobnerBasis}
For the ideal $J_{\mathfrak{a}}^z = (x^2 + xy + y^2, x^{2^l} + a, y^{2^l} + b)$, we have the following Gr\"{o}bner bases.
\begin{enumerate}
\item\label{itm:CC1_a_b_Zero}	Suppose $a = b = 0$. Then
\begin{align*}
G = \{x^2 + xy + y^2, yx^{2^l - 1}, x^{2^l}\}
\end{align*}
is a Gr\"{o}bner basis.
\item\label{itm:CC1_a_b_NotZero}	Suppose $a^2 + ab + b^2 = 0$ with $a, b \neq 0$. Then
\begin{align*}
G = \{(\zeta_3 + l)x + y, x^{2^l} + a\}
\end{align*}
is a Gr\"{o}bner basis.
\end{enumerate}
\end{lemma}

\begin{proof}
It is a straight forward calculation using the S-polynomials of Buchberger's algorithm to verify that $G$ is a Gr\"{o}bner basis for the ideal $(G)$ for all the cases. It is also a straight forward calculation by polynomial divisions by elements in $G$ to verify that $J_{\mathfrak{a}}^z \subset (G)$ for all the cases. It remains to show that $(G) \subset J_{\mathfrak{a}}^z$ for all the cases. This is shown below assuming the same hypotheses and using the same set $G$ as in the lemma for the corresponding cases.

\medskip
\noindent
\textit{\cref{itm:CC1_a_b_Zero}.} The only nontrivial containment we need to show is $yx^{2^l - 1} \in J_{\mathfrak{a}}^z$, i.e., $yx^{2^l - 1} \equiv 0 \pmod{J_{\mathfrak{a}}^z}$. This follows if we show our claim $nx^{2^n} + x^{2^n - 1}y + y^{2^n} \equiv 0 \pmod{J_{\mathfrak{a}}^z}$ for all integers $0 \leq n \leq l$ by taking the $n = l$ case because $x^{2^l} \equiv y^{2^l} \equiv 0 \pmod{J_{\mathfrak{a}}^z}$. We show this by induction. The base case $n = 0$ is trivial. Suppose that the claim holds for some integer $0 \leq n - 1 \leq l - 1$, i.e., $(n - 1)x^{2^{n - 1}} + x^{2^{n - 1} - 1}y + y^{2^{n - 1}} \equiv 0 \pmod{J_{\mathfrak{a}}^z}$. Squaring this and using $x^2 + xy + y^2 \equiv 0 \pmod{J_{\mathfrak{a}}^z}$, we get
\begin{align*}
&(n - 1)^2x^{2^n} + x^{2^n - 2}y^2 + y^{2^n} \equiv 0 \pmod{J_{\mathfrak{a}}^z} \\
\implies{}&(n^2 - 1)x^{2^n} + x^{2^n - 2}(x^2 + xy) + y^{2^n} \equiv 0 \pmod{J_{\mathfrak{a}}^z} \\
\implies{}&nx^{2^n} + x^{2^n - 1}y + y^{2^n} \equiv 0 \pmod{J_{\mathfrak{a}}^z}
\end{align*}
where we use the fact that $n^2 \equiv n \pmod{2}$ by Fermat's little theorem or by directly checking both cases $n \equiv 0 \pmod{2}$ and $n \equiv 1 \pmod{2}$.

\medskip
\noindent
\textit{\cref{itm:CC1_a_b_NotZero}.} The only nontrivial containment we need to show is $(\zeta_3 + l)x + y \in J_{\mathfrak{a}}^z$, i.e., $(\zeta_3 + l)x + y \equiv 0 \pmod{J_{\mathfrak{a}}^z}$. The equation $nx^{2^n} + x^{2^n - 1}y + y^{2^n} \equiv 0 \pmod{J_{\mathfrak{a}}^z}$ holds for all integers $0 \leq n \leq l$ in this case as well. The proof is exactly the same as above. By hypothesis, $a, b \neq 0$. Using the case $n = l$ and multiplying by $\frac{x}{a}$, we get
\begin{align*}
&\frac{x}{a}(lx^{2^l} + x^{2^l - 1}y + y^{2^l}) \equiv 0 \pmod{J_{\mathfrak{a}}^z} \\
\implies{}&\frac{1}{a}(lx^{2^l + 1} + x^{2^l}y + bx) \equiv 0 \pmod{J_{\mathfrak{a}}^z} \\
\implies{}&\frac{1}{a}(lax + ay + bx) \equiv 0 \pmod{J_{\mathfrak{a}}^z} \\
\implies{}&\left(\frac{b}{a} + l\right)x + y \equiv 0 \pmod{J_{\mathfrak{a}}^z} \\
\implies{}&\left(\zeta_3 + l\right)x + y \equiv 0 \pmod{J_{\mathfrak{a}}^z}.
\end{align*}
\end{proof}
Now we simply read off the dimension using the Gr\"{o}bner basis and obtain the formula
\begin{align*}
\dim_{\mathbb F}(\mathbb F[x, y]/J_{\mathfrak{a}}^z) =
\begin{cases}
2^{l + 1} - 1, & a = b = 0 \\
2^l, & a^2 + ab + b^2 = 0 \text{ with } a, b \neq 0.
\end{cases}
\end{align*}
Recalling definitions, the above gives the formula
\begin{align*}
k_{\mathfrak{a}, \mathfrak{m}} =
\begin{cases}
2^{l + 2} - 2, & x_0 = y_0 = 1 \\
2^{l + 1}, & (x_0 + 1)^2 + (x_0 + 1)(y_0 + 1) + (y_0 + 1)^2 = 0 \text{ with } x_0, y_0 \neq 1.
\end{cases}
\end{align*}
Now we wish to calculate $k_{\mathfrak{a}} = \sum_{\mathfrak{m}} k_{\mathfrak{a}, \mathfrak{m}}$ where the sum is over all maximal ideals $\mathfrak{m} \subset \mathbb F[x, y, z]/I_{\mathfrak{a}}$. For the calculation, we first need to explicitly calculate the points in the variety $V(I_{\mathfrak{a}})$. To solve $h = 0$, we divide by $x^2$ to get $\left(\frac{y}{x}\right)^2 + \frac{y}{x} + 1 = 0$ and hence the solutions are $\{(x_0', \zeta_3 x_0') \in A_{\mathbb F}^2: x_0' \in \mathbb F\} \cup \{(x_0', {\zeta_3}^2 x_0') \in A_{\mathbb F}^2: x_0' \in \mathbb F\}$. Thus we calculate that the solutions of $f = g = 0$ is the set of points $V(\mathfrak{a}) = V_1(\mathfrak{a}) \cup V_2(\mathfrak{a})$ where
\begin{align*}
V_1(\mathfrak{a}) &= \{(x_0' + 1, \zeta_3 x_0' + 1, {\zeta_3}^2 x_0' + 1) \in A_{\mathbb F}^3: x_0' \in \mathbb F\} \\
V_2(\mathfrak{a}) &= \{(x_0' + 1, {\zeta_3}^2 x_0' + 1, \zeta_3 x_0' + 1) \in A_{\mathbb F}^3: x_0' \in \mathbb F\}.
\end{align*}
Note that the intersection is $V_{\cap}(\mathfrak{a}) = V_1(\mathfrak{a}) \cap V_2(\mathfrak{a}) = \{(1, 1, 1)\}$ which consists of the only point satisfying $x_0 = y_0 = 1$. Imposing the periodic conditions $x^L - 1 = y^L - 1 = z^L - 1 = 0$, we have the corresponding set of solutions $V(I_{\mathfrak{a}}) = V_1(I_{\mathfrak{a}}) \cup V_2(I_{\mathfrak{a}})$ where $V_j(I_{\mathfrak{a}}) = \{(x_0, y_0, z_0) \in V_j(\mathfrak{a}): {x_0}^{L'} = {y_0}^{L'} = {z_0}^{L'} = 1\}$ for all $j \in \{1, 2\}$. The intersection is $V_{\cap}(I_{\mathfrak{a}}) = \{(1, 1, 1)\}$. Thus, $\#V_{\cap}(I_{\mathfrak{a}}) = 1$ and
\begin{align*}
&\#(V_1(I_{\mathfrak{a}}) \setminus V_{\cap}(I_{\mathfrak{a}})) = \#(V_2(I_{\mathfrak{a}}) \setminus V_{\cap}(I_{\mathfrak{a}})) \\
={}&\deg\left(\gcd((x + 1)^{L'} + 1, (\zeta_3 x + 1)^{L'} + 1, ({\zeta_3}^2 x + 1)^{L'} + 1)\right) - 1.
\end{align*}
we put everything together using the formula
\begin{align}
k_{\mathfrak{a}} &= (2^{l + 2} - 2)\#V_{\cap}(I_{\mathfrak{a}}) + 2^{l + 1}\#(V_1(I_{\mathfrak{a}}) \setminus V_{\cap}(I_{\mathfrak{a}})) + 2^{l + 1}\#(V_2(I_{\mathfrak{a}}) \setminus V_{\cap}(I_{\mathfrak{a}})) \\
\label{eqn:CC1DegeneracyGCDFormula}
&= 2^{l + 2}\deg\left(\gcd((x + 1)^{L'} + 1, (\zeta_3 x + 1)^{L'} + 1, ({\zeta_3}^2 x + 1)^{L'} + 1)\right) - 2.
\end{align}
In general, it is difficult to obtain a more explicit formula. However, we obtain a more explicit formula for some special cases.

First suppose $L' = 1$. Then it is easy to see that $x$ is a greatest common divisor in \cref{eqn:CC1DegeneracyGCDFormula}. Hence \cref{eqn:CC1DegeneracyGCDFormula} gives $k_{\mathfrak{a}} = 2^{l + 2} - 2 = 4L - 2$.

Now suppose $L' = 2^n + 1$ for some $n \in \mathbb Z_{\geq 0}$. Then we have the factorization
\begin{align*}
(x + 1)^{2^n + 1} + 1 &= (x + 1)(x + 1)^{2^n} + 1 = (x + 1)(x^{2^n} + 1) + 1 \\
&= x^{2^n + 1} + x^{2^n} + x = x(x^{2^n} + x^{2^n - 1} + 1).
\end{align*}
We apply the same equation for $(\zeta_3 x + 1)^{2^n + 1} + 1$ and $({\zeta_3}^2 x + 1)^{2^n + 1} + 1$ to get the factorization
\begin{align*}
(\zeta_3 x + 1)^{2^n + 1} + 1 &= \zeta_3x((\zeta_3 x)^{2^n} + (\zeta_3 x)^{2^n - 1} + 1) \\
({\zeta_3}^2 x + 1)^{2^n + 1} + 1 &= {\zeta_3}^2 x(({\zeta_3}^2 x)^{2^n} + ({\zeta_3}^2 x)^{2^n - 1} + 1).
\end{align*}
Let $r$ be a root of the second factor $x^{2^n} + x^{2^n - 1} + 1$. Suppose that $2^n \equiv 1 \pmod{3}$. In this case we have
\begin{align*}
(\zeta_3 r)^{2^n} + (\zeta_3 r)^{2^n - 1} + 1 = \zeta_3 r^{2^n} + r^{2^n - 1} + 1 = \zeta_3 r^{2^n} + r^{2^n} = r^{2^n}(\zeta_3 + 1)^{2^n}.
\end{align*}
So $(\zeta_3 r)^{2^n} + (\zeta_3 r)^{2^n - 1} + 1 = 0$ implies $r = 0$ which is a contradiction. Suppose that $2^n \equiv 2 \pmod{3}$. In this case we have
\begin{align*}
(\zeta_3 r)^{2^n} + (\zeta_3 r)^{2^n - 1} + 1 = {\zeta_3}^2 r^{2^n} + \zeta_3 r^{2^n - 1} + 1 = {\zeta_3}^2 r^{2^n} + \zeta_3 r^{2^n} + \zeta_3 + 1.
\end{align*}
So $(\zeta_3 r)^{2^n} + (\zeta_3 r)^{2^n - 1} + 1 = 0$ implies
\begin{align*}
&\zeta_3 r^{2^n}(\zeta_3 + 1) + (\zeta_3 + 1) = 0 \\
\implies{}&r^{2^n} + {\zeta_3}^{2^n} = 0 \\
\implies{}&(r + \zeta_3)^{2^n} = 0.
\end{align*}
Thus $r = \zeta_3$. But then
\begin{align*}
({\zeta_3}^2 r)^{2^n} + ({\zeta_3}^2 r)^{2^n - 1} + 1 = ({\zeta_3}^2 \cdot \zeta_3)^{2^n} + ({\zeta_3}^2 \cdot \zeta_3)^{2^n - 1} + 1 = 1 + 1 + 1 = 1 \neq 0.
\end{align*}
Thus, in any case, there is no common root among the second factors. So $x$ is a greatest common divisor in the \cref{eqn:CC1DegeneracyGCDFormula}. Now \cref{eqn:CC1DegeneracyGCDFormula} gives $k_{\mathfrak{a}} = 2^{l + 2} - 2$.

Now suppose $L' = 4^n - 1$ for some $n \in \mathbb Z_{>0}$. We calculate that
\begin{align*}
&(x + 1)^{4^n} = x^{4^n} - 1 = (x - 1)(x^{4^n - 1} + x^{4^n - 2} + \dotsb + x + 1) \\
\implies{}&(x + 1)^{4^n - 1} + 1 = x^{4^n - 1} + x^{4^n - 2} + \dotsb + x^2 + x.
\end{align*}
Thus we have the factorization
\begin{align*}
(x + 1)^{4^n - 1} + 1 &= x(x^{4^n - 2} + x^{4^n - 3} + \dotsb + x + 1) \\
&= x(x^2 + x + 1)(x^{4^n - 4} + x^{4^n - 7} + \dotsb + x^3 + 1) \\
&= x(x + \zeta_3)(x + {\zeta_3}^2)(x^{4^n - 4} + x^{4^n - 7} + \dotsb + x^3 + 1).
\end{align*}
We apply the same equation for $(\zeta_3 x + 1)^{4^n - 1} + 1$ and $({\zeta_3}^2 x + 1)^{4^n - 1} + 1$ to get the factorizations
\begin{align*}
(\zeta_3 x + 1)^{4^n - 1} + 1 &= x(x + 1)(x + \zeta_3)(x^{4^n - 4} + x^{4^n - 7} + \dotsb + x^3 + 1) \\
({\zeta_3}^2 x + 1)^{4^n - 1} + 1 &= x(x + 1)(x + {\zeta_3}^2)(x^{4^n - 4} + x^{4^n - 7} + \dotsb + x^3 + 1).
\end{align*}
It is easy to see that $x(x^{4^n - 4} + x^{4^n - 7} + \dotsb + x^3 + 1)$ is a greatest common divisor in \cref{eqn:CC1DegeneracyGCDFormula} whose degree is $4^n - 3 = L' - 2$. Using this in \cref{eqn:CC1DegeneracyGCDFormula} gives $k_{\mathfrak{a}} = 4L - 2^{l + 3} - 2$.

Now suppose $L' = 2^{2n + 1} - 1$ for some $n \in \mathbb Z_{\geq 0}$. Suppose $r$ is a root of $(x + 1)^{2^{2n + 1} - 1} + 1$. We calculate that
\begin{align*}
&(x + 1)^{2^{2n + 1}} = x^{2^{2n + 1}} - 1 = (x - 1)(x^{2^{2n + 1} - 1} + x^{2^{2n + 1} - 2} + \dotsb + x + 1) \\
\implies{}&(x + 1)^{2^{2n + 1} - 1} + 1 = x^{2^{2n + 1} - 1} + x^{2^{2n + 1} - 2} + \dotsb + x^2 + x.
\end{align*}
Thus we have the factorization
\begin{align*}
(x + 1)^{2^{2n + 1} - 1} + 1 = x(x^{2^{2n + 1} - 2} + x^{2^{2n + 1} - 3} + \dotsb + x + 1).
\end{align*}
We apply the same equation for $(\zeta_3 x + 1)^{2^{2n + 1} - 1} + 1$ and $({\zeta_3}^2 x + 1)^{2^{2n + 1} - 1} + 1$ to get the factorizations
\begin{align*}
(\zeta_3 x + 1)^{2^{2n + 1} - 1} + 1 &= \zeta_3 x((\zeta_3 x)^{2^{2n + 1} - 2} + (\zeta_3 x)^{2^{2n + 1} - 3} + \dotsb + \zeta_3 x + 1) \\
({\zeta_3}^2 x + 1)^{2^{2n + 1} - 1} + 1 &= {\zeta_3}^2 x(({\zeta_3}^2 x)^{2^{2n + 1} - 2} + ({\zeta_3}^2 x)^{2^{2n + 1} - 3} + \dotsb + {\zeta_3}^2 x + 1).
\end{align*}
Suppose $r$ is a root of the second factor $x^{2^{2n + 1} - 2} + x^{2^{2n + 1} - 3} + \dotsb + x + 1$. By multiplying by $r - 1$, we find that $r^{2^{2n + 1} - 1} = 1$, i.e., $r$ is a $(2^{2n + 1} - 1)$\textsuperscript{th} root of unity. Similarly, so are $\zeta_3 r$ and ${\zeta_3}^2 r$. But then we have
\begin{align*}
(\zeta_3 r)^{2^{2n + 1} - 1} = 1 \implies (\zeta_3)^{2^{2n + 1} - 1} = 1
\end{align*}
which is a contradiction because $2^{2n + 1} - 1 \equiv 1 \pmod{3}$ for all $n \in \mathbb Z_{\geq 0}$. Thus $x$ is a greatest common divisor in \cref{eqn:CC1DegeneracyGCDFormula}. Now \cref{eqn:CC1DegeneracyGCDFormula} gives $k_{\mathfrak{a}} = 2^{l + 2} - 2$. We summarize the results in the following theorem.

\begin{theorem}
\label{thm:NumberOfEncodedQubitsCC1}
Consider CC$_1$ defined by the ideal $\mathfrak{a} = (x + y + z + 1, xy + yz + zx + 1)$. Imposing the periodic conditions $x^L - 1 = y^L - 1 = z^L - 1 = 0$ for some $L = 2^l L'$ where $2 \nmid L' \in \mathbb Z_{>0}$ and $l \in \mathbb Z_{\geq 0}$, the number of encoded qubits is given by
\begin{align*}
k_{\mathfrak{a}} = 2^{l + 2}\deg\left(\gcd((x + 1)^{L'} + 1, (\zeta_3 x + 1)^{L'} + 1, ({\zeta_3}^2 x + 1)^{L'} + 1)\right) - 2.
\end{align*}
Moreover, we have the explicit formulas
\begin{align*}
k_{\mathfrak{a}} =
\begin{cases}
4L - 2, & L = 2^l \\
2^{l + 2} - 2, & L = 2^l \cdot (2^n + 1) \\
4L - 2^{l + 3} - 2, & L = 2^l \cdot (4^n - 1) \\
2^{l + 2} - 2, & L = 2^l \cdot (2^{2n + 1} - 1).
\end{cases}
\end{align*}
\end{theorem}

\begin{remark}
\cref{thm:NumberOfEncodedQubitsCC1} implies that $k_{\mathfrak{a}}$ as a function of $L$, obeys the scaling relation
\begin{align*}
k_{\mathfrak{a}}(2^r L) = 2^r k_{\mathfrak{a}}(L) + 2(2^r - 1)
\end{align*}
for all $r \in \mathbb Z_{\geq 0}$.
\end{remark}

\subsection{Cubic code 6}
\label{subsec:CC6_degen}
The stabilizer generators of cubic code 6 (CC$_6$) are given by 
\begin{align}
\begin{array}{c}
\drawgenerator{XI}{II}{IX}{II}{XX}{XX}{XI}{IX}
\quad
\drawgenerator{ZZ}{ZZ}{IZ}{ZI}{IZ}{II}{ZI}{II}
\end{array}
\end{align}
Hence, the stabilizer ideal that defines CC$_6$ is $\mathfrak{a} = (f, g) = (x + y + z + 1, yz + zx + y + 1)$. We eliminate the variable $y$ to obtain
\begin{align*}
J_{\mathfrak{a}}^y = (h, x^{2^l} + a, z^{2^l} + c) = (z^2 + x, x^{2^l} + a, z^{2^l} + c).
\end{align*}
Recalling definitions, we know that $a = {x_0}^{2^l} + 1, c = {z_0}^{2^l} + 1$ where ${x_0}^{L'} = {z_0}^{L'} = 1$ such that $c^2 + a = 0$.

\begin{lemma}
\label{lem:CC$_6$GrobnerBasis}
For the ideal $J_{\mathfrak{a}}^y = (z^2 + x, x^{2^l} + a, z^{2^l} + c)$, we have that
\begin{align*}
G = \{z^2 + x, x^{2^{l - 1}} + c\}
\end{align*}
is a Gr\"{o}bner basis.
\end{lemma}

\begin{proof}
We assume the same hypotheses and use the same set $G$ as in the lemma. It is a straight forward calculation using the S-polynomials of Buchberger's algorithm to verify that $G$ is a Gr\"{o}bner basis for the ideal $(G)$. It is also a straight forward calculation by polynomial divisions by elements in $G$ to verify that $J_{\mathfrak{a}}^y \subset (G)$. It remains to show that $(G) \subset J_{\mathfrak{a}}^y$.

The only nontrivial containment we need to show is $x^{2^{l - 1}} + c \in J_{\mathfrak{a}}^y$, i.e., $x^{2^{l - 1}} + c \equiv 0 \pmod{J_{\mathfrak{a}}^y}$. We calculate that
\begin{align*}
x^{2^l} + a = (x^{2^{l - 1}})^2 + c^2 = (x^{2^{l - 1}} + c)^2.
\end{align*}
Hence $x^{2^l} + a \equiv 0 \pmod{J_{\mathfrak{a}}^y}$ implies $x^{2^{l - 1}} + c \equiv 0 \pmod{J_{\mathfrak{a}}^y}$ as desired.
\end{proof}

Now we simply read off the dimension using the Gr\"{o}bner basis and obtain the formula
\begin{align*}
k_{\mathfrak{a}, \mathfrak{m}} = 2\dim_{\mathbb F}(\mathbb F[x, y]/J_{\mathfrak{a}}^y) = 2^{l + 1}.
\end{align*}
We wish to calculate $k_{\mathfrak{a}} = \sum_{\mathfrak{m}} k_{\mathfrak{a}, \mathfrak{m}}$ where the sum is over all maximal ideals $\mathfrak{m} \subset \mathbb F[x, y, z]/I_{\mathfrak{a}}$. For the calculation, we first need to explicitly calculate the points in the variety $V(I_{\mathfrak{a}})$. By first finding solutions of $h = 0$, we easily calculate that the solutions of $f = g = 0$ is the set of points $V(\mathfrak{a}) = \{({z_0'}^2 + 1, {z_0'}^2 + z_0' + 1, z_0' + 1) \in \mathbb A_{\mathbb F}^3: z_0' \in \mathbb F\}$. Imposing the periodic conditions $x^L - 1 = y^L - 1 = z^L - 1 = 0$, we have the corresponding set of solutions $V(I_{\mathfrak{a}}) = \{(x_0, y_0, z_0) \in V(\mathfrak{a}): {x_0}^{L'} = {y_0}^{L'} = {z_0}^{L'} = 1\}$. Thus
\begin{align*}
\#V(I_{\mathfrak{a}}) = \deg\left(\gcd((z^2 + 1)^{L'} + 1, (z^2 + z + 1)^{L'} + 1, (z + 1)^{L'} + 1)\right).
\end{align*}
Since $(z^2 + 1)^{L'} + 1 = ((z + 1)^{L'} + 1)^2$, we can in fact simplify the equation to
\begin{align*}
\#V(I_{\mathfrak{a}}) = \deg\left(\gcd((z + 1)^{L'} + 1, (z^2 + z + 1)^{L'} + 1)\right).
\end{align*}
We put everything together using the formula
\begin{align*}
k_{\mathfrak{a}} = 2^{l + 1}\#V(I_{\mathfrak{a}}) = 2^{l + 1}\deg\left(\gcd((z + 1)^{L'} + 1, (z^2 + z + 1)^{L'} + 1)\right).
\end{align*}
In general, it is difficult to obtain a more explicit formula. We note however that calculating this for a single value of $L'$ can already provide a formula for an \textit{infinite family} of values of $L$, namely, for all $L \in \{2^lL' \in \mathbb Z_{>0}: l \in \mathbb Z_{\geq 0}\}$. It is not difficult to do this this explicitly by hand for small values of $L'$. We do this now for $L' = 1$ and $L' = 3$.

First suppose $L' = 1$. Then we have the factorizations
\begin{align*}
(z + 1) + 1 &= z\\
(z^2 + z + 1) + 1 &= z^2 + z = z(z + 1).
\end{align*}
So we calculate that
\begin{align*}
\#V(I_{\mathfrak{a}}) = \deg\left(\gcd((z + 1) + 1, (z^2 + z + 1) + 1)\right) = \deg(z) = 1.
\end{align*}
Thus $k_{\mathfrak{a}} = 2^{l + 1} = 2L$.

Now suppose $L' = 3$. Then we have the factorizations
\begin{align*}
(z + 1)^3 + 1 = z^3 + 3z^2 + 3z + 1 + 1 = z^3 + z^2 + z = z(z + \zeta_3)(z + {\zeta_3}^2)
\end{align*}
and
\begin{align*}
(z^2 + z + 1)^3 + 1 &= (z^2)^3 + 3(z^2)^2(z + 1) + 3z^2(z + 1)^2 + (z + 1)^3 + 1 \\
&= z^6 + z^5 + z^3 + z \\
&= z(z^5 + z^4 + z^2 + 1).
\end{align*}
So $\gcd((z + 1)^3 + 1, (z^2 + z + 1)^3 + 1) = z$ because we check that $\zeta_3$ and ${\zeta_3}^2$ are not roots of $z^5 + z^4 + z^2 + 1$ by putting $z = \zeta_3$ and $z = {\zeta_3}^2$. Again $\#V(I_{\mathfrak{a}}) = 1$. Thus $k_{\mathfrak{a}} = 2^{l + 1} = \frac{2}{3}L$.

We summarize the results in the following theorem.

\begin{theorem}
\label{thm:NumberOfEncodedQubitsCC6}
Consider CC$_6$ defined by the ideal $\mathfrak{a} = (x + y + z + 1, yz + zx + y + 1)$. Imposing the periodic conditions $x^L - 1 = y^L - 1 = z^L - 1 = 0$ for some $L = 2^l L'$ where $2 \nmid L' \in \mathbb Z_{>0}$ and $l \in \mathbb Z_{\geq 0}$, the number of encoded qubits is given by
\begin{align*}
k_{\mathfrak{a}} = 2^{l + 1}\deg\left(\gcd((z + 1)^{L'} + 1, (z^2 + z + 1)^{L'} + 1)\right).
\end{align*}
Moreover, we have the explicit formulas
\begin{align*}
k_{\mathfrak{a}} =
\begin{cases}
2L, & L = 2^l \\
\frac{2}{3}L, & L = 2^l \cdot 3.
\end{cases}
\end{align*}
\end{theorem}

\begin{remark}
\cref{thm:NumberOfEncodedQubitsCC6} implies that $k_{\mathfrak{a}}$ as a function of $L$, obeys the scaling relation
\begin{align*}
k_{\mathfrak{a}}(2^r L) = 2^r k_{\mathfrak{a}}(L) 
\end{align*}
for all $r \in \mathbb Z_{\geq 0}$.
\end{remark}

\subsection{Cubic code 11}
\label{subsec:CC11_degen}
The stabilizer generators of cubic code 11 (CC$_{11}$) are given by 
\begin{align}
\begin{array}{c}
\drawgenerator{IX}{II}{XI}{XX}{IX}{XX}{II}{XI}
\quad
\drawgenerator{ZI}{ZZ}{II}{IZ}{ZI}{II}{IZ}{ZZ}
\end{array}
\end{align}
Hence, the stabilizer ideal that defines CC$_{11}$ is $\mathfrak{a} = (f, g) = (yz + x + y + 1, xy + x + y + z)$. We eliminate the variable $x$ and use the substitutions $y \mapsto y + 1$ and $z \mapsto z + 1$ to obtain
\begin{align*}
J_{\mathfrak{a}}^x = (h, y^{2^l} + b, z^{2^l} + c) = (z(y^2 + y + 1), y^{2^l} + b, z^{2^l} + c).
\end{align*}
Recalling definitions, we know that $b = {y_0}^{2^l} + 1, c = {z_0}^{2^l} + 1$ where ${y_0}^{L'} = {z_0}^{L'} = 1$ such that $c(b^2 + b + 1) = 0$. The last equation is satisfied if and only if $b^2 + b + 1 = 0$ or $c = 0$.

\begin{lemma}
\label{lem:CC11GrobnerBasis}
For the ideal $J_{\mathfrak{a}}^x = (z(y^2 + y + 1), y^{2^l} + b, z^{2^l} + c)$, we have the following Gr\"{o}bner bases.
\begin{enumerate}
\item	Suppose $b^2 + b + 1 = 0$ and $c = 0$. Then we have the following cases.
\begin{enumerate}
\item\label{itm:CC11Eq1ZeroEq2Zero3NotDivide}	Suppose $3 \nmid 2^l + 1$. Then
\begin{align*}
G = \{yz + bz, y^{2^l} + b, z^{2^l}\}
\end{align*}
is a Gr\"{o}bner basis.
\item\label{itm:CC11Eq1ZeroEq2Zero3Divide}	Suppose $3 \mid 2^l + 1$. Then
\begin{align*}
G = \{yz + (b + 1)z, y^{2^l} + b, z^{2^l}\}
\end{align*}
is a Gr\"{o}bner basis.
\end{enumerate}
\item\label{itm:CC11Eq1ZeroEq2NotZero}	Suppose $b^2 + b + 1 = 0$ and $c \neq 0$. Then
\begin{align*}
G = \{y + (b + l), z^{2^l} + c\}
\end{align*}
is a Gr\"{o}bner basis.
\item\label{itm:CC11Eq1NotZeroEq2Zero}	Suppose $b^2 + b + 1 \neq 0$ and $c = 0$. Then
\begin{align*}
G = \{z, y^{2^l} + b\}
\end{align*}
is a Gr\"{o}bner basis.
\end{enumerate}
\end{lemma}

\begin{proof}
It is a straight forward calculation using the S-polynomials of Buchberger's algorithm to verify that $G$ is a Gr\"{o}bner basis for the ideal $(G)$ for all the cases. It is also a straight forward calculation by polynomial divisions by elements in $G$ to verify that $J_{\mathfrak{a}}^x \subset (G)$ for all the cases. It remains to show that $(G) \subset J_{\mathfrak{a}}^x$ for all the cases. This is shown below assuming the same hypotheses and using the same set $G$ as in the lemma for the corresponding cases.

\medskip
\noindent
\textit{\cref{itm:CC11Eq1ZeroEq2Zero3NotDivide}.} The only nontrivial containment we need to show is $yz + bz \in J_{\mathfrak{a}}^x$, i.e., $yz + bz \equiv 0 \pmod{J_{\mathfrak{a}}^x}$. We calculate that
\begin{align*}
&(by - 1)\left((by)^{2^l - 1} + (by)^{2^l - 2} + \dotsb + by + 1\right) \\
\equiv{}&(by)^{2^l} - 1 \pmod{J_{\mathfrak{a}}^x} \\
\equiv{}&b^{2^l} \cdot y^{2^l} - 1 \pmod{J_{\mathfrak{a}}^x} \\
\equiv{}&b^{2^l + 1} + 1 \pmod{J_{\mathfrak{a}}^x}.
\end{align*}
The hypothesis $b^2 + b + 1 = 0$ implies $b \neq 1$ with $b^3 = 1$. Hence the hypothesis $3 \nmid 2^l + 1$ implies $b^{2^l + 1} + 1 \neq 0$. Let $P = (by)^{2^l - 1} + (by)^{2^l - 2} + \dotsb + by + 1$. Realizing that $z(y + b)(y + (b + 1)) = z(y^2 + y + 1)$, we multiply by $\frac{b}{b^{2^l + 1} + 1}P$ to get
\begin{align*}
&z(y + b)(y + (b + 1)) \cdot \frac{b}{b^{2^l + 1} + 1}P \equiv 0 \pmod{J_{\mathfrak{a}}^x} \\
\implies{}&z(y + b) \cdot \frac{(by - 1)P}{b^{2^l + 1} + 1} \equiv 0 \pmod{J_{\mathfrak{a}}^x} \\
\implies{}&yz + bz \equiv 0 \pmod{J_{\mathfrak{a}}^x}.
\end{align*}

\medskip
\noindent
\textit{\cref{itm:CC11Eq1ZeroEq2Zero3Divide}.} The only nontrivial containment we need to show is $yz + (b + 1)z \in J_{\mathfrak{a}}^x$, i.e., $yz + (b + 1)z \equiv 0 \pmod{J_{\mathfrak{a}}^x}$. By a similar calculation as in the proof of \cref{itm:CC11Eq1ZeroEq2Zero3NotDivide}, we have
\begin{align*}
&((b + 1)y - 1)\left(((b + 1)y)^{2^l - 1} + ((b + 1)y)^{2^l - 2} + \dotsb + (b + 1)y + 1\right) \\
\equiv{}&(b + 1)^{2^l + 1} + 1 \pmod{J_{\mathfrak{a}}^x}.
\end{align*}
This time, the hypothesis $3 \mid 2^l + 1$ implies $b^{2^l + 1} + 1 = 0$. Hence
\begin{align*}
(b + 1)^{2^l + 1} + 1 &= (b + 1)(b + 1)^{2^l} = (b + 1)(b^{2^l} + 1) = b^{2^l + 1} + b^{2^l} + b + 1 \\
&= b^{2^l} + b = \frac{b^{2^l + 1} + b^2}{b} = \frac{b^2 + 1}{b} = \frac{b}{b} = 1.
\end{align*}
Let $P = ((b + 1)y)^{2^l - 1} + ((b + 1)y)^{2^l - 2} + \dotsb + (b + 1)y + 1$. Then the above implies $((b + 1)y - 1)P \equiv 1 \pmod{J_{\mathfrak{a}}^x}$. Again using $z(y + (b + 1))(y + b) = z(y^2 + y + 1)$, we multiply by $(b + 1)P$ to get
\begin{align*}
&z(y + (b + 1))(y + b) \cdot (b + 1)P \equiv 0 \pmod{J_{\mathfrak{a}}^x} \\
\implies{}&z(y + (b + 1)) \cdot ((b + 1)y - 1)P \equiv 0 \pmod{J_{\mathfrak{a}}^x} \\
\implies{}&yz + (b + 1)z \equiv 0 \pmod{J_{\mathfrak{a}}^x}.
\end{align*}

\medskip
\noindent
\textit{\cref{itm:CC11Eq1ZeroEq2NotZero}.} The only nontrivial containment we need to show is $y + (b + l) \in J_{\mathfrak{a}}^x$, i.e., $y + (b + l) \equiv 0 \pmod{J_{\mathfrak{a}}^x}$. This follows if we show our claim $y^{2^{l - n}} + (b + n) \equiv 0 \pmod{J_{\mathfrak{a}}^x}$ for all integers $0 \leq n \leq l$ by taking the $n = l$ case. We show this by induction. The base case $n = 0$ is trivial. Suppose that the claim holds for some integer $0 \leq n - 1 \leq l - 1$, i.e., $y^{2^{l - (n - 1)}} + (b + (n - 1)) \equiv 0 \pmod{J_{\mathfrak{a}}^x}$. Then
\begin{align*}
&z^{2^l - 2^{l - n}} \cdot (z(y^2 + y + 1))^{2^{l - n}} + z^{2^l}\left(y^{2^{l - (n - 1)}} + (b + (n - 1))\right) \equiv 0 \pmod{J_{\mathfrak{a}}^x} \\
\implies{}&z^{2^l}\left(y^{2^{l - n + 1}} + y^{2^{l - n}} + 1\right) + z^{2^l}\left(y^{2^{l - n + 1}} + (b + (n - 1))\right) \equiv 0 \pmod{J_{\mathfrak{a}}^x} \\
\implies{}&z^{2^l}(y^{2^{l - n}} + (b + n)) \equiv 0 \pmod{J_{\mathfrak{a}}^x} \\
\implies{}&c(y^{2^{l - n}} + (b + n)) \equiv 0 \pmod{J_{\mathfrak{a}}^x}.
\end{align*}
The hypothesis $c \neq 0$ implies $y^{2^{l - n}} + (b + n) \equiv 0 \pmod{J_{\mathfrak{a}}^x}$ as desired.

\medskip
\noindent
\textit{\cref{itm:CC11Eq1NotZeroEq2Zero}.} 
The only nontrivial containment we need to show is $z \in J_{\mathfrak{a}}^x$, i.e., $z \equiv 0 \pmod{J_{\mathfrak{a}}^x}$. We have
\begin{align*}
&((y^2 + y) - 1)\left((y^2 + y)^{2^l - 1} + (y^2 + y)^{2^l - 2} + \dotsb + (y^2 + y) + 1\right) \\
\equiv{}&(y^2 + y)^{2^l} - 1 \pmod{J_{\mathfrak{a}}^x} \\
\equiv{}&(y^{2^l})^2 + y^{2^l} + 1 \pmod{J_{\mathfrak{a}}^x} \\
\equiv{}&b^2 + b + 1 \pmod{J_{\mathfrak{a}}^x}.
\end{align*}
Now $b^2 + b + 1 \neq 0$ by hypothesis. Let $P = (y^2 + y)^{2^l - 1} + (y^2 + y)^{2^l - 2} + \dotsb + (y^2 + y) + 1$. Then the above implies $(y^2 + y + 1) \cdot \frac{P}{b^2 + b + 1} \equiv 1 \pmod{J_{\mathfrak{a}}^x}$. Thus
\begin{align*}
z \equiv z(y^2 + y + 1) \cdot \frac{P}{b^2 + b + 1} \equiv 0 \pmod{J_{\mathfrak{a}}^x}.
\end{align*}
\end{proof}
Now we simply read off the dimensions using the Gr\"{o}bner bases and obtain the formula
\begin{align*}
\dim_{\mathbb F}(\mathbb F[x, y]/J_{\mathfrak{a}}^x)
&=
\begin{cases}
2^{l + 1} - 1, & b^2 + b + 1 = 0 \text{ and } c = 0 \\
2^l, & b^2 + b + 1 = 0 \text{ and } c \neq 0 \\
2^l, & b^2 + b + 1 \neq 0 \text{ and } c = 0.
\end{cases}
\end{align*}
Recalling the definitions, the above gives the formula
\begin{align*}
k_{\mathfrak{a}, \mathfrak{m}}
&=
\begin{cases}
2^{l + 2} - 2, & y_0 \in \{\zeta_3, {\zeta_3}^2\} \text{ and } z_0 = 1 \\
2^{l + 1}, & y_0 \in \{\zeta_3, {\zeta_3}^2\} \text{ and } z_0 \neq 1 \\
2^{l + 1}, & y_0 \notin \{\zeta_3, {\zeta_3}^2\} \text{ and } z_0 = 1.
\end{cases}
\end{align*}
Note that this holds if $l = 0$ as well. Now we wish to calculate $k_{\mathfrak{a}} = \sum_{\mathfrak{m}} k_{\mathfrak{a}, \mathfrak{m}}$ where the sum is over all maximal ideals $\mathfrak{m} \subset \mathbb F[x, y, z]/I_{\mathfrak{a}}$. For the calculation, we first need to explicitly calculate the points in the variety $V(I_{\mathfrak{a}})$. By first finding solutions of $h = 0$, we easily calculate that the solutions of $f = g = 0$ is the set of points $V(\mathfrak{a}) = V_1(\mathfrak{a}) \cup V_2(\mathfrak{a}) \cup V_3(\mathfrak{a})$ where
\begin{align*}
V_1(\mathfrak{a}) &= \{(1, y_0, 1) \in \mathbb A_{\mathbb F}^3: y_0 \in \mathbb F\} \\
V_2(\mathfrak{a}) &= \{(\zeta_3 z_0' + 1, \zeta_3, z_0' + 1) \in \mathbb A_{\mathbb F}^3: z_0' \in \mathbb F\} \\
V_3(\mathfrak{a}) &= \{({\zeta_3}^2 z_0' + 1, {\zeta_3}^2, z_0' + 1) \in \mathbb A_{\mathbb F}^3: z_0' \in \mathbb F\}.
\end{align*}
Note that the intersection is
\begin{align*}
V_{\cap}(\mathfrak{a}) = V_1(\mathfrak{a}) \cap V_2(\mathfrak{a}) \cap V_3(\mathfrak{a}) = \{(1, \zeta_3, 1), (1, {\zeta_3}^2, 1)\}
\end{align*}
which consists of the only points satisfying both $y_0 \in \{\zeta_3, {\zeta_3}^2\}$ and $z_0 = 1$. Imposing the periodic conditions $x^L - 1 = y^L - 1 = z^L - 1 = 0$, we have the corresponding set of solutions $V(I_{\mathfrak{a}}) = V_1(I_{\mathfrak{a}}) \cup V_2(I_{\mathfrak{a}}) \cup V_3(I_{\mathfrak{a}})$ where $V_j(I_{\mathfrak{a}}) = \{(x_0, y_0, z_0) \in V_j(\mathfrak{a}): {x_0}^{L'} = {y_0}^{L'} = {z_0}^{L'} = 1\}$ for all $j \in \{1, 2, 3\}$. The intersection is
\begin{align*}
V_{\cap}(I_{\mathfrak{a}}) = V_1(I_{\mathfrak{a}}) \cap V_2(I_{\mathfrak{a}}) \cap V_3(I_{\mathfrak{a}}) =
\begin{cases}
\varnothing, & 3 \nmid L' \\
\{(1, \zeta_3, 1), (1, {\zeta_3}^2, 1)\}, & 3 \mid L'.
\end{cases}
\end{align*}
Thus we immediately calculate
\begin{align*}
\#V_{\cap}(I_{\mathfrak{a}}) &=
\begin{cases}
0, & 3 \nmid L' \\
2, & 3 \mid L'
\end{cases}\\
\#(V_1(I_{\mathfrak{a}}) \setminus V_{\cap}(I_{\mathfrak{a}})) &=
\begin{cases}
L', & 3 \nmid L' \\
L' - 2, & 3 \mid L'
\end{cases}\\
\#(V_2(I_{\mathfrak{a}}) \setminus V_{\cap}(I_{\mathfrak{a}})) &=
\begin{cases}
0, & 3 \nmid L' \\
\deg\left(\gcd((z + 1)^{L'} + 1, (\zeta_3 z + 1)^{L'} + 1)\right) - 1, & 3 \mid L'
\end{cases}\\
\#(V_3(I_{\mathfrak{a}}) \setminus V_{\cap}(I_{\mathfrak{a}})) &=
\begin{cases}
0, & 3 \nmid L' \\
\deg\left(\gcd((z + 1)^{L'} + 1, ({\zeta_3}^2 z + 1)^{L'} + 1)\right) - 1, & 3 \mid L'
\end{cases}.
\end{align*}
Note that $\#(V_2(I_{\mathfrak{a}}) \setminus V_{\cap}(I_{\mathfrak{a}}))$ and $\#(V_3(I_{\mathfrak{a}}) \setminus V_{\cap}(I_{\mathfrak{a}}))$ are in fact equal in all cases using any extension $\tilde{\sigma} \in \Gal(\mathbb F/\mathbb F_2)$ of $\sigma \in \Gal(\mathbb F_{2^2}/\mathbb F_2)$ specified by $\sigma(\zeta_3) = {\zeta_3}^2$, noting that $\mathbb F_2(\zeta_3) \cong \mathbb F_{2^2}$. We put everything together using the formula
\begin{align}
\begin{split}
k_{\mathfrak{a}} ={}&(2^{l + 2} - 2)\#V_{\cap}(I_{\mathfrak{a}}) + 2^{l + 1}\#(V_1(I_{\mathfrak{a}}) \setminus V_{\cap}(I_{\mathfrak{a}})) + 2^{l + 1}\#(V_2(I_{\mathfrak{a}}) \setminus V_{\cap}(I_{\mathfrak{a}})) \\
&{}+ 2^{l + 1}\#(V_3(I_{\mathfrak{a}}) \setminus V_{\cap}(I_{\mathfrak{a}}))
\end{split}\\
\label{eqn:CC11DegeneracyGCDFormula}
={}&
\begin{cases}
2L, & 3 \nmid L \\
2L - 4 + 2^{l + 2}\deg\left(\gcd((z + 1)^{L'} + 1, (\zeta_3 z + 1)^{L'} + 1)\right), & 3 \mid L.
\end{cases}
\end{align}
In general, it is difficult to obtain a more explicit formula for the case $3 \mid L'$. However, it is possible for a special case.

Suppose $3 \mid L' = 4^n - 1$ for some $n \in \mathbb Z_{>0}$. We calculate that
\begin{align*}
&(z + 1)^{4^n} = z^{4^n} - 1 = (z - 1)(z^{4^n - 1} + z^{4^n - 2} + \dotsb + z + 1) \\
\implies{}&(z + 1)^{4^n - 1} + 1 = z^{4^n - 1} + z^{4^n - 2} + \dotsb + z^2 + z.
\end{align*}
Thus we have the factorization
\begin{align*}
(z + 1)^{4^n - 1} + 1 &= z(z^{4^n - 2} + z^{4^n - 3} + \dotsb + z + 1) \\
&= z(z^2 + z + 1)(z^{4^n - 4} + z^{4^n - 7} + \dotsb + z^3 + 1) \\
&= z(z + \zeta_3)(z + {\zeta_3}^2)(z^{4^n - 4} + z^{4^n - 7} + \dotsb + z^3 + 1).
\end{align*}
We apply the same equation for $(\zeta_3 z + 1)^{4^n - 1} + 1$ and use ${\zeta_3}^3 = 1$ to get the factorization
\begin{align*}
(\zeta_3 z + 1)^{4^n - 1} + 1 = z(z + 1)(z + \zeta_3)(z^{4^n - 4} + z^{4^n - 7} + \dotsb + z^3 + 1).
\end{align*}
It is easy to see that $z(z + \zeta_3)(z^{4^n - 4} + z^{4^n - 7} + \dotsb + z^3 + 1)$ is a greatest common divisor in \cref{eqn:CC11DegeneracyGCDFormula} whose degree is $4^n - 2 = L' - 1$. Using this in \cref{eqn:CC11DegeneracyGCDFormula} gives $k_{\mathfrak{a}} = 6L - 4(2^l + 1)$.

We summarize the results in the following theorem.

\begin{theorem}
\label{thm:NumberOfEncodedQubitsCC11}
Consider CC$_{11}$ defined by the ideal $\mathfrak{a} = (yz + x + y + 1, xy + x + y + z)$. Imposing the periodic conditions $x^L - 1 = y^L - 1 = z^L - 1 = 0$ for some $L = 2^l L'$ where $2 \nmid L' \in \mathbb Z_{>0}$ and $l \in \mathbb Z_{\geq 0}$, the number of encoded qubits is given by
\begin{align*}
k_{\mathfrak{a}} =
\begin{cases}
2L, & 3 \nmid L \\
2L - 4 + 2^{l + 2}\deg\left(\gcd((z + 1)^{L'} + 1, (\zeta_3 z + 1)^{L'} + 1)\right), & 3 \mid L.
\end{cases}
\end{align*}
Moreover, if $3 \mid L' = 4^n - 1$ for some $n \in \mathbb Z_{>0}$, then we have the explicit formula
\begin{align*}
k_{\mathfrak{a}} = 6L - 4(2^l + 1).
\end{align*}
\end{theorem}

\begin{remark}
\cref{thm:NumberOfEncodedQubitsCC11} implies that $k_{\mathfrak{a}}$ as a function of $L$, obeys the scaling relation
\begin{align*}
k_{\mathfrak{a}}(2^r L) =
\begin{cases}
2^r k_{\mathfrak{a}}(L), & 3 \nmid L \\
2^r k_{\mathfrak{a}}(L) + 4(2^r - 1), & 3 \mid L.
\end{cases}
\end{align*}
for all $r \in \mathbb Z_{\geq 0}$.
\end{remark}

\subsection{Cubic code 11B}
The stabilizer generators of cubic code 11B ($\text{CCB}_{11}$) are given by 
\begin{align}
\begin{array}{c}
\xymatrix@!0{%
&  \ar@{-}[rr] && \ar@{-}[rr] && IX \ar@{-}[dl]\\
XI \ar@{-}[ur] && XX \ar@{-}[ll]\ar@{-}[ur]\ar@{-}[rr] && XX\\
&  \ar@{.}[uu]\ar@{.}[dl]\ar@{.}[rr] && \ar@{.}[rr] \ar@{.}[uu] &&  \ar@{-}[uu]\\
\ar@{-}[uu]\ar@{-}[rr] && IX \ar@{-}[uu]\ar@{.}[ur] \ar@{-}[rr] && \ar@{-}[uu] \ar@{-}[ur]
}
\quad
\xymatrix@!0{%
&  \ar@{-}[rr] && ZI \ar@{-}[rr] &&  \ar@{-}[dl] \\
   \ar@{-}[ur] &&  \ar@{-}[ll]\ar@{-}[ur] && \ar@{-}[ll]  \\
&  ZZ\ar@{.}[uu]\ar@{.}[dl]\ar@{.}[rr] && ZZ\ar@{.}[rr] \ar@{.}[uu] && IZ \ar@{-}[uu]  \\
  ZI \ar@{-}[uu]\ar@{-}[rr] &&  \ar@{-}[uu]\ar@{.}[ur] \ar@{-}[rr] && \ar@{-}[uu] \ar@{-}[ur]
}
\end{array}
\end{align}
Hence, the stabilizer ideal that defines $\text{CCB}_{11}$ is given by $\mathfrak{a} = (f, g) = (1 + y + y^2,1 + x + y + yz)$. We eliminate the variable $x$ to obtain
\begin{align*}
J_{\mathfrak{a}}^x = (h, y^{2^l} + b_0, z^{2^l} + c_0) = (y^2 + y + 1, y^{2^l} + b_0, z^{2^l} + c_0).
\end{align*}
Recalling definitions, we know that $b_0 = {y_0}^{2^l}$, $c_0 = {z_0}^{2^l}$ where ${y_0}^{L'} = {z_0}^{L'} = 1$ such that ${b_0}^2 + b_0 + 1 = 0$. Notice that we have not made the substitutions $y \mapsto y + 1,  z \mapsto z + 1$ in this example as the ideal is simple enough as it is. 

\begin{lemma}
\label{lem:CCB11GrobnerBasis}
For the ideal $J_{\mathfrak{a}}^x = (y^2 + y + 1, y^{2^l} + b_0, z^{2^l} + c_0)$, we have the following Gr\"{o}bner bases.
\begin{enumerate}
\item\label{itm:CCB111}	Suppose $3 \nmid 2^l + 1$. Then
\begin{align*}
G = \{y + b_0, z^{2^l} + c_0\}
\end{align*}
is a Gr\"{o}bner basis.
\item\label{itm:CCB112}	Suppose $3 \mid 2^l + 1$. Then
\begin{align*}
G = \{y + b_0 + 1, z^{2^l} + c_0\}
\end{align*}
is a Gr\"{o}bner basis.
\end{enumerate}
\end{lemma}

\begin{proof}
It is a straight forward calculation using the S-polynomials of Buchberger's algorithm to verify that $G$ is a Gr\"{o}bner basis for the ideal $(G)$ for all the cases. It is also a straight forward calculation by polynomial divisions by elements in $G$ to verify that $J_{\mathfrak{a}}^x \subset (G)$ for all the cases. It remains to show that $(G) \subset J_{\mathfrak{a}}^x$ for all the cases. This is shown below assuming the same hypotheses and using the same set $G$ as in the lemma for the corresponding cases.

\medskip
\noindent
\textit{\cref{itm:CCB111}.} The only nontrivial containment we need to show is $y + b_0 \in J_{\mathfrak{a}}^x$, i.e., $y + b_0 \equiv 0 \pmod{J_{\mathfrak{a}}^x}$. We calculate that
\begin{align*}
&(b_0y - 1)\left((b_0y)^{2^l - 1} + (b_0y)^{2^l - 2} + \dotsb + b_0y + 1\right) \\
\equiv{}&(b_0y)^{2^l} - 1 \pmod{J_{\mathfrak{a}}^x} \\
\equiv{}&{b_0}^{2^l} \cdot y^{2^l} - 1 \pmod{J_{\mathfrak{a}}^x} \\
\equiv{}&{b_0}^{2^l + 1} + 1 \pmod{J_{\mathfrak{a}}^x}.
\end{align*}
The hypothesis ${b_0}^2 + b_0 + 1 = 0$ implies $b_0 \neq 1$ with ${b_0}^3 = 1$. Hence the hypothesis $3 \nmid 2^l + 1$ implies ${b_0}^{2^l + 1} + 1 \neq 0$. Let $P = (b_0y)^{2^l - 1} + (b_0y)^{2^l - 2} + \dotsb + b_0y + 1$. Realizing that $(y + b_0)(y + (b_0 + 1)) = (y^2 + y + 1)$, we multiply by $\frac{b_0}{{b_0}^{2^l + 1} + 1}P$ to get
\begin{align*}
&(y + b_0)(y + (b_0 + 1)) \cdot \frac{b_0}{{b_0}^{2^l + 1} + 1}P \equiv 0 \pmod{J_{\mathfrak{a}}^x} \\
\implies{}&(y + b_0) \cdot \frac{(b_0y - 1)P}{{b_0}^{2^l + 1} + 1} \equiv 0 \pmod{J_{\mathfrak{a}}^x} \\
\implies{}&y + b_0 \equiv 0 \pmod{J_{\mathfrak{a}}^x}.
\end{align*}

\medskip
\noindent
\textit{\cref{itm:CCB112}.} The only nontrivial containment we need to show is $y + (b_0 + 1) \in J_{\mathfrak{a}}^x$, i.e., $y + (b_0 + 1) \equiv 0 \pmod{J_{\mathfrak{a}}^x}$. By a similar calculation as in the proof of \cref{itm:CCB111}, we have
\begin{align*}
&((b_0 + 1)y - 1)\left(((b_0 + 1)y)^{2^l - 1} + ((b_0 + 1)y)^{2^l - 2} + \dotsb + (b_0 + 1)y + 1\right) \\
\equiv{}&(b_0 + 1)^{2^l + 1} + 1 \pmod{J_{\mathfrak{a}}^x}.
\end{align*}
This time, the hypothesis $3 \mid 2^l + 1$ implies ${b_0}^{2^l + 1} + 1 = 0$. Hence
\begin{align*}
(b_0 + 1)^{2^l + 1} + 1 &= (b_0 + 1)(b_0 + 1)^{2^l} = (b_0 + 1)({b_0}^{2^l} + 1) = {b_0}^{2^l + 1} + {b_0}^{2^l} + b_0 + 1 \\
&= {b_0}^{2^l} + b_0 = \frac{{b_0}^{2^l + 1} + {b_0}^2}{b_0} = \frac{{b_0}^2 + 1}{b_0} = \frac{b_0}{b_0} = 1.
\end{align*}
Let $P = ((b_0 + 1)y)^{2^l - 1} + ((b_0 + 1)y)^{2^l - 2} + \dotsb + (b_0 + 1)y + 1$. Then the above implies $((b_0 + 1)y - 1)P \equiv 1 \pmod{J_{\mathfrak{a}}^x}$. Again using $(y + (b_0 + 1))(y + b_0) = (y^2 + y + 1)$, we multiply by $(b_0 + 1)P$ to get
\begin{align*}
&(y + (b_0 + 1))(y + b_0) \cdot (b_0 + 1)P \equiv 0 \pmod{J_{\mathfrak{a}}^x} \\
\implies{}&(y + (b_0 + 1)) \cdot ((b_0 + 1)y - 1)P \equiv 0 \pmod{J_{\mathfrak{a}}^x} \\
\implies{}&y + (b_0 + 1) \equiv 0 \pmod{J_{\mathfrak{a}}^x}.
\end{align*}
\end{proof}

Now we simply read off the dimensions using the Gr\"{o}bner bases and obtain the formula
\begin{align*}
\dim_{\mathbb F}(\mathbb F[x, y]/J_{\mathfrak{a}}^x)
&=
2^l
\end{align*}
Recalling the definitions, the above gives the formula
\begin{align*}
k_{\mathfrak{a}, \mathfrak{m}}
&=
2^{l + 1}
\end{align*}

Now we wish to calculate $k_{\mathfrak{a}} = \sum_{\mathfrak{m}} k_{\mathfrak{a}, \mathfrak{m}}$ where the sum is over all maximal ideals $\mathfrak{m} \subset \mathbb F[x, y, z]/I_{\mathfrak{a}}$. For the calculation, we first need to explicitly calculate the points in the variety $V(I_{\mathfrak{a}})$. By first finding solutions of $h = 0$, we easily calculate that the solutions of $f = g = 0$ is the set of points $V(\mathfrak{a}) = V_1(\mathfrak{a}) \cup V_2(\mathfrak{a})$ where
\begin{align*}
V_1(\mathfrak{a}) &= \{(\zeta_3 z_0' + 1, \zeta_3, z_0' + 1) \in \mathbb A_{\mathbb F}^3: z_0' \in \mathbb F\} \\
V_2(\mathfrak{a}) &= \{({\zeta_3}^2 z_0' + 1, {\zeta_3}^2, z_0' + 1) \in \mathbb A_{\mathbb F}^3: z_0' \in \mathbb F\}.
\end{align*}

Note that the intersection is $V_{\cap}(\mathfrak{a}) = V_1(\mathfrak{a}) \cap V_2(\mathfrak{a}) = \varnothing$. Imposing the periodic conditions $x^L - 1 = y^L - 1 = z^L - 1 = 0$, we have the corresponding set of solutions $V(I_{\mathfrak{a}}) = V_1(I_{\mathfrak{a}}) \cup V_2(I_{\mathfrak{a}})$ where $V_j(I_{\mathfrak{a}}) = \{(x_0, y_0, z_0) \in V_j(\mathfrak{a}): {x_0}^{L'} = {y_0}^{L'} = {z_0}^{L'} = 1\}$ for all $j \in \{1, 2\}$. The intersection is $V_{\cap}(I_{\mathfrak{a}}) = \varnothing$. Thus, $\#V_{\cap}(I_{\mathfrak{a}}) = 0$ and
\begin{align*}
\#(V_1(I_{\mathfrak{a}}) \setminus V_{\cap}(I_{\mathfrak{a}})) = \#(V_2(I_{\mathfrak{a}}) \setminus V_{\cap}(I_{\mathfrak{a}})) = \begin{cases}
0, & 3 \nmid L' \\
\deg\left(\gcd((z + 1)^{L'} + 1, (\zeta_3 z + 1)^{L'} + 1)\right), & 3 \mid L'.
\end{cases}
\end{align*}
Recalling arguments from \cref{subsec:CC11_degen}, we put everything together using the formula
\begin{align*}
k_{\mathfrak{a}} &= 2^{l + 1}\#V_{\cap}(I_{\mathfrak{a}}) + 2^{l + 1}\#(V_1(I_{\mathfrak{a}}) \setminus V_{\cap}(I_{\mathfrak{a}})) + 2^{l + 1}\#(V_2(I_{\mathfrak{a}}) \setminus V_{\cap}(I_{\mathfrak{a}})) \\
&= 
\begin{cases}
0, & 3 \nmid L' \\
2^{l + 2}\deg\left(\gcd((z + 1)^{L'} + 1, (\zeta_3 z + 1)^{L'} + 1)\right), & 3 \mid L'.
\end{cases}
\end{align*}
We also recall from \cref{subsec:CC11_degen} that for the special case $3 \mid L' = 4^n - 1$ for some $n \in \mathbb Z_{>0}$, we have the explicit calculation
\begin{align*}
\deg\left(\gcd((z + 1)^{L'} + 1, (\zeta_3 z + 1)^{L'} + 1)\right) = 4^n - 2 = L' - 1.
\end{align*}

We summarize the results in the following theorem.

\begin{theorem}
\label{thm:NumberOfEncodedQubitsCCB11}
Consider $\text{CCB}_{11}$ defined by the ideal $\mathfrak{a} = (1 + y + y^2,1 + x + y + yz)$. Imposing the periodic conditions $x^L - 1 = y^L - 1 = z^L - 1 = 0$ for some $L = 2^l L'$ where $2 \nmid L' \in \mathbb Z_{>0}$ and $l \in \mathbb Z_{\geq 0}$, the number of encoded qubits is given by
\begin{align*}
k_{\mathfrak{a}} =
\begin{cases}
0, & 3 \nmid L \\
2^{l + 2}\deg\left(\gcd((z + 1)^{L'} + 1, (\zeta_3 z + 1)^{L'} + 1)\right), & 3 \mid L.
\end{cases}
\end{align*}
Moreover, if $3 \mid L' = 4^n - 1$ for some $n \in \mathbb Z_{>0}$, then we have the explicit formula
\begin{align*}
k_{\mathfrak{a}} = 4(L - 2^l).
\end{align*}
\end{theorem}

\begin{remark}
\cref{thm:NumberOfEncodedQubitsCCB11} implies that $k_{\mathfrak{a}}$ as a function of $L$, obeys the scaling relation
\begin{align*}
k_{\mathfrak{a}}(2^rL)= 2^rk_{\mathfrak{a}}(L)
\end{align*}
for all $r \in \mathbb Z_{\geq 0}$. This is consistent with the self-bifurcating behavior of $\text{CCB}_{11}$.
\end{remark}

\section{Charge annihilators}
% \subsection{Definition and Result}
In our argument for choosing the coarse-graining factor to be 2, we argued how the charge annihilator or the set of trivial charges shows self-reproducing behavior under coarse-graining by a factor of 2. To be precise, the charge annihilator given by the stabilizer ideal $\langle f,g \rangle$ changes to $\langle f^2,g^2 \rangle$ after coarse-graining. In the {\small{MATHEMATICA}} file SMERG.nb, we show the calculation for all the cubic codes. In fact, for all of them the annihilator after coarse-graining is given by $\av{f^2,g^2}$ where $f$ and $g$ are polynomials that give the stabilizer ideal as well as the charge annihilator of the original model. We now explain how we use elimination theory with Gr\"{o}bner bases to compute the annihilator of the charge module after coarse-graining. The charge annihilator after coarse-graining is defined as the intersection of the original annihilator and the coarse-grained Laurent polynomial ring, i.e., $\av{f,g}\cap \mathbb{F}_2[x^{\pm 2},y^{\pm 2},z^{\pm 2}]$.

Let $R = \mathbb{F}_2[x, y, z]$ and $R' = \mathbb{F}_2[x^2, y^2, z^2] \subset R$ be a subring. Consider an ideal $I = \langle f, g \rangle \subset R$ for some $f, g \in R$. We wish to compute the annihilator $\Ann_{R'}(R/I)$ of the left $R'$-module $R/I$. Since $\Ann_{R'}(R/I) = I \cap R'$, we develop a method to compute $I \cap R'$.

Let $\tilde{R} = \mathbb{F}_2[x, y, z, a, b, c]$ and $\tilde{R}' = \mathbb{F}_2[a, b, c] \subset \tilde{R}$ be a subring. Consider the ideal $\tilde{I} = I + \langle x^2 - a, y^2 - b, z^2 - c \rangle \subset \tilde{R}$. Then elimination theory directly provides an algorithm to compute $\tilde{I} \cap \tilde{R}'$ using a Gr\"{o}bner basis according to \cite[Theorem 2.3.4]{AL94}. We now focus on arguing that computing $\tilde{I} \cap \tilde{R}'$ is essentially the same as computing $I \cap R'$.

Define the surjective homomorphism $\phi: \tilde{R} \to R$ uniquely determined by the mappings
\begin{align*}
x &\mapsto x & a &\mapsto x^2 \\
y &\mapsto y & b &\mapsto y^2 \\
z &\mapsto z & c &\mapsto z^2.
\end{align*}
Note that $\phi(\tilde{R}') = R'$ and $\phi(\tilde{I}) = I$. A simple set theoretic calculation gives $\phi(\tilde{I} \cap \tilde{R}') \subset \phi(\tilde{I}) \cap \phi(\tilde{R}') = I \cap R'$. However, the reverse containment does not hold in general and this is the nontriviality which is to be shown.

\begin{theorem}
We have $I \cap R' = \phi(\tilde{I} \cap \tilde{R}')$. Hence, viewing $I$ as an ideal generated in $\mathbb{F}_2[x^{\pm 1}, y^{\pm 1}, z^{\pm 1}]$, the charge annihilator after coarse-graining can be calculated by $I \cap \mathbb{F}_2[x^{\pm 2}, y^{\pm 2}, z^{\pm 2}] = \langle \phi(\tilde{I} \cap \tilde{R}')\rangle \subset \mathbb{F}_2[x^{\pm 2}, y^{\pm 2}, z^{\pm 2}]$.
\end{theorem}

\begin{proof}
First we show the useful fact $\ker(\phi) = \langle x^2 - a, y^2 - b, z^2 - c \rangle$. Define the surjective homomorphism $\phi_a: \mathbb{F}_2[x, y, z, a, b, c] \to \mathbb{F}_2[x, y, z, b, c]$ uniquely determined by the mappings
\begin{align*}
x &\mapsto x & a &\mapsto x^2 \\
y &\mapsto y & b &\mapsto b \\
z &\mapsto z & c &\mapsto c.
\end{align*}
and also define the surjective homomorphisms $\phi_b: \mathbb{F}_2[x, y, z, b, c] \to \mathbb{F}_2[x, y, z, c]$ and $\phi_c: \mathbb{F}_2[x, y, z, c] \to \mathbb{F}_2[x, y, z]$ in a similar fashion. Then we have the commutative diagram
\begin{equation*}
\begin{tikzcd}
\mathbb{F}_2[x, y, z, a, b, c] \arrow[r, "\phi_a"] \arrow[rrr, bend right = 15, "\phi"] & \mathbb{F}_2[x, y, z, b, c] \arrow[r, "\phi_b"] & \mathbb{F}_2[x, y, z, c] \arrow[r, "\phi_c"] & \mathbb{F}_2[x, y, z].
\end{tikzcd}
\end{equation*}
In light of the containments
\begin{align*}
\mathbb{F}_2[x, y, z] \subset \mathbb{F}_2[x, y, z, c] \subset \mathbb{F}_2[x, y, z, b, c] \subset \mathbb{F}_2[x, y, z, a, b, c]
\end{align*}
we have $\ker(\phi) = \ker(\phi_c \circ \phi_b \circ \phi_a) = \ker(\phi_a) + \ker(\phi_b) + \ker(\phi_c)$. Now, making the identification $\mathbb{F}_2[x, y, z, c] = \mathbb{F}_2[x, y, z][c]$ and using the division algorithm with the divisor being the monic polynomial $c - z^2 \in \mathbb{F}_2[x, y, z][c]$, we conclude that $\ker(\phi_c) = \langle z^2 - c \rangle \subset \mathbb{F}_2[x, y, z, c]$. By similar arguments, we also have $\ker(\phi_b) = \langle y^2 - b \rangle \subset \mathbb{F}_2[x, y, z, b, c]$ and $\ker(\phi_a) = \langle x^2 - a \rangle \subset \mathbb{F}_2[x, y, z, a, b, c]$. Thus $\ker(\phi) = \langle x^2 - a, y^2 - b, z^2 - c \rangle \subset \mathbb{F}_2[x, y, z, a, b, c]$.

We now compute that $\phi^{-1}(I) = I + \ker(\phi) = I + \langle x^2 - a, y^2 - b, z^2 - c \rangle = \tilde{I}$. Hence the surjective homomorphism $\phi|_{\tilde{R}' + \tilde{I}}: \tilde{R}' + \tilde{I} \to R' + I$ lifts to the isomorphism $\overline{\phi|_{\tilde{R}' + \tilde{I}}}: \frac{\tilde{R}' + \tilde{I}}{\tilde{I}} \to \frac{R' + I}{I}$. We use this and the second isomorphism theorem to get
\begin{align*}
\frac{\tilde{R}'}{\tilde{I} \cap \tilde{R}'} \cong \frac{\tilde{R}' + \tilde{I}}{\tilde{I}} \cong \frac{R' + I}{I} \cong \frac{R'}{I \cap R'}.
\end{align*}
More explicitly tracing the maps in the isomorphisms above, we have the commutative diagram
\begin{equation*}
\begin{tikzcd}
\dfrac{\tilde{R}' + \tilde{I}}{\tilde{I}} \arrow[rrr, "\overline{\phi|_{\tilde{R}' + \tilde{I}}}"] &[-20pt] & &[-20pt] \dfrac{R' + I}{I} \arrow[ddd] \\[-20pt]
& \tilde{r}' + \tilde{I} \arrow[r, mapsto] & \phi(\tilde{r}') + I \arrow[d, mapsto] & \\
& \tilde{r}' + \tilde{I} \cap \tilde{R}' \arrow[u, mapsto] \arrow[r, mapsto, dashed] & \phi(\tilde{r}') + I \cap R' & \\[-20pt]
\dfrac{\tilde{R}'}{\tilde{I} \cap \tilde{R}'} \arrow[uuu] \arrow[rrr, dashed] &&& \dfrac{R'}{I \cap R'}.
\end{tikzcd}
\end{equation*} 
We see that the isomorphism $\frac{\tilde{R}'}{\tilde{I} \cap \tilde{R}'} \cong \frac{R'}{I \cap R'}$ is induced by the surjective homomorphism $\phi|_{\tilde{R}'}: \tilde{R}' \to R'$. This implies $\tilde{I} \cap \tilde{R}' = \{\tilde{r}' \in \tilde{R}': \phi(\tilde{r}') \in I \cap R'\} = \phi^{-1}(I \cap R')$. Hence we conclude $\phi(\tilde{I} \cap \tilde{R}') = I \cap R'$.
\end{proof}

%% file: main.bbl
%merlin.mbs apsrev4-1.bst 2010-07-25 4.21a (PWD, AO, DPC) hacked
%Control: key (0)
%Control: author (72) initials jnrlst
%Control: editor formatted (1) identically to author
%Control: production of article title (-1) disabled
%Control: page (0) single
%Control: year (1) truncated
%Control: production of eprint (0) enabled
\begin{thebibliography}{82}%
\makeatletter
\providecommand \@ifxundefined [1]{%
 \@ifx{#1\undefined}
}%
\providecommand \@ifnum [1]{%
 \ifnum #1\expandafter \@firstoftwo
 \else \expandafter \@secondoftwo
 \fi
}%
\providecommand \@ifx [1]{%
 \ifx #1\expandafter \@firstoftwo
 \else \expandafter \@secondoftwo
 \fi
}%
\providecommand \natexlab [1]{#1}%
\providecommand \enquote  [1]{``#1''}%
\providecommand \bibnamefont  [1]{#1}%
\providecommand \bibfnamefont [1]{#1}%
\providecommand \citenamefont [1]{#1}%
\providecommand \href@noop [0]{\@secondoftwo}%
\providecommand \href [0]{\begingroup \@sanitize@url \@href}%
\providecommand \@href[1]{\@@startlink{#1}\@@href}%
\providecommand \@@href[1]{\endgroup#1\@@endlink}%
\providecommand \@sanitize@url [0]{\catcode `\\12\catcode `\$12\catcode
  `\&12\catcode `\#12\catcode `\^12\catcode `\_12\catcode `\%12\relax}%
\providecommand \@@startlink[1]{}%
\providecommand \@@endlink[0]{}%
\providecommand \url  [0]{\begingroup\@sanitize@url \@url }%
\providecommand \@url [1]{\endgroup\@href {#1}{\urlprefix }}%
\providecommand \urlprefix  [0]{URL }%
\providecommand \Eprint [0]{\href }%
\providecommand \doibase [0]{http://dx.doi.org/}%
\providecommand \selectlanguage [0]{\@gobble}%
\providecommand \bibinfo  [0]{\@secondoftwo}%
\providecommand \bibfield  [0]{\@secondoftwo}%
\providecommand \translation [1]{[#1]}%
\providecommand \BibitemOpen [0]{}%
\providecommand \bibitemStop [0]{}%
\providecommand \bibitemNoStop [0]{.\EOS\space}%
\providecommand \EOS [0]{\spacefactor3000\relax}%
\providecommand \BibitemShut  [1]{\csname bibitem#1\endcsname}%
\let\auto@bib@innerbib\@empty
%</preamble>
\bibitem [{\citenamefont {Wilson}(1975)}]{Wilson1975}%
  \BibitemOpen
  \bibfield  {author} {\bibinfo {author} {\bibfnamefont {K.~G.}\ \bibnamefont
  {Wilson}},\ }{The renormalization group: Critical phenomena and the Kondo
  problem},\ \href {\doibase 10.1103/RevModPhys.47.773} {\bibfield  {journal}
  {\bibinfo  {journal} {Rev. Mod. Phys.}\ }\textbf {\bibinfo {volume} {47}},\
  \bibinfo {pages} {773}} (\bibinfo {year} {1975})\BibitemShut {NoStop}%
\bibitem [{\citenamefont {Haah}(2014)}]{haah2014bifurcation}%
  \BibitemOpen
  \bibfield  {author} {\bibinfo {author} {\bibfnamefont {J.}~\bibnamefont
  {Haah}},\ }{Bifurcation in entanglement renormalization group flow of a
  gapped spin model},\ \href {\doibase 10.1103/PhysRevB.89.075119} {\bibfield
  {journal} {\bibinfo  {journal} {Phys. Rev. B}\ }\textbf {\bibinfo {volume}
  {89}},\ \bibinfo {pages} {75119}},\ \Eprint {http://arxiv.org/abs/1310.4507}
  {arXiv:1310.4507}  (\bibinfo {year} {2014})\BibitemShut {NoStop}%
\bibitem [{\citenamefont {Vidal}(2007)}]{Vidal2007}%
  \BibitemOpen
  \bibfield  {author} {\bibinfo {author} {\bibfnamefont {G.}~\bibnamefont
  {Vidal}},\ }{Entanglement renormalization},\ \href {\doibase
  10.1103/PhysRevLett.99.220405} {\bibfield  {journal} {\bibinfo  {journal}
  {Phys. Rev. Lett.}\ }\textbf {\bibinfo {volume} {99}},\ \bibinfo {pages}
  {220405}},\ \Eprint {http://arxiv.org/abs/cond-mat/0512165}
  {arXiv:cond-mat/0512165}  (\bibinfo {year} {2007})\BibitemShut {NoStop}%
\bibitem [{\citenamefont {K{\"{o}}nig}\ \emph {et~al.}(2009)\citenamefont
  {K{\"{o}}nig}, \citenamefont {Reichardt},\ and\ \citenamefont
  {Vidal}}]{Konig2009}%
  \BibitemOpen
  \bibfield  {author} {\bibinfo {author} {\bibfnamefont {R.}~\bibnamefont
  {K{\"{o}}nig}}, \bibinfo {author} {\bibfnamefont {B.~W.}\ \bibnamefont
  {Reichardt}}, \ and\ \bibinfo {author} {\bibfnamefont {G.}~\bibnamefont
  {Vidal}},\ }{Exact entanglement renormalization for string-net models},\
  \href {\doibase 10.1103/PhysRevB.79.195123} {\bibfield  {journal} {\bibinfo
  {journal} {Phys. Rev. B}\ }\textbf {\bibinfo {volume} {79}},\
  10.1103/PhysRevB.79.195123},\ \Eprint {http://arxiv.org/abs/0806.4583}
  {arXiv:0806.4583}  (\bibinfo {year} {2009})\BibitemShut {NoStop}%
\bibitem [{\citenamefont {Aguado}\ and\ \citenamefont
  {Vidal}(2008)}]{Aguado2008}%
  \BibitemOpen
  \bibfield  {author} {\bibinfo {author} {\bibfnamefont {M.}~\bibnamefont
  {Aguado}}\ and\ \bibinfo {author} {\bibfnamefont {G.}~\bibnamefont {Vidal}},\
  }{Entanglement renormalization and topological order},\ \href
  {http://dx.doi.org/10.1103/PhysRevLett.100.070404} {\bibfield  {journal}
  {\bibinfo  {journal} {Phys. Rev. Lett.}\ }\textbf {\bibinfo {volume} {100}},\
  \bibinfo {pages} {070404}},\ \Eprint {http://arxiv.org/abs/0712.0348}
  {arXiv:0712.0348}  (\bibinfo {year} {2008})\BibitemShut {NoStop}%
\bibitem [{\citenamefont {Haah}(2018)}]{Haah2018a}%
  \BibitemOpen
  \bibfield  {author} {\bibinfo {author} {\bibfnamefont {J.}~\bibnamefont
  {Haah}},\ }{Classification of translation invariant topological Pauli
  stabilizer codes for prime dimensional qudits on two-dimensional lattices},\
  \href {http://arxiv.org/abs/1812.11193} {\ }\Eprint
  {http://arxiv.org/abs/1812.11193} {arXiv:1812.11193}  (\bibinfo {year}
  {2018})\BibitemShut {NoStop}%
\bibitem [{\citenamefont {Bombin}\ \emph {et~al.}(2012)\citenamefont {Bombin},
  \citenamefont {Duclos-Cianci},\ and\ \citenamefont
  {Poulin}}]{bombin2012universal}%
  \BibitemOpen
  \bibfield  {author} {\bibinfo {author} {\bibfnamefont {H.}~\bibnamefont
  {Bombin}}, \bibinfo {author} {\bibfnamefont {G.}~\bibnamefont
  {Duclos-Cianci}}, \ and\ \bibinfo {author} {\bibfnamefont {D.}~\bibnamefont
  {Poulin}},\ }{Universal topological phase of two-dimensional stabilizer
  codes},\ \href {\doibase 10.1088/1367-2630/14/7/073048} {\bibfield  {journal}
  {\bibinfo  {journal} {New J. Phys.}\ }\textbf {\bibinfo {volume} {14}},\
  \bibinfo {pages} {73048}},\ \Eprint {http://arxiv.org/abs/1103.4606}
  {arXiv:1103.4606}  (\bibinfo {year} {2012})\BibitemShut {NoStop}%
\bibitem [{\citenamefont {Bomb{\'{i}}n}(2014)}]{bombin2014structure}%
  \BibitemOpen
  \bibfield  {author} {\bibinfo {author} {\bibfnamefont {H.}~\bibnamefont
  {Bomb{\'{i}}n}},\ }{Structure of 2D Topological Stabilizer Codes},\ \href
  {\doibase 10.1007/s00220-014-1893-4} {\bibfield  {journal} {\bibinfo
  {journal} {Commun. Math. Phys.}\ }\textbf {\bibinfo {volume} {327}},\
  \bibinfo {pages} {387}},\ \Eprint {http://arxiv.org/abs/1107.2707}
  {arXiv:1107.2707}  (\bibinfo {year} {2014})\BibitemShut {NoStop}%
\bibitem [{\citenamefont {Chamon}(2005)}]{chamon2005quantum}%
  \BibitemOpen
  \bibfield  {author} {\bibinfo {author} {\bibfnamefont {C.}~\bibnamefont
  {Chamon}},\ }{Quantum glassiness in strongly correlated clean systems: An
  example of topological overprotection},\ \href {\doibase
  10.1103/PhysRevLett.94.040402} {\bibfield  {journal} {\bibinfo  {journal}
  {Phys. Rev. Lett.}\ }\textbf {\bibinfo {volume} {94}},\ \bibinfo {pages}
  {40402}},\ \Eprint {http://arxiv.org/abs/cond-mat/0404182}
  {arXiv:cond-mat/0404182}  (\bibinfo {year} {2005})\BibitemShut {NoStop}%
\bibitem [{\citenamefont {Castelnovo}\ \emph {et~al.}(2010)\citenamefont
  {Castelnovo}, \citenamefont {Chamon},\ and\ \citenamefont
  {Sherrington}}]{PhysRevB.81.184303}%
  \BibitemOpen
  \bibfield  {author} {\bibinfo {author} {\bibfnamefont {C.}~\bibnamefont
  {Castelnovo}}, \bibinfo {author} {\bibfnamefont {C.}~\bibnamefont {Chamon}},
  \ and\ \bibinfo {author} {\bibfnamefont {D.}~\bibnamefont {Sherrington}},\
  }{Quantum mechanical and information theoretic view on classical glass
  transitions},\ \href {\doibase 10.1103/PhysRevB.81.184303} {\bibfield
  {journal} {\bibinfo  {journal} {Phys. Rev. B}\ }\textbf {\bibinfo {volume}
  {81}},\ \bibinfo {pages} {184303}},\ \Eprint {http://arxiv.org/abs/1003.3832}
  {arXiv:1003.3832}  (\bibinfo {year} {2010})\BibitemShut {NoStop}%
\bibitem [{\citenamefont {Bravyi}\ \emph {et~al.}(2010)\citenamefont {Bravyi},
  \citenamefont {Leemhuis},\ and\ \citenamefont
  {Terhal}}]{bravyi2011topological}%
  \BibitemOpen
  \bibfield  {author} {\bibinfo {author} {\bibfnamefont {S.}~\bibnamefont
  {Bravyi}}, \bibinfo {author} {\bibfnamefont {B.}~\bibnamefont {Leemhuis}}, \
  and\ \bibinfo {author} {\bibfnamefont {B.~M.}\ \bibnamefont {Terhal}},\
  }{Topological order in an exactly solvable 3D spin model},\ \href {\doibase
  10.1016/j.aop.2010.11.002} {\bibfield  {journal} {\bibinfo  {journal} {Ann.
  Phys.}\ }\textbf {\bibinfo {volume} {326}},\ \bibinfo {pages} {839}},\
  \Eprint {http://arxiv.org/abs/1006.4871} {arXiv:1006.4871}  (\bibinfo {year}
  {2010})\BibitemShut {NoStop}%
\bibitem [{\citenamefont {Castelnovo}\ and\ \citenamefont
  {Chamon}(2012)}]{Chamon_quantum_glassiness}%
  \BibitemOpen
  \bibfield  {author} {\bibinfo {author} {\bibfnamefont {C.}~\bibnamefont
  {Castelnovo}}\ and\ \bibinfo {author} {\bibfnamefont {C.}~\bibnamefont
  {Chamon}},\ }{Topological quantum glassiness},\ \href {\doibase
  10.1080/14786435.2011.609152} {\bibfield  {journal} {\bibinfo  {journal}
  {Philosophical Magazine}\ }\textbf {\bibinfo {volume} {92}},\ \bibinfo
  {pages} {304}},\ \Eprint {http://arxiv.org/abs/1108.2051} {arXiv:1108.2051}
  (\bibinfo {year} {2012})\BibitemShut {NoStop}%
\bibitem [{\citenamefont {Ma}\ \emph {et~al.}(2017)\citenamefont {Ma},
  \citenamefont {Lake}, \citenamefont {Chen},\ and\ \citenamefont
  {Hermele}}]{PhysRevB.95.245126}%
  \BibitemOpen
  \bibfield  {author} {\bibinfo {author} {\bibfnamefont {H.}~\bibnamefont
  {Ma}}, \bibinfo {author} {\bibfnamefont {E.}~\bibnamefont {Lake}}, \bibinfo
  {author} {\bibfnamefont {X.}~\bibnamefont {Chen}}, \ and\ \bibinfo {author}
  {\bibfnamefont {M.}~\bibnamefont {Hermele}},\ }{Fracton topological order via
  coupled layers},\ \href {\doibase 10.1103/PhysRevB.95.245126} {\bibfield
  {journal} {\bibinfo  {journal} {Phys. Rev. B}\ }\textbf {\bibinfo {volume}
  {95}},\ \bibinfo {pages} {245126}},\ \Eprint
  {http://arxiv.org/abs/1701.00747} {arXiv:1701.00747}  (\bibinfo {year}
  {2017})\BibitemShut {NoStop}%
\bibitem [{\citenamefont {Vijay}\ \emph {et~al.}(2016)\citenamefont {Vijay},
  \citenamefont {Haah},\ and\ \citenamefont {Fu}}]{vijay2016fracton}%
  \BibitemOpen
  \bibfield  {author} {\bibinfo {author} {\bibfnamefont {S.}~\bibnamefont
  {Vijay}}, \bibinfo {author} {\bibfnamefont {J.}~\bibnamefont {Haah}}, \ and\
  \bibinfo {author} {\bibfnamefont {L.}~\bibnamefont {Fu}},\ }{Fracton
  Topological Order, Generalized Lattice Gauge Theory and Duality},\ \href
  {\doibase 10.1103/PhysRevB.94.235157} {\bibfield  {journal} {\bibinfo
  {journal} {Phys. Rev. B}\ }\textbf {\bibinfo {volume} {94}},\ \bibinfo
  {pages} {235157}},\ \Eprint {http://arxiv.org/abs/1603.04442}
  {arXiv:1603.04442}  (\bibinfo {year} {2016})\BibitemShut {NoStop}%
\bibitem [{\citenamefont {Williamson}(2016)}]{Williamson_cubic_code}%
  \BibitemOpen
  \bibfield  {author} {\bibinfo {author} {\bibfnamefont {D.~J.}\ \bibnamefont
  {Williamson}},\ }{Fractal symmetries: Ungauging the cubic code},\ \href
  {\doibase 10.1103/PhysRevB.94.155128} {\bibfield  {journal} {\bibinfo
  {journal} {Phys. Rev. B}\ }\textbf {\bibinfo {volume} {94}},\ \bibinfo
  {pages} {155128}},\ \Eprint {http://arxiv.org/abs/1603.05182}
  {arXiv:1603.05182}  (\bibinfo {year} {2016})\BibitemShut {NoStop}%
\bibitem [{\citenamefont {Vijay}(2017)}]{vijay2017isotropic}%
  \BibitemOpen
  \bibfield  {author} {\bibinfo {author} {\bibfnamefont {S.}~\bibnamefont
  {Vijay}},\ }{Isotropic Layer Construction and Phase Diagram for Fracton
  Topological Phases},\ \href {http://arxiv.org/abs/1701.00762} {\ }\Eprint
  {http://arxiv.org/abs/1701.00762} {arXiv:1701.00762}  (\bibinfo {year}
  {2017})\BibitemShut {NoStop}%
\bibitem [{\citenamefont {Vijay}\ and\ \citenamefont
  {Fu}(2017)}]{vijay2017generalization}%
  \BibitemOpen
  \bibfield  {author} {\bibinfo {author} {\bibfnamefont {S.}~\bibnamefont
  {Vijay}}\ and\ \bibinfo {author} {\bibfnamefont {L.}~\bibnamefont {Fu}},\ }{A
  Generalization of Non-Abelian Anyons in Three Dimensions},\ \href
  {http://arxiv.org/abs/1706.07070} {\ }\Eprint
  {http://arxiv.org/abs/1706.07070} {arXiv:1706.07070}  (\bibinfo {year}
  {2017})\BibitemShut {NoStop}%
\bibitem [{\citenamefont {Slagle}\ and\ \citenamefont
  {Kim}(2017)}]{PhysRevB.96.165106}%
  \BibitemOpen
  \bibfield  {author} {\bibinfo {author} {\bibfnamefont {K.}~\bibnamefont
  {Slagle}}\ and\ \bibinfo {author} {\bibfnamefont {Y.~B.}\ \bibnamefont
  {Kim}},\ }{Fracton topological order from nearest-neighbor two-spin
  interactions and dualities},\ \href {\doibase 10.1103/PhysRevB.96.165106}
  {\bibfield  {journal} {\bibinfo  {journal} {Phys. Rev. B}\ }\textbf {\bibinfo
  {volume} {96}},\ \bibinfo {pages} {165106}},\ \Eprint
  {http://arxiv.org/abs/1704.03870} {arXiv:1704.03870}  (\bibinfo {year}
  {2017})\BibitemShut {NoStop}%
\bibitem [{\citenamefont {Hal{\'{a}}sz}\ \emph {et~al.}(2017)\citenamefont
  {Hal{\'{a}}sz}, \citenamefont {Hsieh},\ and\ \citenamefont
  {Balents}}]{HHB_models}%
  \BibitemOpen
  \bibfield  {author} {\bibinfo {author} {\bibfnamefont {G.~B.}\ \bibnamefont
  {Hal{\'{a}}sz}}, \bibinfo {author} {\bibfnamefont {T.~H.}\ \bibnamefont
  {Hsieh}}, \ and\ \bibinfo {author} {\bibfnamefont {L.}~\bibnamefont
  {Balents}},\ }{Fracton Topological Phases from Strongly Coupled Spin
  Chains},\ \href {\doibase 10.1103/PhysRevLett.119.257202} {\bibfield
  {journal} {\bibinfo  {journal} {Phys. Rev. Lett.}\ }\textbf {\bibinfo
  {volume} {119}},\ \bibinfo {pages} {257202}},\ \Eprint
  {http://arxiv.org/abs/1707.02308} {arXiv:1707.02308}  (\bibinfo {year}
  {2017})\BibitemShut {NoStop}%
\bibitem [{\citenamefont {Devakul}(2018{\natexlab{a}})}]{PhysRevB.97.155111}%
  \BibitemOpen
  \bibfield  {author} {\bibinfo {author} {\bibfnamefont {T.}~\bibnamefont
  {Devakul}},\ }{$Z_3$ topological order in the face-centered-cubic quantum
  plaquette model},\ \href {\doibase 10.1103/PhysRevB.97.155111} {\bibfield
  {journal} {\bibinfo  {journal} {Phys. Rev. B}\ }\textbf {\bibinfo {volume}
  {97}},\ \bibinfo {pages} {155111}},\ \Eprint
  {http://arxiv.org/abs/1712.05377} {arXiv:1712.05377}  (\bibinfo {year}
  {2018}{\natexlab{a}})\BibitemShut {NoStop}%
\bibitem [{\citenamefont {Devakul}\ \emph
  {et~al.}(2018{\natexlab{a}})\citenamefont {Devakul}, \citenamefont
  {Parameswaran},\ and\ \citenamefont {Sondhi}}]{PhysRevB.97.041110}%
  \BibitemOpen
  \bibfield  {author} {\bibinfo {author} {\bibfnamefont {T.}~\bibnamefont
  {Devakul}}, \bibinfo {author} {\bibfnamefont {S.~A.}\ \bibnamefont
  {Parameswaran}}, \ and\ \bibinfo {author} {\bibfnamefont {S.~L.}\
  \bibnamefont {Sondhi}},\ }{Correlation function diagnostics for type-I
  fracton phases},\ \href {\doibase 10.1103/PhysRevB.97.041110} {\bibfield
  {journal} {\bibinfo  {journal} {Phys. Rev. B}\ }\textbf {\bibinfo {volume}
  {97}},\ \bibinfo {pages} {41110}},\ \Eprint {http://arxiv.org/abs/1709.10071}
  {arXiv:1709.10071}  (\bibinfo {year} {2018}{\natexlab{a}})\BibitemShut
  {NoStop}%
\bibitem [{\citenamefont {Prem}\ \emph {et~al.}(2019)\citenamefont {Prem},
  \citenamefont {Huang}, \citenamefont {Song},\ and\ \citenamefont
  {Hermele}}]{prem2018cage}%
  \BibitemOpen
  \bibfield  {author} {\bibinfo {author} {\bibfnamefont {A.}~\bibnamefont
  {Prem}}, \bibinfo {author} {\bibfnamefont {S.-J.}\ \bibnamefont {Huang}},
  \bibinfo {author} {\bibfnamefont {H.}~\bibnamefont {Song}}, \ and\ \bibinfo
  {author} {\bibfnamefont {M.}~\bibnamefont {Hermele}},\ }Cage-net fracton
  models,\ \href {\doibase 10.1103/PhysRevX.9.021010} {\bibfield  {journal}
  {\bibinfo  {journal} {Phys. Rev. X}\ }\textbf {\bibinfo {volume} {9}},\
  \bibinfo {pages} {021010}},\ \Eprint {http://arxiv.org/abs/1806.04687}
  {arXiv:1806.04687}  (\bibinfo {year} {2019})\BibitemShut {NoStop}%
\bibitem [{\citenamefont {Bulmash}\ and\ \citenamefont
  {Iadecola}(2019)}]{Bulmash2018}%
  \BibitemOpen
  \bibfield  {author} {\bibinfo {author} {\bibfnamefont {D.}~\bibnamefont
  {Bulmash}}\ and\ \bibinfo {author} {\bibfnamefont {T.}~\bibnamefont
  {Iadecola}},\ }Braiding and gapped boundaries in fracton topological phases,\
  \href {\doibase 10.1103/PhysRevB.99.125132} {\bibfield  {journal} {\bibinfo
  {journal} {Phys. Rev. B}\ }\textbf {\bibinfo {volume} {99}},\ \bibinfo
  {pages} {125132}},\ \Eprint {http://arxiv.org/abs/1810.00012}
  {arXiv:1810.00012}  (\bibinfo {year} {2019})\BibitemShut {NoStop}%
\bibitem [{\citenamefont {Song}\ \emph
  {et~al.}(2019{\natexlab{a}})\citenamefont {Song}, \citenamefont {Prem},
  \citenamefont {Huang},\ and\ \citenamefont {Martin-Delgado}}]{hao_twisted}%
  \BibitemOpen
  \bibfield  {author} {\bibinfo {author} {\bibfnamefont {H.}~\bibnamefont
  {Song}}, \bibinfo {author} {\bibfnamefont {A.}~\bibnamefont {Prem}}, \bibinfo
  {author} {\bibfnamefont {S.-J.}\ \bibnamefont {Huang}}, \ and\ \bibinfo
  {author} {\bibfnamefont {M.~A.}\ \bibnamefont {Martin-Delgado}},\ }Twisted
  fracton models in three dimensions,\ \href {\doibase
  10.1103/PhysRevB.99.155118} {\bibfield  {journal} {\bibinfo  {journal} {Phys.
  Rev. B}\ }\textbf {\bibinfo {volume} {99}},\ \bibinfo {pages} {155118}},\
  \Eprint {http://arxiv.org/abs/1805.06899} {arXiv:1805.06899}  (\bibinfo
  {year} {2019}{\natexlab{a}})\BibitemShut {NoStop}%
\bibitem [{\citenamefont {Brown}\ and\ \citenamefont
  {Williamson}(2019)}]{Brown2019}%
  \BibitemOpen
  \bibfield  {author} {\bibinfo {author} {\bibfnamefont {B.~J.}\ \bibnamefont
  {Brown}}\ and\ \bibinfo {author} {\bibfnamefont {D.~J.}\ \bibnamefont
  {Williamson}},\ }{Parallelized quantum error correction with fracton
  topological codes},\ \href {http://arxiv.org/abs/1901.08061} {\ }\Eprint
  {http://arxiv.org/abs/1901.08061} {arXiv:1901.08061}  (\bibinfo {year}
  {2019})\BibitemShut {NoStop}%
\bibitem [{\citenamefont {Weinstein}\ \emph {et~al.}(2018)\citenamefont
  {Weinstein}, \citenamefont {Cobanera}, \citenamefont {Ortiz},\ and\
  \citenamefont {Nussinov}}]{finite_temp_Xcube}%
  \BibitemOpen
  \bibfield  {author} {\bibinfo {author} {\bibfnamefont {Z.}~\bibnamefont
  {Weinstein}}, \bibinfo {author} {\bibfnamefont {E.}~\bibnamefont {Cobanera}},
  \bibinfo {author} {\bibfnamefont {G.}~\bibnamefont {Ortiz}}, \ and\ \bibinfo
  {author} {\bibfnamefont {Z.}~\bibnamefont {Nussinov}},\ }Absence of finite
  temperature phase transitions in the x-cube model and its zp generalization,\
  \href@noop {} {\ }\Eprint {http://arxiv.org/abs/1812.04561}
  {arXiv:1812.04561}  (\bibinfo {year} {2018})\BibitemShut {NoStop}%
\bibitem [{\citenamefont {Hsieh}\ and\ \citenamefont
  {Hal\'asz}(2017)}]{hsieh_halasz_partons}%
  \BibitemOpen
  \bibfield  {author} {\bibinfo {author} {\bibfnamefont {T.~H.}\ \bibnamefont
  {Hsieh}}\ and\ \bibinfo {author} {\bibfnamefont {G.~B.}\ \bibnamefont
  {Hal\'asz}},\ }Fractons from partons,\ \href {\doibase
  10.1103/PhysRevB.96.165105} {\bibfield  {journal} {\bibinfo  {journal} {Phys.
  Rev. B}\ }\textbf {\bibinfo {volume} {96}},\ \bibinfo {pages} {165105}},\
  \Eprint {http://arxiv.org/abs/1703.02973} {arXiv:1703.02973}  (\bibinfo
  {year} {2017})\BibitemShut {NoStop}%
\bibitem [{\citenamefont {Prem}\ and\ \citenamefont
  {Williamson}(2019)}]{Prem2019}%
  \BibitemOpen
  \bibfield  {author} {\bibinfo {author} {\bibfnamefont {A.}~\bibnamefont
  {Prem}}\ and\ \bibinfo {author} {\bibfnamefont {D.~J.}\ \bibnamefont
  {Williamson}},\ }{Gauging permutation symmetries as a route to non-Abelian
  fractons},\ \href {http://arxiv.org/abs/1905.06309} {\ }\Eprint
  {http://arxiv.org/abs/1905.06309} {arXiv:1905.06309}  (\bibinfo {year}
  {2019})\BibitemShut {NoStop}%
\bibitem [{\citenamefont {Bulmash}\ and\ \citenamefont
  {Barkeshli}(2019)}]{Bulmash2019}%
  \BibitemOpen
  \bibfield  {author} {\bibinfo {author} {\bibfnamefont {D.}~\bibnamefont
  {Bulmash}}\ and\ \bibinfo {author} {\bibfnamefont {M.}~\bibnamefont
  {Barkeshli}},\ }{Gauging fractons: immobile non-Abelian quasiparticles,
  fractals, and position-dependent degeneracies},\ \href
  {http://arxiv.org/abs/1905.05771} {\ }\Eprint
  {http://arxiv.org/abs/1905.05771} {arXiv:1905.05771}  (\bibinfo {year}
  {2019})\BibitemShut {NoStop}%
\bibitem [{\citenamefont {Dua}\ \emph {et~al.}(2019{\natexlab{a}})\citenamefont
  {Dua}, \citenamefont {Kim}, \citenamefont {Cheng},\ and\ \citenamefont
  {Williamson}}]{Dua_Classification_2019}%
  \BibitemOpen
  \bibfield  {author} {\bibinfo {author} {\bibfnamefont {A.}~\bibnamefont
  {Dua}}, \bibinfo {author} {\bibfnamefont {I.~H.}\ \bibnamefont {Kim}},
  \bibinfo {author} {\bibfnamefont {M.}~\bibnamefont {Cheng}}, \ and\ \bibinfo
  {author} {\bibfnamefont {D.~J.}\ \bibnamefont {Williamson}},\ }Sorting
  topological stabilizer models in three dimensions,\ \href@noop {} {\ }\Eprint
  {http://arxiv.org/abs/1908.08049} {arXiv:1908.08049}  (\bibinfo {year}
  {2019}{\natexlab{a}})\BibitemShut {NoStop}%
\bibitem [{\citenamefont {Dua}\ \emph {et~al.}(2019{\natexlab{b}})\citenamefont
  {Dua}, \citenamefont {Williamson}, \citenamefont {Haah},\ and\ \citenamefont
  {Cheng}}]{Dua2019_compactify}%
  \BibitemOpen
  \bibfield  {author} {\bibinfo {author} {\bibfnamefont {A.}~\bibnamefont
  {Dua}}, \bibinfo {author} {\bibfnamefont {D.~J.}\ \bibnamefont {Williamson}},
  \bibinfo {author} {\bibfnamefont {J.}~\bibnamefont {Haah}}, \ and\ \bibinfo
  {author} {\bibfnamefont {M.}~\bibnamefont {Cheng}},\ }{Compactifying fracton
  stabilizer models},\ \href {http://arxiv.org/abs/1903.12246} {\bibfield
  {journal} {\bibinfo  {journal} {Phys. Rev. B}\ }\textbf {\bibinfo {volume}
  {99}},\ \bibinfo {pages} {245135}},\ \Eprint
  {http://arxiv.org/abs/1903.12246} {arXiv:1903.12246}  (\bibinfo {year}
  {2019}{\natexlab{b}})\BibitemShut {NoStop}%
\bibitem [{\citenamefont {Evenbly}\ and\ \citenamefont
  {Vidal}(2014{\natexlab{a}})}]{Evenbly2014real}%
  \BibitemOpen
  \bibfield  {author} {\bibinfo {author} {\bibfnamefont {G.}~\bibnamefont
  {Evenbly}}\ and\ \bibinfo {author} {\bibfnamefont {G.}~\bibnamefont
  {Vidal}},\ }{Real-Space Decoupling Transformation for Quantum Many-Body
  Systems},\ \href {\doibase 10.1103/PhysRevLett.112.220502} {\bibfield
  {journal} {\bibinfo  {journal} {Phys. Rev. Lett.}\ }\textbf {\bibinfo
  {volume} {112}},\ \bibinfo {pages} {220502}} (\bibinfo {year}
  {2014}{\natexlab{a}})\BibitemShut {NoStop}%
\bibitem [{\citenamefont {Evenbly}\ and\ \citenamefont
  {Vidal}(2014{\natexlab{b}})}]{Evenbly2014class}%
  \BibitemOpen
  \bibfield  {author} {\bibinfo {author} {\bibfnamefont {G.}~\bibnamefont
  {Evenbly}}\ and\ \bibinfo {author} {\bibfnamefont {G.}~\bibnamefont
  {Vidal}},\ }{Class of Highly Entangled Many-Body States that can be
  Efficiently Simulated},\ \href {\doibase 10.1103/PhysRevLett.112.240502}
  {\bibfield  {journal} {\bibinfo  {journal} {Phys. Rev. Lett.}\ }\textbf
  {\bibinfo {volume} {112}},\ \bibinfo {pages} {240502}} (\bibinfo {year}
  {2014}{\natexlab{b}})\BibitemShut {NoStop}%
\bibitem [{\citenamefont {Evenbly}\ and\ \citenamefont
  {Vidal}(2014{\natexlab{c}})}]{Evenbly2014scaling}%
  \BibitemOpen
  \bibfield  {author} {\bibinfo {author} {\bibfnamefont {G.}~\bibnamefont
  {Evenbly}}\ and\ \bibinfo {author} {\bibfnamefont {G.}~\bibnamefont
  {Vidal}},\ }{Scaling of entanglement entropy in the (branching) multiscale
  entanglement renormalization ansatz},\ \href {\doibase
  10.1103/PhysRevB.89.235113} {\bibfield  {journal} {\bibinfo  {journal} {Phys.
  Rev. B}\ }\textbf {\bibinfo {volume} {89}},\ \bibinfo {pages} {235113}}
  (\bibinfo {year} {2014}{\natexlab{c}})\BibitemShut {NoStop}%
\bibitem [{\citenamefont {Chen}\ \emph {et~al.}(2010)\citenamefont {Chen},
  \citenamefont {Gu},\ and\ \citenamefont {Wen}}]{chen2010local}%
  \BibitemOpen
  \bibfield  {author} {\bibinfo {author} {\bibfnamefont {X.}~\bibnamefont
  {Chen}}, \bibinfo {author} {\bibfnamefont {Z.~C.}\ \bibnamefont {Gu}}, \ and\
  \bibinfo {author} {\bibfnamefont {X.~G.}\ \bibnamefont {Wen}},\ }{Local
  unitary transformation, long-range quantum entanglement, wave function
  renormalization, and topological order},\ \href {\doibase
  10.1103/PhysRevB.82.155138} {\bibfield  {journal} {\bibinfo  {journal} {Phys.
  Rev. B}\ }\textbf {\bibinfo {volume} {82}},\ \bibinfo {pages} {155138}},\
  \Eprint {http://arxiv.org/abs/1004.3835} {arXiv:1004.3835}  (\bibinfo {year}
  {2010})\BibitemShut {NoStop}%
\bibitem [{\citenamefont {Swingle}\ and\ \citenamefont
  {McGreevy}(2016{\natexlab{a}})}]{swingle_greevy1}%
  \BibitemOpen
  \bibfield  {author} {\bibinfo {author} {\bibfnamefont {B.}~\bibnamefont
  {Swingle}}\ and\ \bibinfo {author} {\bibfnamefont {J.}~\bibnamefont
  {McGreevy}},\ }Renormalization group constructions of topological quantum
  liquids and beyond,\ \href {\doibase 10.1103/PhysRevB.93.045127} {\bibfield
  {journal} {\bibinfo  {journal} {Phys. Rev. B}\ }\textbf {\bibinfo {volume}
  {93}},\ \bibinfo {pages} {045127}} (\bibinfo {year}
  {2016}{\natexlab{a}})\BibitemShut {NoStop}%
\bibitem [{\citenamefont {Swingle}\ and\ \citenamefont
  {McGreevy}(2016{\natexlab{b}})}]{swingle_greevy2}%
  \BibitemOpen
  \bibfield  {author} {\bibinfo {author} {\bibfnamefont {B.}~\bibnamefont
  {Swingle}}\ and\ \bibinfo {author} {\bibfnamefont {J.}~\bibnamefont
  {McGreevy}},\ }Mixed $s$-sourcery: Building many-body states using bubbles of
  nothing,\ \href {\doibase 10.1103/PhysRevB.94.155125} {\bibfield  {journal}
  {\bibinfo  {journal} {Phys. Rev. B}\ }\textbf {\bibinfo {volume} {94}},\
  \bibinfo {pages} {155125}} (\bibinfo {year} {2016}{\natexlab{b}})\BibitemShut
  {NoStop}%
\bibitem [{\citenamefont {Shirley}\ \emph
  {et~al.}(2019{\natexlab{a}})\citenamefont {Shirley}, \citenamefont {Slagle},\
  and\ \citenamefont {Chen}}]{shirley2018FoliatedFracton}%
  \BibitemOpen
  \bibfield  {author} {\bibinfo {author} {\bibfnamefont {W.}~\bibnamefont
  {Shirley}}, \bibinfo {author} {\bibfnamefont {K.}~\bibnamefont {Slagle}}, \
  and\ \bibinfo {author} {\bibfnamefont {X.}~\bibnamefont {Chen}},\ }{Foliated
  fracton order from gauging subsystem symmetries},\ \href {\doibase
  10.21468/SciPostPhys.6.4.041} {\bibfield  {journal} {\bibinfo  {journal}
  {SciPost Phys.}\ }\textbf {\bibinfo {volume} {6}},\ \bibinfo {pages} {41}},\
  \Eprint {http://arxiv.org/abs/1806.08679} {arXiv:1806.08679}  (\bibinfo
  {year} {2019}{\natexlab{a}})\BibitemShut {NoStop}%
\bibitem [{\citenamefont {Shirley}\ \emph
  {et~al.}(2018{\natexlab{a}})\citenamefont {Shirley}, \citenamefont {Slagle},
  \citenamefont {Wang},\ and\ \citenamefont {Chen}}]{shirley2017fracton}%
  \BibitemOpen
  \bibfield  {author} {\bibinfo {author} {\bibfnamefont {W.}~\bibnamefont
  {Shirley}}, \bibinfo {author} {\bibfnamefont {K.}~\bibnamefont {Slagle}},
  \bibinfo {author} {\bibfnamefont {Z.}~\bibnamefont {Wang}}, \ and\ \bibinfo
  {author} {\bibfnamefont {X.}~\bibnamefont {Chen}},\ }Fracton models on
  general three-dimensional manifolds,\ \href
  {https://journals.aps.org/prx/abstract/10.1103/PhysRevX.8.031051} {\bibfield
  {journal} {\bibinfo  {journal} {Phys. Rev. X}\ }\textbf {\bibinfo {volume}
  {8}},\ \bibinfo {pages} {031051}},\ \Eprint {http://arxiv.org/abs/1712.05892}
  {arXiv:1712.05892}  (\bibinfo {year} {2018}{\natexlab{a}})\BibitemShut
  {NoStop}%
\bibitem [{\citenamefont {Shirley}\ \emph
  {et~al.}(2018{\natexlab{b}})\citenamefont {Shirley}, \citenamefont {Slagle},\
  and\ \citenamefont {Chen}}]{shirley2018Foliated}%
  \BibitemOpen
  \bibfield  {author} {\bibinfo {author} {\bibfnamefont {W.}~\bibnamefont
  {Shirley}}, \bibinfo {author} {\bibfnamefont {K.}~\bibnamefont {Slagle}}, \
  and\ \bibinfo {author} {\bibfnamefont {X.}~\bibnamefont {Chen}},\ }{Foliated
  fracton order in the checkerboard model},\ \href
  {http://arxiv.org/abs/1806.08633} {\ }\Eprint
  {http://arxiv.org/abs/1806.08633} {arXiv:1806.08633}  (\bibinfo {year}
  {2018}{\natexlab{b}})\BibitemShut {NoStop}%
\bibitem [{\citenamefont {Shirley}\ \emph
  {et~al.}(2019{\natexlab{b}})\citenamefont {Shirley}, \citenamefont {Slagle},\
  and\ \citenamefont {Chen}}]{shirley2018universal}%
  \BibitemOpen
  \bibfield  {author} {\bibinfo {author} {\bibfnamefont {W.}~\bibnamefont
  {Shirley}}, \bibinfo {author} {\bibfnamefont {K.}~\bibnamefont {Slagle}}, \
  and\ \bibinfo {author} {\bibfnamefont {X.}~\bibnamefont {Chen}},\ }{Universal
  entanglement signatures of foliated fracton phases},\ \href {\doibase
  10.21468/scipostphys.6.1.015} {\bibfield  {journal} {\bibinfo  {journal}
  {SciPost Phys.}\ }\textbf {\bibinfo {volume} {6}},\ \bibinfo {pages} {1}},\
  \Eprint {http://arxiv.org/abs/1803.10426} {arXiv:1803.10426}  (\bibinfo
  {year} {2019}{\natexlab{b}})\BibitemShut {NoStop}%
\bibitem [{\citenamefont {Shirley}\ \emph
  {et~al.}(2018{\natexlab{c}})\citenamefont {Shirley}, \citenamefont {Slagle},\
  and\ \citenamefont {Chen}}]{shirley2018Fractional}%
  \BibitemOpen
  \bibfield  {author} {\bibinfo {author} {\bibfnamefont {W.}~\bibnamefont
  {Shirley}}, \bibinfo {author} {\bibfnamefont {K.}~\bibnamefont {Slagle}}, \
  and\ \bibinfo {author} {\bibfnamefont {X.}~\bibnamefont {Chen}},\
  }{Fractional excitations in foliated fracton phases},\ \href
  {http://arxiv.org/abs/1806.08625} {\ }\Eprint
  {http://arxiv.org/abs/1806.08625} {arXiv:1806.08625}  (\bibinfo {year}
  {2018}{\natexlab{c}})\BibitemShut {NoStop}%
\bibitem [{\citenamefont {Haah}(2011)}]{haah2011local}%
  \BibitemOpen
  \bibfield  {author} {\bibinfo {author} {\bibfnamefont {J.}~\bibnamefont
  {Haah}},\ }{Local stabilizer codes in three dimensions without string logical
  operators},\ \href {\doibase 10.1103/PhysRevA.83.042330} {\bibfield
  {journal} {\bibinfo  {journal} {Phys. Rev. A}\ }\textbf {\bibinfo {volume}
  {83}},\ \bibinfo {pages} {42330}},\ \Eprint {http://arxiv.org/abs/1101.1962}
  {arXiv:1101.1962}  (\bibinfo {year} {2011})\BibitemShut {NoStop}%
\bibitem [{\citenamefont {Yoshida}(2013)}]{yoshida2013exotic}%
  \BibitemOpen
  \bibfield  {author} {\bibinfo {author} {\bibfnamefont {B.}~\bibnamefont
  {Yoshida}},\ }{Exotic topological order in fractal spin liquids},\ \href
  {\doibase 10.1103/PhysRevB.88.125122} {\bibfield  {journal} {\bibinfo
  {journal} {Phys. Rev. B}\ }\textbf {\bibinfo {volume} {88}},\ \bibinfo
  {pages} {125122}},\ \Eprint {http://arxiv.org/abs/1302.6248}
  {arXiv:1302.6248}  (\bibinfo {year} {2013})\BibitemShut {NoStop}%
\bibitem [{\citenamefont {Zeng}\ \emph {et~al.}(2019)\citenamefont {Zeng},
  \citenamefont {Chen}, \citenamefont {Zhou},\ and\ \citenamefont
  {Wen}}]{qimqm}%
  \BibitemOpen
  \bibfield  {author} {\bibinfo {author} {\bibfnamefont {B.}~\bibnamefont
  {Zeng}}, \bibinfo {author} {\bibfnamefont {X.}~\bibnamefont {Chen}}, \bibinfo
  {author} {\bibfnamefont {D.-L.}\ \bibnamefont {Zhou}}, \ and\ \bibinfo
  {author} {\bibfnamefont {X.-G.}\ \bibnamefont {Wen}},\ }\href@noop {} {\emph
  {\bibinfo {title} {Quantum Information Meets Quantum Matter}}},\ \bibinfo
  {edition} {economy}\ ed.,\ Quantum Science and Technology\ (\bibinfo
  {publisher} {Springer-Verlag New York},\ \bibinfo {year} {2019})\ pp.\
  \bibinfo {pages} {XXII, 364}\BibitemShut {NoStop}%
\bibitem [{\citenamefont {Hastings}\ and\ \citenamefont
  {Wen}(2005)}]{hastings2005quasiadiabatic}%
  \BibitemOpen
  \bibfield  {author} {\bibinfo {author} {\bibfnamefont {M.~B.}\ \bibnamefont
  {Hastings}}\ and\ \bibinfo {author} {\bibfnamefont {X.~G.}\ \bibnamefont
  {Wen}},\ }{Quasiadiabatic continuation of quantum states: The stability of
  topological ground-state degeneracy and emergent gauge invariance},\ \href
  {\doibase 10.1103/PhysRevB.72.045141} {\bibfield  {journal} {\bibinfo
  {journal} {Phys. Rev. B}\ }\textbf {\bibinfo {volume} {72}},\ \bibinfo
  {pages} {45141}},\ \Eprint {http://arxiv.org/abs/cond-mat/0503554}
  {arXiv:cond-mat/0503554}  (\bibinfo {year} {2005})\BibitemShut {NoStop}%
\bibitem [{\citenamefont {Haah}\ \emph {et~al.}(2018)\citenamefont {Haah},
  \citenamefont {Fidkowski},\ and\ \citenamefont {Hastings}}]{Haah2018}%
  \BibitemOpen
  \bibfield  {author} {\bibinfo {author} {\bibfnamefont {J.}~\bibnamefont
  {Haah}}, \bibinfo {author} {\bibfnamefont {L.}~\bibnamefont {Fidkowski}}, \
  and\ \bibinfo {author} {\bibfnamefont {M.~B.}\ \bibnamefont {Hastings}},\
  }{Nontrivial Quantum Cellular Automata in Higher Dimensions},\ \href
  {http://arxiv.org/abs/1812.01625} {\ }\Eprint
  {http://arxiv.org/abs/1812.01625} {arXiv:1812.01625}  (\bibinfo {year}
  {2018})\BibitemShut {NoStop}%
\bibitem [{\citenamefont {Kitaev}(2003)}]{qdouble}%
  \BibitemOpen
  \bibfield  {author} {\bibinfo {author} {\bibfnamefont {A.~Y.}\ \bibnamefont
  {Kitaev}},\ }{Fault-tolerant quantum computation by anyons},\ \href {\doibase
  10.1016/S0003-4916(02)00018-0} {\bibfield  {journal} {\bibinfo  {journal}
  {Ann. Phys.}\ }\textbf {\bibinfo {volume} {303}},\ \bibinfo {pages} {2}},\
  \Eprint {http://arxiv.org/abs/quant-ph/9707021} {arXiv:quant-ph/9707021}
  (\bibinfo {year} {2003})\BibitemShut {NoStop}%
\bibitem [{\citenamefont {Levin}\ and\ \citenamefont {Wen}(2005)}]{Levin2005}%
  \BibitemOpen
  \bibfield  {author} {\bibinfo {author} {\bibfnamefont {M.~A.}\ \bibnamefont
  {Levin}}\ and\ \bibinfo {author} {\bibfnamefont {X.~G.}\ \bibnamefont
  {Wen}},\ }{String-net condensation: A physical mechanism for topological
  phases},\ \href {\doibase 10.1103/PhysRevB.71.045110} {\bibfield  {journal}
  {\bibinfo  {journal} {Phys. Rev. B}\ }\textbf {\bibinfo {volume} {71}},\
  \bibinfo {pages} {045110}},\ \Eprint {http://arxiv.org/abs/cond-mat/0404617}
  {arXiv:cond-mat/0404617}  (\bibinfo {year} {2005})\BibitemShut {NoStop}%
\bibitem [{\citenamefont {Koenig}\ \emph {et~al.}(2010)\citenamefont {Koenig},
  \citenamefont {Kuperberg},\ and\ \citenamefont
  {Reichardt}}]{koenig2010quantum}%
  \BibitemOpen
  \bibfield  {author} {\bibinfo {author} {\bibfnamefont {R.}~\bibnamefont
  {Koenig}}, \bibinfo {author} {\bibfnamefont {G.}~\bibnamefont {Kuperberg}}, \
  and\ \bibinfo {author} {\bibfnamefont {B.~W.}\ \bibnamefont {Reichardt}},\
  }{Quantum computation with Turaev-Viro codes},\ \href {\doibase
  10.1016/j.aop.2010.08.001} {\bibfield  {journal} {\bibinfo  {journal} {Ann.
  Phys.}\ }\textbf {\bibinfo {volume} {325}},\ \bibinfo {pages} {2707}},\
  \Eprint {http://arxiv.org/abs/1002.2816} {arXiv:1002.2816}  (\bibinfo {year}
  {2010})\BibitemShut {NoStop}%
\bibitem [{\citenamefont {Walker}\ and\ \citenamefont
  {Wang}(2012)}]{walker2012}%
  \BibitemOpen
  \bibfield  {author} {\bibinfo {author} {\bibfnamefont {K.}~\bibnamefont
  {Walker}}\ and\ \bibinfo {author} {\bibfnamefont {Z.}~\bibnamefont {Wang}},\
  }{(3+1)-TQFTs and topological insulators},\ \href {\doibase
  10.1007/s11467-011-0194-z} {\bibfield  {journal} {\bibinfo  {journal}
  {Frontiers of Physics}\ }\textbf {\bibinfo {volume} {7}},\ \bibinfo {pages}
  {150}},\ \Eprint {http://arxiv.org/abs/1104.2632} {arXiv:1104.2632}
  (\bibinfo {year} {2012})\BibitemShut {NoStop}%
\bibitem [{\citenamefont {Williamson}\ and\ \citenamefont
  {Wang}(2017)}]{williamson2016hamiltonian}%
  \BibitemOpen
  \bibfield  {author} {\bibinfo {author} {\bibfnamefont {D.~J.}\ \bibnamefont
  {Williamson}}\ and\ \bibinfo {author} {\bibfnamefont {Z.}~\bibnamefont
  {Wang}},\ }{Hamiltonian models for topological phases of matter in three
  spatial dimensions},\ \href {\doibase 10.1016/j.aop.2016.12.018} {\bibfield
  {journal} {\bibinfo  {journal} {Ann. Phys.}\ }\textbf {\bibinfo {volume}
  {377}},\ \bibinfo {pages} {311}},\ \Eprint {http://arxiv.org/abs/1606.07144}
  {arXiv:1606.07144}  (\bibinfo {year} {2017})\BibitemShut {NoStop}%
\bibitem [{\citenamefont {Calderbank}\ and\ \citenamefont
  {Shor}(1996)}]{PhysRevA.54.1098}%
  \BibitemOpen
  \bibfield  {author} {\bibinfo {author} {\bibfnamefont {A.~R.}\ \bibnamefont
  {Calderbank}}\ and\ \bibinfo {author} {\bibfnamefont {P.~W.}\ \bibnamefont
  {Shor}},\ }{Good quantum error-correcting codes exist},\ \href {\doibase
  10.1103/PhysRevA.54.1098} {\bibfield  {journal} {\bibinfo  {journal} {Phys.
  Rev. A}\ }\textbf {\bibinfo {volume} {54}},\ \bibinfo {pages} {1098}},\
  \Eprint {http://arxiv.org/abs/quant-ph/9512032} {arXiv:quant-ph/9512032}
  (\bibinfo {year} {1996})\BibitemShut {NoStop}%
\bibitem [{\citenamefont {Steane}(1996)}]{Steane2551}%
  \BibitemOpen
  \bibfield  {author} {\bibinfo {author} {\bibfnamefont {A.}~\bibnamefont
  {Steane}},\ }{Multiple-particle interference and quantum error correction},\
  \href {\doibase 10.1098/rspa.1996.0136} {\bibfield  {journal} {\bibinfo
  {journal} {Proc. Roy. Soc. Lond. A}\ }\textbf {\bibinfo {volume} {452}},\
  \bibinfo {pages} {2551}},\ \Eprint {http://arxiv.org/abs/quant-ph/9601029}
  {arXiv:quant-ph/9601029}  (\bibinfo {year} {1996})\BibitemShut {NoStop}%
\bibitem [{\citenamefont {Imai}(1977)}]{Imai1977TDC}%
  \BibitemOpen
  \bibfield  {author} {\bibinfo {author} {\bibfnamefont {H.}~\bibnamefont
  {Imai}},\ }A theory of two-dimensional cyclic codes,\ \href
  {https://doi.org/10.1016/S0019-9958(77)90232-7} {\bibfield  {journal}
  {\bibinfo  {journal} {Information and Control}\ }\textbf {\bibinfo {volume}
  {34}},\ \bibinfo {pages} {1}} (\bibinfo {year} {1977})\BibitemShut {NoStop}%
\bibitem [{\citenamefont {MacWilliams}\ and\ \citenamefont
  {Sloane}(1977)}]{cecc}%
  \BibitemOpen
  \bibfield  {author} {\bibinfo {author} {\bibfnamefont {F.~J.}\ \bibnamefont
  {MacWilliams}}\ and\ \bibinfo {author} {\bibfnamefont {N.~J.~A.}\
  \bibnamefont {Sloane}},\ }\href@noop {} {\emph {\bibinfo {title} {The theory
  of error-correcting codes}}}\ (\bibinfo  {publisher} {North-Holland
  Publishing Co., Amsterdam-New York-Oxford},\ \bibinfo {year} {1977})\ pp.\
  \bibinfo {pages} {i--ix and 1--762},\ \bibinfo {note} {north-Holland
  Mathematical Library, Vol. 16}\BibitemShut {NoStop}%
\bibitem [{\citenamefont {Güneri}\ and\ \citenamefont
  {Özbudak}(2008)}]{multivariable_cecc}%
  \BibitemOpen
  \bibfield  {author} {\bibinfo {author} {\bibfnamefont {C.}~\bibnamefont
  {Güneri}}\ and\ \bibinfo {author} {\bibfnamefont {F.}~\bibnamefont
  {Özbudak}},\ }Multidimensional cyclic codes and artin–schreier type
  hypersurfaces over finite fields,\ \href {\doibase
  https://doi.org/10.1016/j.ffa.2006.12.003} {\bibfield  {journal} {\bibinfo
  {journal} {Finite Fields and Their Applications}\ }\textbf {\bibinfo {volume}
  {14}},\ \bibinfo {pages} {44}} (\bibinfo {year} {2008})\BibitemShut {NoStop}%
\bibitem [{\citenamefont {Haah}(2013{\natexlab{a}})}]{haah2013commuting}%
  \BibitemOpen
  \bibfield  {author} {\bibinfo {author} {\bibfnamefont {J.}~\bibnamefont
  {Haah}},\ }{Commuting Pauli Hamiltonians as Maps between Free Modules},\
  \href {\doibase 10.1007/s00220-013-1810-2} {\bibfield  {journal} {\bibinfo
  {journal} {Commun. Math. Phys.}\ }\textbf {\bibinfo {volume} {324}},\
  \bibinfo {pages} {351}},\ \Eprint {http://arxiv.org/abs/1204.1063}
  {arXiv:1204.1063}  (\bibinfo {year} {2013}{\natexlab{a}})\BibitemShut
  {NoStop}%
\bibitem [{\citenamefont {Haah}(2013{\natexlab{b}})}]{haah2013}%
  \BibitemOpen
  \bibfield  {author} {\bibinfo {author} {\bibfnamefont {J.}~\bibnamefont
  {Haah}},\ }{Lattice quantum codes and exotic topological phases of matter},\
  \href {http://arxiv.org/abs/1305.6973} {\ }\Eprint
  {http://arxiv.org/abs/1305.6973} {arXiv:1305.6973}  (\bibinfo {year}
  {2013}{\natexlab{b}})\BibitemShut {NoStop}%
\bibitem [{\citenamefont {Devakul}\ \emph
  {et~al.}(2018{\natexlab{b}})\citenamefont {Devakul}, \citenamefont {You},
  \citenamefont {Burnell},\ and\ \citenamefont {Sondhi}}]{devakul2018fractal}%
  \BibitemOpen
  \bibfield  {author} {\bibinfo {author} {\bibfnamefont {T.}~\bibnamefont
  {Devakul}}, \bibinfo {author} {\bibfnamefont {Y.}~\bibnamefont {You}},
  \bibinfo {author} {\bibfnamefont {F.~J.}\ \bibnamefont {Burnell}}, \ and\
  \bibinfo {author} {\bibfnamefont {S.~L.}\ \bibnamefont {Sondhi}},\ }{Fractal
  Symmetric Phases of Matter},\ \href {http://arxiv.org/abs/1805.04097} {\
  }\Eprint {http://arxiv.org/abs/1805.04097} {arXiv:1805.04097}  (\bibinfo
  {year} {2018}{\natexlab{b}})\BibitemShut {NoStop}%
\bibitem [{\citenamefont {Devakul}\ and\ \citenamefont
  {Williamson}(2018)}]{devakul2018universal}%
  \BibitemOpen
  \bibfield  {author} {\bibinfo {author} {\bibfnamefont {T.}~\bibnamefont
  {Devakul}}\ and\ \bibinfo {author} {\bibfnamefont {D.~J.}\ \bibnamefont
  {Williamson}},\ }{Universal quantum computation using fractal
  symmetry-protected cluster phases},\ \href {\doibase
  10.1103/PhysRevA.98.022332} {\bibfield  {journal} {\bibinfo  {journal} {Phys.
  Rev. A}\ }\textbf {\bibinfo {volume} {98}},\ \bibinfo {pages} {022332}},\
  \Eprint {http://arxiv.org/abs/1806.04663} {arXiv:1806.04663}  (\bibinfo
  {year} {2018})\BibitemShut {NoStop}%
\bibitem [{\citenamefont {Stephen}\ \emph {et~al.}(2018)\citenamefont
  {Stephen}, \citenamefont {Nautrup}, \citenamefont {Bermejo-Vega},
  \citenamefont {Eisert},\ and\ \citenamefont
  {Raussendorf}}]{Stephen2018computationally}%
  \BibitemOpen
  \bibfield  {author} {\bibinfo {author} {\bibfnamefont {D.~T.}\ \bibnamefont
  {Stephen}}, \bibinfo {author} {\bibfnamefont {H.~P.}\ \bibnamefont
  {Nautrup}}, \bibinfo {author} {\bibfnamefont {J.}~\bibnamefont
  {Bermejo-Vega}}, \bibinfo {author} {\bibfnamefont {J.}~\bibnamefont
  {Eisert}}, \ and\ \bibinfo {author} {\bibfnamefont {R.}~\bibnamefont
  {Raussendorf}},\ }{Subsystem symmetries, quantum cellular automata, and
  computational phases of quantum matter},\ \href
  {http://arxiv.org/abs/1806.08780} {\ }\Eprint
  {http://arxiv.org/abs/1806.08780} {arXiv:1806.08780}  (\bibinfo {year}
  {2018})\BibitemShut {NoStop}%
\bibitem [{\citenamefont {Devakul}(2018{\natexlab{b}})}]{Devakul2018}%
  \BibitemOpen
  \bibfield  {author} {\bibinfo {author} {\bibfnamefont {T.}~\bibnamefont
  {Devakul}},\ }{Classifying local fractal subsystem symmetry protected
  topological phases},\ \href {https://arxiv.org/pdf/1812.02721.pdf} {\
  }\Eprint {http://arxiv.org/abs/1812.02721} {arXiv:1812.02721}  (\bibinfo
  {year} {2018}{\natexlab{b}})\BibitemShut {NoStop}%
\bibitem [{\citenamefont {Daniel}\ \emph {et~al.}(2019)\citenamefont {Daniel},
  \citenamefont {Alexander},\ and\ \citenamefont {Miyake}}]{Daniel2019}%
  \BibitemOpen
  \bibfield  {author} {\bibinfo {author} {\bibfnamefont {A.~K.}\ \bibnamefont
  {Daniel}}, \bibinfo {author} {\bibfnamefont {R.~N.}\ \bibnamefont
  {Alexander}}, \ and\ \bibinfo {author} {\bibfnamefont {A.}~\bibnamefont
  {Miyake}},\ }{Computational universality of symmetry-protected topologically
  ordered cluster phases on 2D Archimedean lattices},\ \href
  {http://arxiv.org/abs/1907.13279} {\ }\Eprint
  {http://arxiv.org/abs/1907.13279} {arXiv:1907.13279}  (\bibinfo {year}
  {2019})\BibitemShut {NoStop}%
\bibitem [{\citenamefont {Shirley}\ \emph
  {et~al.}(2019{\natexlab{c}})\citenamefont {Shirley}, \citenamefont {Slagle},\
  and\ \citenamefont {Chen}}]{ShirleyERG}%
  \BibitemOpen
  \bibfield  {author} {\bibinfo {author} {\bibfnamefont {W.}~\bibnamefont
  {Shirley}}, \bibinfo {author} {\bibfnamefont {K.}~\bibnamefont {Slagle}}, \
  and\ \bibinfo {author} {\bibfnamefont {X.}~\bibnamefont {Chen}},\
  }{Entanglement Renormalization of Fractonic Gauge Theories (in
  preparation)},\ \href@noop {} {\ } (\bibinfo {year}
  {2019}{\natexlab{c}})\BibitemShut {NoStop}%
\bibitem [{\citenamefont {Haah}(2016)}]{haah2016algebraic}%
  \BibitemOpen
  \bibfield  {author} {\bibinfo {author} {\bibfnamefont {J.}~\bibnamefont
  {Haah}},\ }{Algebraic Methods for Quantum Codes on Lattices},\ \href
  {\doibase 10.15446/recolma.v50n2.62214} {\bibfield  {journal} {\bibinfo
  {journal} {Revista Colombiana de Matem{\'{a}}ticas}\ }\textbf {\bibinfo
  {volume} {50}},\ \bibinfo {pages} {299}},\ \Eprint
  {http://arxiv.org/abs/1607.01387} {arXiv:1607.01387}  (\bibinfo {year}
  {2016})\BibitemShut {NoStop}%
\bibitem [{\citenamefont {Yoshida}(2011)}]{yoshida2011classification}%
  \BibitemOpen
  \bibfield  {author} {\bibinfo {author} {\bibfnamefont {B.}~\bibnamefont
  {Yoshida}},\ }{Classification of quantum phases and topology of logical
  operators in an exactly solved model of quantum codes},\ \href {\doibase
  10.1016/j.aop.2010.10.009} {\bibfield  {journal} {\bibinfo  {journal} {Ann.
  Phys.}\ }\textbf {\bibinfo {volume} {326}},\ \bibinfo {pages} {15}},\ \Eprint
  {http://arxiv.org/abs/1007.4601} {arXiv:1007.4601}  (\bibinfo {year}
  {2011})\BibitemShut {NoStop}%
\bibitem [{\citenamefont {Levin}\ and\ \citenamefont
  {Wen}(2003)}]{Levin_wen_fermion}%
  \BibitemOpen
  \bibfield  {author} {\bibinfo {author} {\bibfnamefont {M.}~\bibnamefont
  {Levin}}\ and\ \bibinfo {author} {\bibfnamefont {X.-G.}\ \bibnamefont
  {Wen}},\ }Fermions, strings, and gauge fields in lattice spin models,\ \href
  {\doibase 10.1103/PhysRevB.67.245316} {\bibfield  {journal} {\bibinfo
  {journal} {Phys. Rev. B}\ }\textbf {\bibinfo {volume} {67}},\ \bibinfo
  {pages} {245316}} (\bibinfo {year} {2003})\BibitemShut {NoStop}%
\bibitem [{\citenamefont {Slagle}\ \emph {et~al.}(2019)\citenamefont {Slagle},
  \citenamefont {Aasen},\ and\ \citenamefont
  {Williamson}}]{Slagle2018foliated}%
  \BibitemOpen
  \bibfield  {author} {\bibinfo {author} {\bibfnamefont {K.}~\bibnamefont
  {Slagle}}, \bibinfo {author} {\bibfnamefont {D.}~\bibnamefont {Aasen}}, \
  and\ \bibinfo {author} {\bibfnamefont {D.}~\bibnamefont {Williamson}},\
  }{Foliated field theory and string-membrane-net condensation picture of
  fracton order},\ \href
  {http://arxiv.org/abs/1812.01613{\%}0Ahttp://dx.doi.org/10.21468/SciPostPhys.6.4.043}
  {\bibfield  {journal} {\bibinfo  {journal} {SciPost Phys.}\ }\textbf
  {\bibinfo {volume} {6}}},\ \Eprint {http://arxiv.org/abs/1812.01613}
  {arXiv:1812.01613}  (\bibinfo {year} {2019})\BibitemShut {NoStop}%
\bibitem [{\citenamefont {Wang}\ \emph {et~al.}(2019)\citenamefont {Wang},
  \citenamefont {Shirley},\ and\ \citenamefont {Chen}}]{Wang2019}%
  \BibitemOpen
  \bibfield  {author} {\bibinfo {author} {\bibfnamefont {T.}~\bibnamefont
  {Wang}}, \bibinfo {author} {\bibfnamefont {W.}~\bibnamefont {Shirley}}, \
  and\ \bibinfo {author} {\bibfnamefont {X.}~\bibnamefont {Chen}},\ }{Foliated
  fracton order in the Majorana checkerboard model},\ \href
  {http://arxiv.org/abs/1904.01111} {\ }\Eprint
  {http://arxiv.org/abs/1904.01111} {arXiv:1904.01111}  (\bibinfo {year}
  {2019})\BibitemShut {NoStop}%
\bibitem [{\citenamefont {Dua}(2019)}]{dua_proof_str}%
  \BibitemOpen
  \bibfield  {author} {\bibinfo {author} {\bibfnamefont {A.}~\bibnamefont
  {Dua}},\ }\href@noop {} {\bibfield  {journal} {\bibinfo  {journal}
  {Unpublished}\ }} (\bibinfo {year} {2019})\BibitemShut {NoStop}%
\bibitem [{\citenamefont {Devakul}\ \emph
  {et~al.}(2018{\natexlab{c}})\citenamefont {Devakul}, \citenamefont
  {Williamson},\ and\ \citenamefont {You}}]{subsystemphaserel}%
  \BibitemOpen
  \bibfield  {author} {\bibinfo {author} {\bibfnamefont {T.}~\bibnamefont
  {Devakul}}, \bibinfo {author} {\bibfnamefont {D.~J.}\ \bibnamefont
  {Williamson}}, \ and\ \bibinfo {author} {\bibfnamefont {Y.}~\bibnamefont
  {You}},\ }{Classification of subsystem symmetry-protected topological
  phases},\ \href {\doibase 10.1103/PhysRevB.98.235121} {\bibfield  {journal}
  {\bibinfo  {journal} {Phys. Rev. B}\ }\textbf {\bibinfo {volume} {98}},\
  \bibinfo {pages} {235121}},\ \Eprint {http://arxiv.org/abs/1808.05300}
  {arXiv:1808.05300}  (\bibinfo {year} {2018}{\natexlab{c}})\BibitemShut
  {NoStop}%
\bibitem [{\citenamefont {Shirley}\ \emph
  {et~al.}(2019{\natexlab{d}})\citenamefont {Shirley}, \citenamefont {Slagle},\
  and\ \citenamefont {Chen}}]{Shirley2019}%
  \BibitemOpen
  \bibfield  {author} {\bibinfo {author} {\bibfnamefont {W.}~\bibnamefont
  {Shirley}}, \bibinfo {author} {\bibfnamefont {K.}~\bibnamefont {Slagle}}, \
  and\ \bibinfo {author} {\bibfnamefont {X.}~\bibnamefont {Chen}},\ }{Twisted
  foliated fracton phases},\ \href {http://arxiv.org/abs/1907.09048} {\
  }\Eprint {http://arxiv.org/abs/1907.09048} {arXiv:1907.09048}  (\bibinfo
  {year} {2019}{\natexlab{d}})\BibitemShut {NoStop}%
\bibitem [{\citenamefont {Pretko}(2017{\natexlab{a}})}]{PhysRevB.95.115139}%
  \BibitemOpen
  \bibfield  {author} {\bibinfo {author} {\bibfnamefont {M.}~\bibnamefont
  {Pretko}},\ }{Subdimensional particle structure of higher rank U(1) spin
  liquids},\ \href {\doibase 10.1103/PhysRevB.95.115139} {\bibfield  {journal}
  {\bibinfo  {journal} {Phys. Rev. B}\ }\textbf {\bibinfo {volume} {95}},\
  \bibinfo {pages} {115139}},\ \Eprint {http://arxiv.org/abs/1604.05329}
  {arXiv:1604.05329}  (\bibinfo {year} {2017}{\natexlab{a}})\BibitemShut
  {NoStop}%
\bibitem [{\citenamefont {Pretko}(2017{\natexlab{b}})}]{PhysRevB.96.035119}%
  \BibitemOpen
  \bibfield  {author} {\bibinfo {author} {\bibfnamefont {M.}~\bibnamefont
  {Pretko}},\ }{Generalized electromagnetism of subdimensional particles: A
  spin liquid story},\ \href {\doibase 10.1103/PhysRevB.96.035119} {\bibfield
  {journal} {\bibinfo  {journal} {Phys. Rev. B}\ }\textbf {\bibinfo {volume}
  {96}},\ \bibinfo {pages} {35119}},\ \Eprint {http://arxiv.org/abs/1606.08857}
  {arXiv:1606.08857}  (\bibinfo {year} {2017}{\natexlab{b}})\BibitemShut
  {NoStop}%
\bibitem [{\citenamefont {Ma}\ \emph {et~al.}(2018)\citenamefont {Ma},
  \citenamefont {Hermele},\ and\ \citenamefont {Chen}}]{ma2018fracton}%
  \BibitemOpen
  \bibfield  {author} {\bibinfo {author} {\bibfnamefont {H.}~\bibnamefont
  {Ma}}, \bibinfo {author} {\bibfnamefont {M.}~\bibnamefont {Hermele}}, \ and\
  \bibinfo {author} {\bibfnamefont {X.}~\bibnamefont {Chen}},\ }{Fracton
  topological order from the Higgs and partial-confinement mechanisms of
  rank-two gauge theory},\ \href
  {http://arxiv.org/abs/1802.10108{\%}0Ahttp://dx.doi.org/10.1103/PhysRevB.98.035111}
  {\bibfield  {journal} {\bibinfo  {journal} {Phys. Rev. B}\ }\textbf {\bibinfo
  {volume} {98}},\ \bibinfo {pages} {035111}},\ \Eprint
  {http://arxiv.org/abs/1802.10108} {arXiv:1802.10108}  (\bibinfo {year}
  {2018})\BibitemShut {NoStop}%
\bibitem [{\citenamefont {Bulmash}\ and\ \citenamefont
  {Barkeshli}(2018{\natexlab{a}})}]{PhysRevB.97.235112}%
  \BibitemOpen
  \bibfield  {author} {\bibinfo {author} {\bibfnamefont {D.}~\bibnamefont
  {Bulmash}}\ and\ \bibinfo {author} {\bibfnamefont {M.}~\bibnamefont
  {Barkeshli}},\ }{Higgs mechanism in higher-rank symmetric U(1) gauge
  theories},\ \href {\doibase 10.1103/PhysRevB.97.235112} {\bibfield  {journal}
  {\bibinfo  {journal} {Phys. Rev. B}\ }\textbf {\bibinfo {volume} {97}},\
  \bibinfo {pages} {235112}},\ \Eprint {http://arxiv.org/abs/1802.10099}
  {arXiv:1802.10099}  (\bibinfo {year} {2018}{\natexlab{a}})\BibitemShut
  {NoStop}%
\bibitem [{\citenamefont {Bulmash}\ and\ \citenamefont
  {Barkeshli}(2018{\natexlab{b}})}]{bulmash2018generalized}%
  \BibitemOpen
  \bibfield  {author} {\bibinfo {author} {\bibfnamefont {D.}~\bibnamefont
  {Bulmash}}\ and\ \bibinfo {author} {\bibfnamefont {M.}~\bibnamefont
  {Barkeshli}},\ }{Generalized U(1) Gauge Field Theories and Fractal
  Dynamics},\ \href {http://arxiv.org/abs/1806.01855} {\ }\Eprint
  {http://arxiv.org/abs/1806.01855} {arXiv:1806.01855}  (\bibinfo {year}
  {2018}{\natexlab{b}})\BibitemShut {NoStop}%
\bibitem [{\citenamefont {Williamson}\ \emph {et~al.}(2018)\citenamefont
  {Williamson}, \citenamefont {Bi},\ and\ \citenamefont
  {Cheng}}]{Williamson2018Fractonic}%
  \BibitemOpen
  \bibfield  {author} {\bibinfo {author} {\bibfnamefont {D.~J.}\ \bibnamefont
  {Williamson}}, \bibinfo {author} {\bibfnamefont {Z.}~\bibnamefont {Bi}}, \
  and\ \bibinfo {author} {\bibfnamefont {M.}~\bibnamefont {Cheng}},\
  }{Fractonic Matter in Symmetry-Enriched U(1) Gauge Theory},\ \href
  {http://arxiv.org/abs/1809.10275} {\ }\Eprint
  {http://arxiv.org/abs/1809.10275} {arXiv:1809.10275}  (\bibinfo {year}
  {2018})\BibitemShut {NoStop}%
\bibitem [{\citenamefont {Song}\ \emph
  {et~al.}(2019{\natexlab{b}})\citenamefont {Song}, \citenamefont {Prem},
  \citenamefont {Huang},\ and\ \citenamefont
  {Martin-Delgado}}]{song2018twisted}%
  \BibitemOpen
  \bibfield  {author} {\bibinfo {author} {\bibfnamefont {H.}~\bibnamefont
  {Song}}, \bibinfo {author} {\bibfnamefont {A.}~\bibnamefont {Prem}}, \bibinfo
  {author} {\bibfnamefont {S.-J.}\ \bibnamefont {Huang}}, \ and\ \bibinfo
  {author} {\bibfnamefont {M.~A.}\ \bibnamefont {Martin-Delgado}},\ }Twisted
  fracton models in three dimensions,\ \href {\doibase
  10.1103/PhysRevB.99.155118} {\bibfield  {journal} {\bibinfo  {journal} {Phys.
  Rev. B}\ }\textbf {\bibinfo {volume} {99}},\ \bibinfo {pages} {155118}}
  (\bibinfo {year} {2019}{\natexlab{b}})\BibitemShut {NoStop}%
\bibitem [{\citenamefont {Atiyah}\ and\ \citenamefont
  {Macdonald}(2016)}]{AM16}%
  \BibitemOpen
  \bibfield  {author} {\bibinfo {author} {\bibfnamefont {M.~F.}\ \bibnamefont
  {Atiyah}}\ and\ \bibinfo {author} {\bibfnamefont {I.~G.}\ \bibnamefont
  {Macdonald}},\ }\href@noop {} {\emph {\bibinfo {title} {Introduction to
  commutative algebra}}},\ \bibinfo {edition} {economy}\ ed.,\ Addison-Wesley
  Series in Mathematics\ (\bibinfo  {publisher} {Westview Press, Boulder, CO},\
  \bibinfo {year} {2016})\ pp.\ \bibinfo {pages} {ix+128}\BibitemShut {NoStop}%
\bibitem [{\citenamefont {Adams}\ and\ \citenamefont
  {Loustaunau}(1994)}]{AL94}%
  \BibitemOpen
  \bibfield  {author} {\bibinfo {author} {\bibfnamefont {W.~W.}\ \bibnamefont
  {Adams}}\ and\ \bibinfo {author} {\bibfnamefont {P.}~\bibnamefont
  {Loustaunau}},\ }\href {\doibase 10.1090/gsm/003} {\emph {\bibinfo {title}
  {An introduction to {G}r\"{o}bner bases}}},\ \bibinfo {series} {Graduate
  Studies in Mathematics}, Vol.~\bibinfo {volume} {3}\ (\bibinfo  {publisher}
  {American Mathematical Society, Providence, RI},\ \bibinfo {year}
  {1994})\BibitemShut {NoStop}%
\end{thebibliography}%
